\documentclass[DIV12]{scrartcl}

\usepackage[utf8]{inputenc}

\usepackage{microtype}
\usepackage{graphicx}
\usepackage{booktabs} 
\usepackage{xcolor} 
\usepackage{amssymb}
\usepackage{stmaryrd}
\usepackage{bm}
\usepackage[pdfborder={0 0 0}]{hyperref}
\usepackage{paralist}
\usepackage{amsmath}
\usepackage{colonequals}
\usepackage[linesnumbered,ruled,vlined]{algorithm2e}
\SetKwRepeat{Do}{do}{While}

\hypersetup{
    colorlinks,
    linkcolor={red!50!black},
    citecolor={blue!50!black},
    urlcolor={blue!30!black}
}

\makeatother

\def\subparagraph*#1{\paragraph{#1}}

\newcommand{\parabm}[1]{\bm{#1}}

\newcommand{\myparagraph}[1]{\subparagraph*{#1.}}

\newcommand{\dom}{\mathrm{dom}}
\newcommand{\defeq}{\colonequals}

\newcommand{\sigmai}{\sigma_{\mathrm{int}}}
\newcommand{\sigmae}{\sigma}

\newcommand{\card}[1]{\left|#1\right|}

\newcommand{\calA}{\mathcal{A}}
\newcommand{\calB}{\mathcal{B}}
\newcommand{\calG}{\mathcal{G}}
\newcommand{\calQ}{\mathcal{Q}}
\newcommand{\calD}{\mathcal{D}}
\newcommand{\calI}{\mathcal{I}}
\newcommand{\calE}{\mathcal{E}}
\newcommand{\calN}{\mathcal{N}}
\newcommand{\calS}{\mathcal{S}}
\newcommand{\calT}{\mathcal{T}}

\newcommand{\II}{\mathrm{I}}
\newcommand{\I}{\mathrm{I}}
\newcommand{\QQ}{\mathrm{Q}}
\newcommand{\PP}{\mathrm{P}}

\newcommand{\ki}{k_{\II}}
\newcommand{\kq}{k_{\QQ}}
\newcommand{\kp}{k_{\PP}}

\newcommand{\NN}{\mathbb{N}}

\newcommand{\decode}[1]{\mathrm{dec}(#1)}

\newcommand{\arity}[1]{\mathrm{arity}(#1)}

\renewcommand\i{\mathrm{i}}

\renewcommand\r{\mathrm{r}}
\newcommand\AND{\mathrm{AND}}
\newcommand\OR{\mathrm{OR}}
\newcommand\NOT{\mathrm{NOT}}

\newcommand\true{\top}
\newcommand\false{\bot}

\newcommand\neigh{\mathrm{Nbh}}
\newcommand\la{\langle}
\newcommand\ra{\rangle}
\newcommand\strat{\zeta}
\newcommand\dec{\mathrm{dec}}
\newcommand\enc{\mathrm{enc}}

\newcommand\vars{\mathrm{vars}}

\newcommand{\inp}{\mathsf{inp}}
\newcommand{\gi}{\mathsf{in}}

\newcommand{\guard}{\mathrm{guard}}

\newcommand{\goal}{\text{Goal}()}

\newcommand{\first}{\mathsf{Ch1}}
\newcommand{\second}{\mathsf{Ch2}}

\newcommand{\binary}{\mathsf{Bin}}
\newcommand{\schildself}{\mathcal{S}^{\binary}_{\first,\second,\child,\childself}}

\newcommand{\child}{\mathsf{Child}}
\newcommand{\childself}{\mathsf{Child^?}}

\newcommand{\f}{\mathrm{f}}

\newcommand{\mnull}{\mathsf{null}}

\newcommand\restr[2]{{%
  \kern-\nulldelimiterspace %
  #1 %
  _{|#2} %
  }}

\newcommand{\cfggnd}{\mbox{CFG$^\text{GN}$-Datalog}\xspace}

\newcommand{\Pdir}{P_{\text{directions}}}

\newcommand{\Adom}{\text{Adom}}

\title{Evaluating Datalog via Tree Automata and Cycluits}

\author{
\begin{tabular}[t]{c}
Antoine Amarilli \\
{\normalfont LTCI, Télécom ParisTech, Université Paris-Saclay} \\
{\normalfont antoine.amarilli@telecom-paristech.fr} \\[0.5em]
Pierre Bourhis \\
{\normalfont CRIStAL, CNRS, Université de Lille} \\
{\normalfont pierre.bourhis@lifl.fr} \\[0.5em]
Mikaël Monet \\
{\normalfont LTCI, Télécom ParisTech, Université Paris-Saclay} \\
{\normalfont \& Inria Paris} \\
{\normalfont mikael.monet@telecom-paristech.fr} \\[0.5em]
Pierre Senellart \\
{\normalfont
LTCI, CNRS, Télécom ParisTech, Université Paris-Saclay} \\
{\normalfont \& Inria Paris} \\
{\normalfont \& DI ENS, ENS, CNRS, PSL University} \\
{\normalfont pierre.senellart@ens.fr} \\[0.5em]
\end{tabular}
}

\date{}

\renewcommand{\phi}{\varphi}
\renewcommand{\epsilon}{\varepsilon}
\renewcommand{\leq}{\leqslant}
\renewcommand{\geq}{\geqslant}

\usepackage[bibliography=common]{apxproof}
\renewcommand{\appendixprelim}{}
\newtheoremrep{theorem}{Theorem}
\newtheorem{proposition}[theorem]{Proposition}
\newtheorem{lemma}[theorem]{Lemma}
\newtheorem{property}[theorem]{Property}

\theoremstyle{plain}
\newtheorem{example}[theorem]{Example}
\newtheorem{definition}[theorem]{Definition}

\begin{document}

\maketitle

\begin{abstract}
  We investigate parameterizations of both database instances and queries
that make query evaluation fixed-parameter tractable in combined
complexity.
We show that clique-frontier-guarded Datalog with stratified negation (CFG-Datalog) enjoys
bilinear-time evaluation on structures of bounded treewidth for programs of bounded
rule size. Such programs capture in particular conjunctive queries with
simplicial decompositions of bounded width, guarded negation fragment
queries of bounded CQ-rank, or two-way regular path queries. Our
result is shown by translating to alternating two-way
automata, whose semantics is defined via cyclic provenance circuits
(cycluits) that can be tractably evaluated. 

\end{abstract}

\pagebreak

\section{Introduction}
Arguably the most fundamental task performed by database systems
is \emph{query evaluation}, namely, computing the results
of a query over a
database instance. Unfortunately, this task is well-known to be intractable in
\emph{combined complexity}~\cite{vardi1982complexity} even for simple
query languages such as conjunctive
queries~\cite{abiteboul1995foundations}.

To address this issue, two main directions have been investigated. The first
is to restrict the class of \emph{queries} to ensure tractability, for
instance, to $\alpha$-acyclic conjunctive queries \cite{yannakakis1981algorithms}, this
being motivated by the idea that many real-world queries are simple and usually
small. 
The second approach restricts the
structure of database instances, e.g., requiring them to have bounded
\emph{treewidth}~\cite{robertson1986graph} (we call them \emph{treelike}).
This has been notably studied by
Courcelle~\cite{courcelle1990monadic}, to show the tractability of
monadic-second order logic
on treelike instances, but in \emph{data complexity} (i.e., for fixed
queries); the combined complexity is
generally nonelementary~\cite{meyer1975weak}.

This leaves open the main question studied in this paper: \emph{Which queries
can be efficiently evaluated, in combined complexity, on treelike databases?}
This question has been addressed 
by Gottlob, Pichler, and Fei~\cite{gottlob2010monadic} by
introducing \emph{quasi-guarded Datalog};
however, an unusual feature of this language is that programs must 
explicitly refer to the tree decomposition
of the instance. Instead, we try
to follow Courcelle's approach and investigate which queries can be
efficiently \emph{translated to automata}. Specifically, rather than restricting
to a fixed class of ``efficient'' queries, we study \emph{parameterized}
query classes,
i.e., we define an efficient class of queries for each value of the parameter.
We further make the standard assumption that the signature is fixed; in
particular, its arity is constant.
This allows us to
aim for low combined complexity for query evaluation, namely, fixed-parameter tractability with
linear-time complexity in the product of the input query and instance, which we call
\emph{FPT-bilinear} complexity.

Surprisingly, we are not aware of further existing work on tractable 
combined query evaluation for parameterized instances and queries, except from
an unexpected angle: the translation of restricted query fragments to tree
automata on treelike instances was used in the context of \emph{guarded logics}
and other fragments,
to decide
\emph{satisfiability}~\cite{benedikt2016step} and
\emph{containment}~\cite{barcelo2014does}.
To do this, one usually establishes a \emph{treelike
model property} to restrict the search to models of low treewidth (but
dependent on the formula), and then translates the formula to an
automaton, so that the problems reduce to emptiness testing: expressive
automata formalisms, such as \emph{alternating two-way automata}, are typically
used. Exploiting this connection, we show how
query evaluation on treelike instances can benefit from these ideas: for
instance, as we show, some queries can only be translated efficiently to such
concise automata, and not to the more common bottom-up tree automata.

\paragraph{Contributions.}

From there, the first main contribution of
this paper is to consider the language of \emph{clique-frontier-guarded
Datalog} (CFG-Datalog), and show an efficient FPT-linear translation procedure
for this language, parameterized by the body size of rules: this implies
FPT-bilinear combined complexity on treelike instances. While it is a
Datalog fragment, CFG-Datalog shares some similarities with guarded logics;
yet, its design incorporates several features
(fixpoints, clique-guards, negation, guarding positive subformulas) that
are not usually found together in guarded fragments, but are important for query
evaluation. We show how the tractability of this language captures the tractability of such query
classes as two-way regular path queries~\cite{barcelo2013querying} and
$\alpha$-acyclic conjunctive queries.
We further show that, in contrast with guarded negation logics, satisfiability of
CFG-Datalog is undecidable. 

Already for
conjunctive queries, we show that the treewidth of queries is not
the right parameter to ensure efficient translatability. In fact, the second contribution of our work is a lower bound: we show that bounded-treewidth queries cannot be efficiently translated to automata at all, so we cannot hope
to show
combined tractability for them via automata methods. By
contrast, CFG-Datalog implies the combined tractability of bounded-treewidth queries with
an additional requirement (interfaces between bags must be clique-guarded),
which is the notion of \emph{simplicial decompositions} previously studied
by Tarjan~\cite{tarjan1985decomposition}. To our knowledge, we are the first
to introduce this query class and to show its 
tractability on treelike instances. CFG-Datalog can be
understood as an extension of this fragment to disjunction, clique-guardedness, stratified negation, and inflationary fixpoints, that preserves tractability.

To derive our main FPT-bilinear combined complexity result, we define an
operational semantics for our tree automata 
by introducing a notion of
\emph{cyclic provenance circuits},
that we call \emph{cycluits}. These cycluits, the third contribution of our paper,
are well-suited as a provenance representation
for alternating two-way automata encoding CFG-Datalog programs, as they naturally
deal with both recursion and two-way traversal of a treelike instance,
which is less straightforward with provenance
formulas~\cite{green2007provenance} or
circuits~\cite{deutch2014circuits}. 
While we believe that this
natural generalization of Boolean circuits may be of independent interest, it
does not seem to have been studied in detail, except in the context of
integrated circuit design \cite{malik1993analysis,riedel2012cyclic}, where the semantics often features
feedback loops that involve negation; we prohibit these by focusing on \emph{stratified}
circuits, which we show can be evaluated in linear time.
We
show that the provenance of alternating two-way automata can be
represented as a stratified 
cycluit in FPT-bilinear time, generalizing results on bottom-up automata
and circuits from~\cite{amarilli2015provenance}.

The current article is a significant extension of the conference
version~\cite{amarilli2017combined,amarilli2017combinedb}, which in
particular includes all proofs. We improved the definition of our
language to a more natural and more expressive one, allowing us to step
away from the world of guarded negation logics and thus answering a
question that we left open in the conclusion of~\cite{amarilli2017combined}.
We show that, in contrast with guarded negation logics and the
ICG-Datalog language of~\cite{amarilli2017combined}, satisfiability of
CFG-Datalog is undecidable. To make space for the new material, this
paper does not include any of the applications to probabilistic query
evaluation that can be found
in~\cite{amarilli2017combined,amarilli2017combinedb} (see also
\cite{amarilli2017conjunctive} for a more in-depth study of the combined
complexity of probabilistic query evaluation).

\paragraph{Outline.} We give preliminaries in
Section~\ref{sec:prelim}, and then position our approach relative to
existing work in Section~\ref{sec:existing}. We then present our tractable
fragment, first for bounded-simplicial-width conjunctive queries in
Section~\ref{sec:cq}, then for CFG-Datalog in
Section~\ref{sec:CFG}. We then define the automata variants we use and translate CFG-Datalog to
them in Section~\ref{sec:compilation}, before introducing cycluits and showing
our provenance computation result in Section~\ref{sec:provenance}. 
We last present the proof of our translation result in Section~\ref{sec:proof}.

\section{Preliminaries}
\label{sec:prelim}
\subparagraph*{Relational instances.}

A \emph{relational signature} $\sigma$ is a finite set of relation
names written $R$, $S$,
$T$, \dots, each with its associated \emph{arity} written
$\arity{R} \in \NN$.
Throughout this work, \emph{we always assume the signature $\sigma$ to be
fixed} (with a single exception, in Proposition~\ref{prop:undecidable}):
hence, its \emph{arity} $\arity{\sigma}$ (the maximal arity of relations
in~$\sigma$) is assumed to be constant, and we further assume it is $>0$.
A \emph{($\sigma$-)instance} $I$
is a finite set of \emph{facts} on~$\sigma$,
i.e., $R(a_1, \ldots, a_{\arity{R}})$ with $R \in
\sigma$, and where $a_i$ is what we call an \emph{element}.
The \emph{active domain} $\dom(I)$ consists
of the
elements occurring in~$I$, and the size of $I$, denoted $|I|$, is the number of
tuples that $I$ contains.

\begin{example}
	\label{expl:instance}
	Table~\ref{tab:instance} shows an example of relational instance $I$ on signature $\sigma = \{R,S,T\}$ with $\arity{R} = \arity{S} = 2$ and $\arity{T} = 3$.
	The active domain of $I$ is $\dom(I) = \{1,2,3,4,5,6,7,8,9,10,11\}$ and its size is $|I| = 11$.
\begin{table}
\centering
\caption{Example relational instance}

\begin{tabular}{ccc}
	\begin{tabular}[t]{cc}
\toprule
		\multicolumn{2}{c}{$\mathbf{R}$} \\ 
\midrule
$3$ & $7$  \\ 
$3$ & $4$  \\
$5$ & $4$\\
$2$ & $5$\\
$9$ & $10$\\
$7$ & $8$\\
\bottomrule
\end{tabular}
&
\begin{tabular}[t]{cc}
\toprule
	\multicolumn{2}{c}{$\mathbf{S}$} \\ 
\midrule
$3$ & $7$  \\ 
$7$ & $9$\\
$11$ & $9$\\
$2$ & $6$\\
\bottomrule
\end{tabular}
&
\begin{tabular}[t]{ccc}
\toprule
	\multicolumn{3}{c}{$\mathbf{T}$} \\ 
\midrule
$1$ & $2$ & $3$  \\ 
\bottomrule
\end{tabular}
\end{tabular}

\label{tab:instance}
\end{table}
\end{example}

A \emph{subinstance} of $I$ is a $\sigma$-instance that is included in $I$ (as a set of tuples).
An \emph{isomorphism} between two $\sigma$-instances $I$ and $I'$ is a bijective function 
$f: \dom(I) \to \dom(I')$ such that for every relation name $R$,
for each tuple $(a_1,\ldots,a_{\arity{R}})$ in $\dom(I)^{\arity{R}}$,  
we have $R(a_1,\ldots,a_{\arity{R}}) \in I$
if and only if $R(f(a'_1),\ldots,f(a'_{\arity{R}})) \in I'$. When there exists such an isomorphism, we
say that $I$ and $I'$ are \emph{isomorphic}:
intuitively, isomorphic instances have exactly the same structure and differ only by the name of the elements
in their active domains.

\subparagraph*{Query evaluation and fixed-parameter tractability.}

We study query evaluation for several \emph{query
languages} that are subsets of first-order (FO) logic
(e.g., conjunctive queries) or of second-order (SO) logic (e.g., Datalog), without built-in relations.
Unless otherwise stated, we
only consider queries that are \emph{constant-free}, and \emph{Boolean}, so that an instance $I$ either
\emph{satisfies} a query $Q$ ($I \models Q$), or \emph{violates} it
($I \not\models Q$), with the standard semantics~\cite{abiteboul1995foundations}.
We recall that a constant-free Boolean query $Q$ cannot differentiate between isomorphic instances, i.e., for any two isomorphic relational instances $I$ and $I'$,
we have $I \models Q$ if and only if $I' \models Q$.

We study the \emph{query evaluation} (or \emph{model checking})
problem for a query class $\mathcal{Q}$ and instance
class $\mathcal{I}$:
given an instance $I \in
\mathcal{I}$ and query $Q \in \mathcal{Q}$, check if $I \models Q$.
Its \emph{combined complexity} 
for 
$\calI$ and $\calQ$ is 
a function of~$I$ and~$Q$, whereas
\emph{data complexity} 
assumes $Q$ to be fixed.
We 
also study cases where $\calI$ and
$\calQ$ are \emph{parameterized}: given infinite sequences
$\calI_1, \calI_2, \ldots$ and $\calQ_1, \calQ_2, \ldots$, the \emph{query evaluation
problem parameterized by $k_{\mathrm{I}}$, $k_{\mathrm{Q}}$} 
applies to~$\calI_{k_{\II}}$ and $\calQ_{k_{\mathrm{Q}}}$.
The parameterized problem is \emph{fixed-parameter tractable} (FPT), for $(\calI_n)$ and $(\calQ_n)$, if
there is a constant $c \in \mathbb{N}$ and computable function $f$ such that the problem
can be solved 
with combined complexity
$O\left(f(\ki, \kq) \cdot (\card{I} + \card{Q})^c\right)$.
When the complexity is of the form $O\left(f(\ki, \kq) \cdot (\card{I} \cdot \card{Q})\right)$, we call it \emph{FPT-bilinear} (in $\card{I} \cdot \card{Q}$).
When there is only one input (for example when we want to check that an instance $I$ has treewidth $\leq \ki$) and the complexity is $O\left(f(\ki) \cdot \card{I}\right)$, we call it \emph{FPT-linear}.
Observe that calling the problem FPT is more
informative than saying that it is in PTIME for fixed $\ki$ and $\kq$, as we are
further imposing that the polynomial degree $c$ does not depend on $\ki$ and $\kq$:
this follows the usual distinction in parameterized complexity between FPT and
classes such as XP~\cite{flum2006parameterized}.

\subparagraph*{Query languages.}
We first study fragments of FO, in particular, \emph{conjunctive queries} (CQ), i.e., 
existentially quantified conjunctions of atoms.
The \emph{canonical model} of a CQ $Q$ is the instance built from $Q$ by
seeing variables as elements and atoms as facts.
The \emph{primal graph} of~$Q$ has its variables as vertices, and
connects all variable pairs that co-occur in some atom.

Second, we study \emph{Datalog with stratified negation}.
We summarize the definitions here, see
\cite{abiteboul1995foundations} for details. 
A \emph{Datalog program} $P$
(without negation) over $\sigma$ (called the \emph{extensional signature})
consists of an \emph{intensional signature} $\sigmai$ disjoint from $\sigma$
(with the arity of $\sigmai$ being possibly greater than that of~$\sigma$),
a 0-ary \emph{goal predicate} $\text{Goal}$ in~$\sigmai$, and a set of
\emph{rules}.
Each rule is
of the form $R(\mathbf{x}) \leftarrow \psi(\mathbf{x}, \mathbf{y})$, where
the \emph{head} $R(\mathbf{x})$ is an atom with $R \in
\sigmai$, and the \emph{body} $\psi$ is a CQ over the signature $\sigmai \sqcup \sigmae$ 
(with~$\sqcup$ denoting disjoint union), where
we require that every variable of~$\mathbf{x}$ occurs in~$\psi$.
The \emph{semantics} $P(I)$ of $P$ over
an input $\sigma$-instance $I$
is the $(\sigma \sqcup \sigmai)$-instance defined by as the least fixpoint 
of the \emph{immediate consequence operator} $\Xi^P$. Formally,
start with $P(I) \defeq I$, and repeatedly apply the operator $\Xi^P$ which does
the following: simultaneously consider 
each rule $R(\mathbf{x})
\leftarrow \psi(\mathbf{x}, \mathbf{y})$ and every tuple $\mathbf{a}$ of
$\dom(I)$ for which
$P(I) \models \exists \mathbf{y} \,\psi(\mathbf{a}, \mathbf{y})$, 
\emph{derive} the fact $R(\mathbf{a})$, and add all derived facts 
to~$P(I)$ where they can be used in subsequent iterations to derive more facts.
We say that $I \models P$ iff $\text{Goal}()$ is in~$P(I)$.
The \emph{arity} of~$P$ is $\max(\arity{\sigma}, \arity{\sigmai})$,
and $P$~is \emph{monadic} if $\sigmai$ has arity~$1$.

\emph{Datalog with stratified
negation}~\cite{abiteboul1995foundations} allows
negated \emph{intensional} atoms
 in bodies, but
requires~$P$ to have a \emph{stratification}, i.e.,
an ordered partition $P_1
\sqcup \dots \sqcup P_n$ of the rules where:
\begin{compactenum}[(i)]
\item Each $R \in \sigmai$ has a \emph{stratum} $\strat(R) \in \{1, \ldots, n\}$ such that all rules
  with $R$ in the head are in~$P_{\strat(R)}$;
 \item For any $1 \leq i \leq n$ and $\sigmai$-atom $R(\mathbf{z})$ in a body of a rule of~$P_i$,
   we have $\strat(R) \leq i$;
 \item For any $1 \leq i \leq n$ and negated $\sigmai$-atom $R(\mathbf{z})$ in a
   body of~$P_i$, we have 
  $\strat(R) < i$. 
\end{compactenum}
The stratification ensures that we can define the semantics of a stratified
Datalog program by computing its interpretation for strata $P_1, \ldots, P_n$ in order:
atoms in bodies always depend on a lower stratum, and negated atoms depend on
strictly lower strata, whose
interpretation was already fixed. We point out that, although a Datalog program with stratified negation can have many stratifications, all stratifications give rise to the same semantics~\cite[Theorem 15.2.10]{abiteboul1995foundations}. 
Hence, the semantics of $P$, as well as $I \models P$, are well-defined.

\begin{example}
	\label{expl:datalog}
  The following stratified Datalog program, with $\sigma = \{R\}$ and $\sigmai =
  \{T, \mathrm{Goal}\}$, and strata $P_1$,
  $P_2$, tests if there are two elements that are not connected by a directed
  $R$-path:\\
  \null\hfill\(
    P_1: T(x,y) \leftarrow R(x,y), \quad
      T(x,y) \leftarrow R(x,z) \land T(z,y) \qquad\qquad P_2: 
    \mathrm{Goal}() \leftarrow \lnot T(x,y)
    \)\hfill\null
\end{example}

\paragraph*{Treewidth.}

The treewidth measure~\cite{robertson1984graph} quantifies how far a graph is to
being a tree: we will use it
to restrict instances and conjunctive queries.
The \emph{treewidth} of a CQ is that of its canonical model, and
the \emph{treewidth} of an instance $I$ is the smallest $k$ such that $I$ has a 
\emph{tree decomposition} of \emph{width} $k$,
i.e., a finite, rooted, unranked tree $T$, whose nodes~$b$ (called
\emph{bags}) are labeled by a subset $\dom(b)$ of $\dom(I)$ with $\card{\dom(b)}
\leq k+1$, and which satisfies:
\begin{compactenum}[(i)]
\item for every
fact $R(\mathbf{a}) \in I$, there is a bag $b \in T$ with $\mathbf{a} \subseteq
\dom(b)$; 
\item for all $a \in \dom(I)$, the set of bags $\{b \in T \mid a \in
  \dom(b)\}$ is a connected subtree of~$T$.
\end{compactenum}

\begin{example}
	\label{expl:treedec}
	Figure~\ref{fig:treedec} shows a tree decomposition of the instance $I$ from Example~\ref{expl:instance}.
	The width of this tree decomposition is $2$. Moreover, the width of any tree decomposition of $I$ is at least $2$, since 
        there must be a bag containing all elements of the fact $T(1,2,3)$.
	Hence, the treewidth of $I$ is $2$.
\begin{figure}
  \centering
  \includegraphics[scale=0.10]{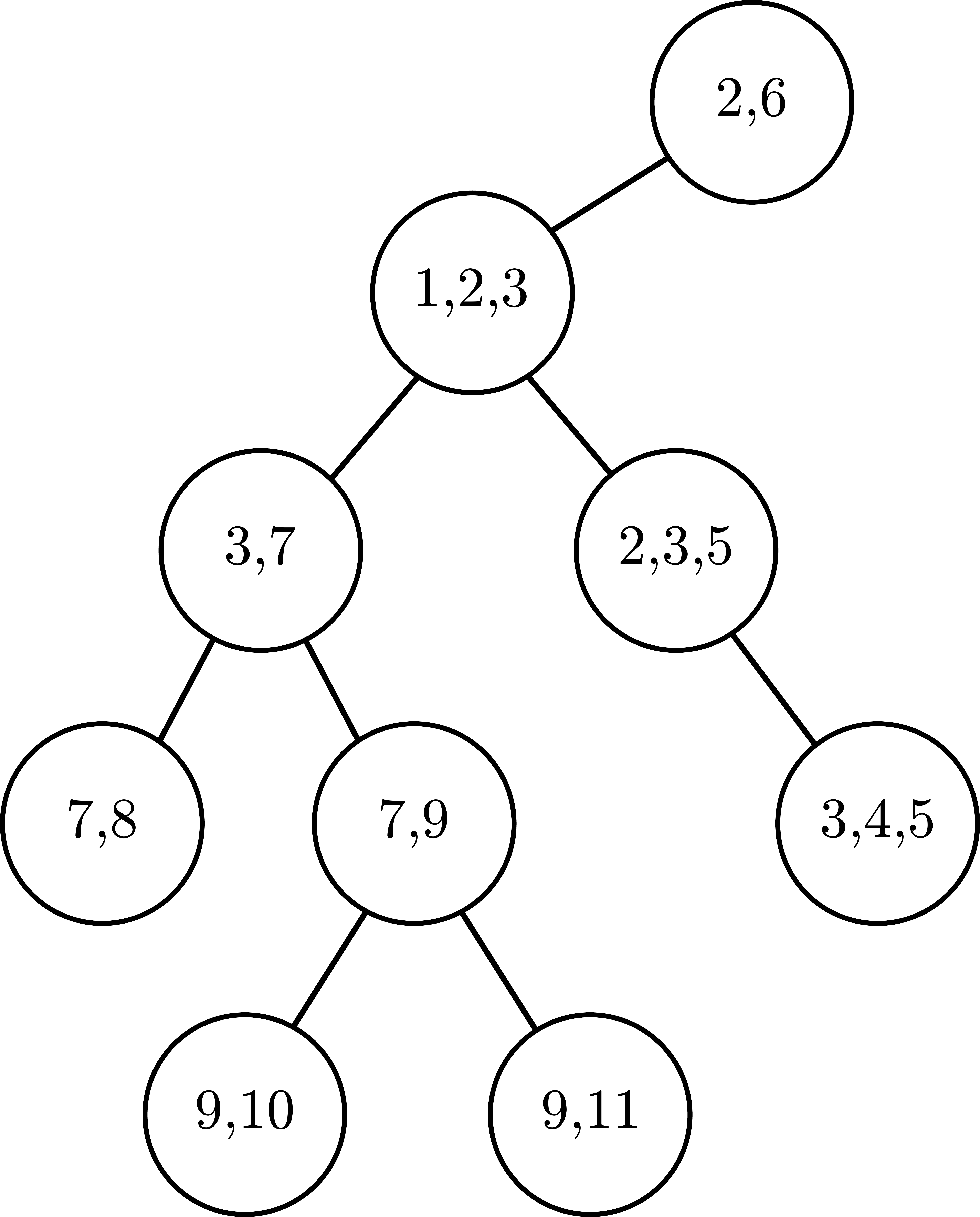}
\caption{Tree decomposition of the instance from Example~\ref{expl:instance}.}
\label{fig:treedec}
\end{figure}
\end{example}

A family of instances is \emph{treelike} if their treewidth is bounded by a
constant.

\section{Approaches for Tractability}
\label{sec:existing}
We now review existing approaches to ensure the tractability of query
evaluation, starting by query languages whose evaluation is tractable
in combined complexity on \emph{all} input instances. We then study more
expressive query languages which are tractable on \emph{treelike} instances, but
where tractability only holds in data complexity. We then present the
goals of our work.

\subsection{Tractable Queries on All Instances}
\label{sec:tractableall}

The best-known query language to ensure tractable query complexity is the language of
\emph{$\alpha$-acyclic queries}~\cite{fagin1983degrees}, i.e., those CQs that have a tree
decomposition where the domain of each bag corresponds exactly to an atom: this
is called a \emph{join tree} \cite{gottlob2002hypertree}.
With Yannakakis's algorithm
\cite{yannakakis1981algorithms}, we can evaluate an $\alpha$-acyclic
conjunctive query $Q$
on an arbitrary instance $I$ in time $O(\card{I} \cdot \card{Q})$.

Yannakakis's result was generalized in two main directions. One
direction~\cite{gottlob2014treewidth}
has investigated more general CQ classes, in particular
CQs of bounded treewidth~\cite{flum2002query},
\emph{hypertreewidth}~\cite{gottlob2002hypertree},
and \emph{fractional hypertreewidth}~\cite{grohe2014constraint}.
Bounding these query parameters to some
fixed $k$ makes query evaluation run in time
$O((\card{I} \cdot \card{Q})^{f(k)})$ for some function $f$, hence in
PTIME; for
treewidth, since the decomposition can be computed
in FPT-linear time~\cite{bodlaender1996linear}, this goes down to
$O(\card{I}^k\cdot\card{Q})$. However, query
evaluation on arbitrary instances is unlikely to be FPT when parameterized by
the query treewidth, 
since it would imply that deciding if a graph contains a $k$-clique is
FPT parameterized by $k$, which is widely believed to be false in
parameterized complexity theory (this is the $\text{W[1]}\neq\text{FPT}$ assumption).
Further, even for treewidth~2 (e.g.,
triangles), it is not known if 
we can achieve linear data complexity~\cite{alon1997finding}.

In another direction, $\alpha$-acyclicity has been generalized to queries
with more expressive operators, e.g.,
disjunction or negation. The result on $\alpha$-acyclic CQs thus extends
to the \emph{guarded fragment} (GF) of first-order logic, which can be evaluated
on arbitrary instances
in time $O(\card{I} \cdot \card{Q})$
\cite{leinders2005semijoin}.
Tractability is independently known for FO$^k$, the fragment of
FO where subformulas use at most $k$ variables, with
a simple evaluation algorithm in $O(\card{I}^k \cdot \card{Q})$~\cite{vardi1995complexity}.

Other important operators are
\emph{fixpoints}, which can be used to express, e.g., reachability
queries.
Though 
FO$^k$ is no longer tractable when adding fixpoints~\cite{vardi1995complexity},
query evaluation is tractable
for
$\mu$GF 
\cite[Theorem~3]{berwanger2001games},
i.e., GF with some restricted least and greatest fixpoint operators, when
\emph{alternation depth} is bounded; without alternation, the combined
complexity
is in $O(\card{I} \cdot \card{Q})$.
We could alternatively express fixpoints in Datalog, 
but, sadly, most known tractable fragments are nonrecursive:
nonrecursive stratified Datalog is tractable 
\cite[Corollary~5.26]{flum2002query}
for rules with restricted bodies
(i.e., strictly acyclic, or bounded strict treewidth). 
This result was generalized in \cite{gottlob2003robbers} when bounding the number
of guards: this nonrecursive fragment is shown to be equivalent to the
$k$-guarded fragment of FO, with connections to the bounded-hypertreewidth
approach.
One recursive tractable fragment is Datalog LITE, which is equivalent to alternation-free
$\mu$GF \cite{gottlob2002datalog}.
Fixpoints were independently studied for graph query languages
such as
reachability queries and \emph{regular path queries}
(RPQ), which enjoy linear combined complexity on arbitrary input instances: this
extends to two-way RPQs (2RPQs) and even strongly acyclic conjunctions of 2RPQs (SAC2RPQs),
which are expressible in alternation-free $\mu$GF. Tractability also extends to acyclic C2RPQs
but with PTIME complexity~\cite{barcelo2013querying}.

\subsection{Tractability on Treelike Instances}
\label{sec:treelike}

We now study another approach for tractable query evaluation:
this time, we restrict the shape of the \emph{instances}, using treewidth. This
ensures that we can
translate them to a tree for efficient query evaluation, using tree automata techniques.

\myparagraph{Tree encodings}
\label{apx:tree-encodings}

Informally, having fixed the signature $\sigma$, for a fixed treewidth $k \in
\NN$, we define a finite tree alphabet $\Gamma^k_\sigma$ such that
$\sigma$-instances of treewidth $\leq k$ can be translated in FPT-linear time
(parameterized by $k$), following the structure of a tree decomposition, to a
(rooted full ordered binary) \emph{$\Gamma^k_\sigma$-tree}, which we call a \emph{tree encoding}. 
 Formally:
 
 \begin{definition}
	 \label{def:tree_enc}
Let $\sigma$ be a signature, and let $k \in \NN$.
	 We define the domain $\calD_k = \{a_1, \ldots, a_{2k+2}\}$ and
the finite alphabet $\Gamma^k_\sigma$ whose elements are pairs $(d, s)$, with $d$
being a subset of up to $k+1$ elements of $\calD_k$, and $s$ being a~$\sigma$-instance consisting at most one $\sigma$-fact over some subset
of~$d$ (i.e., $\dom(s) \subseteq d$): in the latter case, we will abuse notation and identify~$s$ with the one fact that it contains. A
\emph{$(\sigma,k)$-tree encoding} is simply a rooted, binary, ordered,
full $\Gamma^k_\sigma$-tree $\la E,\lambda\ra$.
\end{definition}

The fact that $\la
  E,\lambda\ra$
  is rooted and ordered is merely for technical convenience when running bNTAs, but it is otherwise inessential.

  \begin{example}
	  The tree depicted in black in Figure~\ref{fig:treeenc} is a
          $(\{R,S,T\},2)$-tree encoding. For now, ignore the annotations
          in red and green; the link with
	  Example~\ref{expl:instance} will be explained later.
	  The domain $\calD_2$ is $\{a,b,c,d,e,f\}$, but we only use $\{a,b,c,d\}$.
  \end{example}

  A tree
encoding $\la E,\lambda\ra$ 
can be decoded
  to an
instance $I$ with the elements of~$\calD_k$ being decoded
  to new instance elements. Informally,
we create a fresh instance element
for each occurrence of an element $a_i \in \calD_k$ in an \emph{$a_i$-connected
subtree of~$E$}, i.e., a maximal connected subtree where $a_i$ appears in the
first component of the label of each node. In other words, reusing
the same $a_i$ in adjacent nodes in $\la E,\lambda\ra$ means that they stand for the same
element, and using $a_i$ elsewhere in the tree creates a new element. 
Formally:

\begin{definition}
	\label{def:tree_enc_decode}
	Let $\la E, \lambda \ra$ be a $(\sigma,k)$-tree encoding, where, for each node $n$ of $E$, we write $\lambda(n) = (d_n,s_n)$.
	A set $S$ of \emph{bag decoding functions for $\la E, \lambda \ra$} consists of one function $\dec_n$ with domain $d_n$ for every node $n$ of $E$.
	We say that $S$ is \emph{valid} if $S$ satisfies the following
        condition:
		for every $a \in \calD_k$ and nodes $n_1, n_2$ of $E$
                such that $a \in d_{n_1}$ and $a \in d_{n_2}$,
                we have $\dec_{n_1}(a) = \dec_{n_2}(a)$ if and only if $n_1$ and $n_2$ are in the same
			$a$-connected subtree of $\la E, \lambda \ra$.
\end{definition}

\begin{example}
	\label{expl:valid_bag_dec}
	Consider again the tree encoding $\la E, \lambda \ra$ in Figure~\ref{fig:treeenc}. 
	For each node $n \in E$, let $\dec_n$ be the function that is defined by the green annotations next to $n$.
	Then one can check that $S = \{\dec_n \mid n \in E\}$ is a valid set of bag decoding functions for $\la E, \lambda \ra$.
\end{example}

	We can use a valid set of bag decoding functions $S$ to decode a tree encoding $\la E, \lambda \ra$ to a $\sigma$-instance:
\begin{definition}
        Let $\la E, \lambda \ra$ be a $(\sigma,k)$-tree encoding, where, for each node $n$ of $E$, we write $\lambda(n) = (d_n,s_n)$, and let $S$ be a 
	valid set of bag decoding functions for $\la E, \lambda \ra$.
	The $\sigma$-instance $\dec_S(\la E, \lambda \ra)$ is defined as
        follows.
	The elements of $\dec_S(\la E, \lambda \ra)$ are $\{\dec_n(a) \mid n \in E,~a \in d_n\}$.
	The facts of $\dec_S(\la E, \lambda \ra)$ are $\{R(\dec_n(\mathbf{a})) \mid n \in E,~ s_n = R(\mathbf{a})\}$.
\end{definition}

\begin{example}
	Continuing Example~\ref{expl:valid_bag_dec}, consider again the tree encoding $\la E, \lambda \ra$ and valid set $S$ of decoding functions for $\la E, \lambda \ra$.
	Then, computing $\dec_S(\la E, \lambda \ra)$ yields the instance from Example~\ref{expl:instance}.
\end{example}

A tree encoding can have multiple valid sets $S$ of bag decoding functions. 
However, the choice of $S$ does not matter since they all decode to isomorphic instances:

\begin{lemma}[\cite{amarilli2015provenance}]
  Let $\la E, \lambda \ra$ be a tree encoding, and $S_1, S_2$ be two valid sets of bag decoding functions of $\la E, \lambda \ra$.
  Then $\dec_{S_1}(\la E, \lambda \ra)$ and $\dec_{S_2}(\la E, \lambda \ra)$ are isomorphic.
\end{lemma}

Hence, we will now write $\dec(\la E, \lambda \ra)$, forgetting the
subscript $S$, since we are not interested in distinguishing isomorphic
instances (and since there always exists at least one valid
set of bag decoding functions, for every tree encoding).
Furthermore, it is easy to see that $\decode{\la E,\lambda\ra}$ has treewidth $\leq k$, as a tree decomposition for it can
be constructed from $\la E,\lambda\ra$. 
Conversely, for any instance $I$ of treewidth $\leq k$, we can compute a
$(\sigma,k)$-encoding $\la
E,\lambda\ra$ such that $\decode{\la E,\lambda\ra}$ is $I$ (up
to isomorphism). 
We say that $\la E, \lambda \ra$ \emph{is a tree encoding of $I$}:

\begin{definition}
	Let $I$ be a $\sigma$-instance, and $\la E, \lambda \ra$ be a $(\sigma,k)$-tree encoding (for some $k \in \NN$).
	We say that $\la E, \lambda \ra$ \emph{is a tree encoding of $I$} if $\dec(\la E, \lambda \ra)$ is $I$, up to isomorphism.
\end{definition}

Informally, given $I$ of treewidth $\leq k$, we can construct a
$(\sigma,k)$-tree encoding $\la E, \lambda \ra$ of $I$ from a tree
decomposition of $I$ as follows: copy each
bag of the decomposition multiple times so that 
each fact can be coded in a separate node; arrange these copies in a
binary tree to make the tree encoding binary; make the tree encoding full by
adding empty nodes.
We can easily show that this process is FPT-linear for $k$,
so that we will use the following claim (see \cite{amarilli2016leveraging} for
  our type of encodings):

\begin{lemma}[\cite{flum2002query}]
  \label{lem:getencoding}
  The problem, given an instance~$I$ of treewidth $\leq k$,
  of computing a tree encoding of $I$, is FPT-linear parameterized by~$k$.
\end{lemma}

We sum up this discussion about tree encoding with the full example.

\begin{example}
	Remember that Figure~\ref{fig:treeenc} presents a tree encoding $\la E,\lambda\ra$ for $k=2$ and the signature $\sigma$ from Example~\ref{expl:instance}. 
	Remember we can decode $\la E,\lambda\ra$ by the mappings (drawn in green), obtaining this way 
	the instance $I$ from Example~\ref{expl:instance}. 
	Hence $\la E,\lambda\ra$ is a tree encoding of $I$.
	Moreover, we recall that any valid way of decoding $\la E,\lambda\ra$ would yield an instance isomorphic to $I$.
	We point out a few details that can help understand how these encodings work. 
	The elements ``a'' in bags $\alpha$ and $\beta$ are decoded to distinct instance elements, since $\alpha$
	and $\beta$ are not in a same $a$-connected subtree of the tree encoding. 
	Bags like $\gamma$, that contain elements but no fact, are usually used in order to help making the tree encoding binary. Empty
	bags like $\delta$ can be used in order to make the tree encoding full.
\begin{figure}
  \centering
  \includegraphics[scale=0.12]{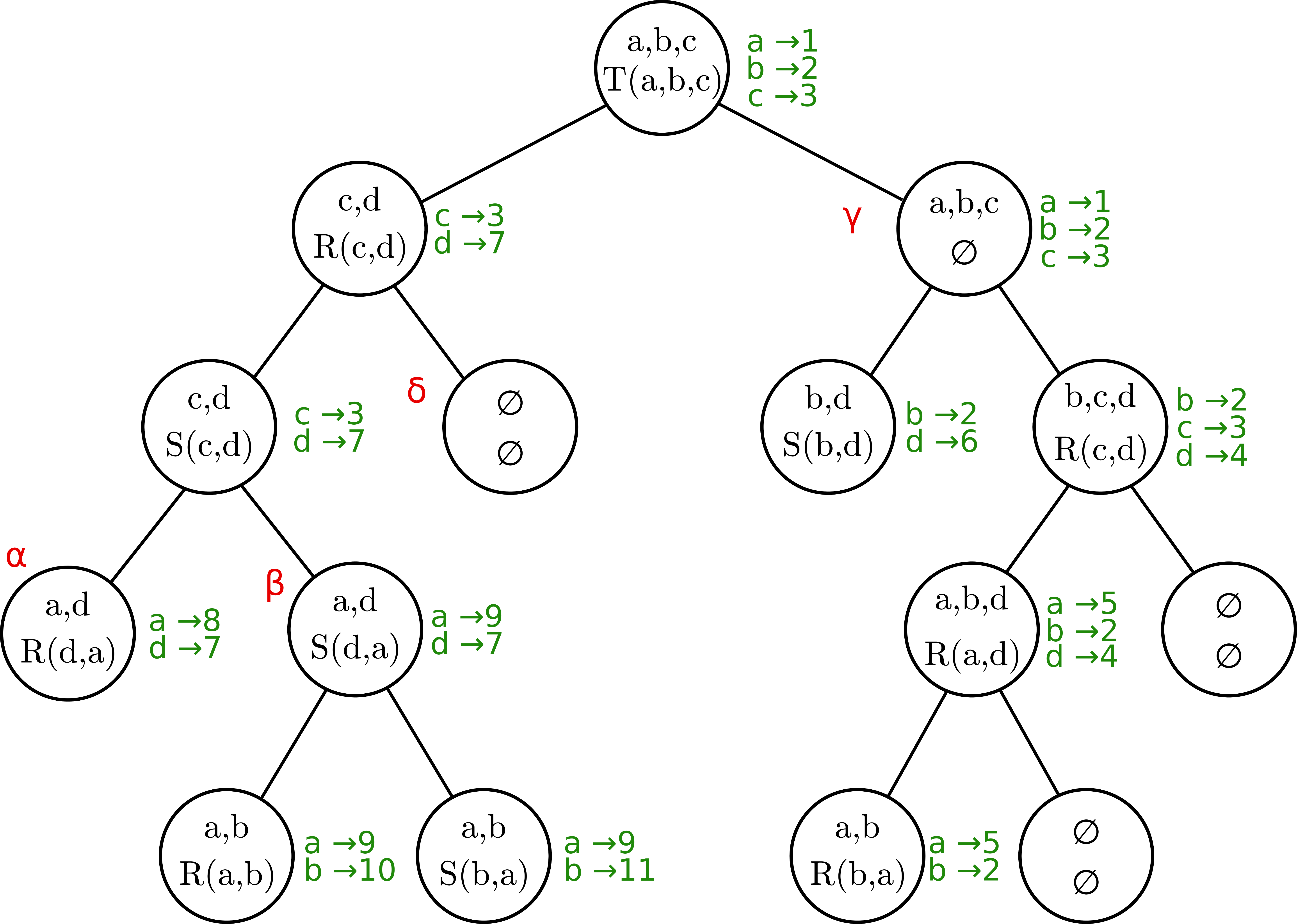}

\caption{Tree encoding of the relational instance from Example~\ref{expl:instance}.}
\label{fig:treeenc}
\end{figure}
\end{example}

\myparagraph{Bottom-up tree automata}

We can then evaluate queries on treelike instances by running \emph{tree
automata} on the tree encoding that represents them.
Formally, given an alphabet $\Gamma$, a \emph{bottom-up
nondeterministic tree automaton}
on $\Gamma$-trees (or $\Gamma$-bNTA) is a tuple $A = (Q,F,\iota,\Delta)$,
where:

\begin{compactenum}[(i)]
	\item $Q$ is a finite set of \emph{states};
	\item $F \subseteq Q$ is a subset of \emph{accepting states};
	\item $\iota : \Gamma \to 2^Q$ 
          is an \emph{initialization function} 
          determining the possible states of a leaf from its label;
	\item $\Delta :  \Gamma \times Q^2
          \to 2^Q$ 
          is a \emph{transition function} 
		determining the possible states for an internal 
                node from its label and the states of its two
                children.
\end{compactenum}
Given a $\Gamma$-tree $\la T,\lambda\ra$ 
(where $\lambda : T \to \Gamma$ is the \emph{labeling function}),
we define a \emph{run} of $A$ on~$\la T,\lambda\ra$
as a function $\phi : T \to Q$ such that
(1)~$\phi(l) \in \iota(\lambda(l))$  for every leaf $l$ of~$T$; and
(2)~$\phi(n) \in \Delta( \lambda(n), \phi(n_1),\phi(n_2))$
for every internal node $n$ of~$T$ with children $n_1$ and $n_2$.
The bNTA $A$ \emph{accepts} $\la T,\lambda\ra$
if
it has a run on~$T$ mapping
the root of~$T$ to a state of~$F$.

We say that a bNTA $A$ \emph{tests a query $Q$ on instances of treewidth~$\leq k$} if, for any
$\Gamma^k_\sigma$-encoding $\la E,\lambda\ra$ coding an instance $I$ (of
treewidth $\leq k$), 
$A$ accepts $\la E,\lambda\ra$ iff $I \models Q$.
By a well-known result of Courcelle \cite{courcelle1990monadic} on graphs
(extended to higher-arity in \cite{flum2002query}),
we can use bNTAs to evaluate all queries in \emph{monadic second-order
logic} (MSO), i.e., first-order logic with second-order variables of arity~$1$.
MSO subsumes in particular CQs 
and monadic
Datalog (but not general Datalog).
Courcelle showed that MSO
queries can be translated to a bNTA that tests them:

\begin{theorem}[\cite{courcelle1990monadic,flum2002query}]
  \label{thm:courcelle}
  For any MSO query $Q$ and treewidth $k \in \NN$, we can
  compute a bNTA that tests $Q$ on instances of treewidth $\leq k$.
\end{theorem}

This implies that evaluating any MSO query $Q$ has FPT-linear
\emph{data complexity} when parameterized by $Q$ 
and the instance treewidth~\cite{courcelle1990monadic,flum2002query},
i.e., is in $O\left(f(\card{Q}, k) \cdot
  \card{I}\right)$ for some computable function~$f$.
However, this tells little about the combined complexity, as 
$f$ is generally nonelementary 
in $Q$ \cite{meyer1975weak}.
A better combined complexity bound is known for unions of conjunctions of two-way regular path queries (UC2RPQs) that are further required to be acyclic and to have a constant number
of edges between pairs of variables: these can be translated into
polynomial-sized alternating two-way automata \cite{barcelo2014does}.

\subsection{Restricted Queries on Treelike Instances}

Our approach combines both ideas: we use instance treewidth as a parameter, but also restrict the queries to ensure tractable
translatability.
We are only aware of two approaches in this spirit.
First, 
Gottlob, Pichler, and Wei~\cite{gottlob2010monadic}
have proposed a
\emph{quasiguarded} Datalog fragment on \emph{relational structures
and their tree decompositions}, 
for which query evaluation is in $O(|I|\cdot|Q|)$. However,
this formalism requires queries to be expressed in terms of the tree
decomposition, and not just in terms of the relational signature.
Second, Berwanger and Grädel~\cite{berwanger2001games} remark (after Theorem~4) that, when
alternation depth and \emph{width} are bounded,
$\mu$CGF (the 
\emph{clique-guarded} fragment of FO with fixpoints) enjoys FPT-bilinear query evaluation when
parameterized by instance treewidth.
Their approach does not rely on automata methods, and subsumes the tractability
of $\alpha$-acyclic CQs and alternation-free $\mu$GF (and hence SAC2RPQs), on treelike instances.
However, $\mu$CGF is a restricted query language (the only CQs
that it can express are those with a chordal primal graph),
whereas we want a richer language, with a
parameterized definition.

Our goal is thus to develop an expressive parameterized
query language, which can be translated in FPT-linear time to an automaton
that tests it (with the treewidth of instances also being a parameter). We can then evaluate
the automaton, and obtain FPT-bilinear combined complexity for query evaluation.
Further, as we will show, the use of tree
automata will yield \emph{provenance representations} for the
query as in~\cite{amarilli2015provenance} (see
Section~\ref{sec:provenance}).

\section{Conjunctive Queries on Treelike Instances}
\label{sec:cq}
To identify classes of queries that can be efficiently translated to tree
automata, we start by the simplest queries: \emph{conjunctive
queries}.

\subparagraph*{$\parabm{\alpha}$-acyclic queries.}
A natural candidate for a tractable query class via automata methods would be
\emph{$\alpha$-acyclic} CQs, which, as we explained in
Section~\ref{sec:tractableall},
can be evaluated in time $O(\card{I} \cdot \card{Q})$ on all instances. Sadly,
we show that such queries cannot be translated efficiently to bNTAs, so
the translation result of Theorem~\ref{thm:courcelle} does not extend directly:
\begin{proposition}
  \label{prp:bntalower}
  There is an arity-two signature $\sigma$ and an infinite family
  $(Q_i)_{i \in \NN}$	
  of $\alpha$-acyclic CQs
  such
  that, for any $i \in \NN$, any bNTA that tests~$Q_i$ on instances of treewidth $\leq 1$ must have 
  $\Omega(2^{\card{Q_i}^{1-\epsilon}})$ states for any $\epsilon>0$.
\end{proposition}

\begin{proof}
  We fix the signature $\sigma$ to consist of binary relations $S$, $S_0$, $S_1$, and
  $C$. We will code binary numbers as gadgets on this fixed signature. The
  coding of $i \in \NN$ at length~$k$, with $k \geq 1 + \lceil \log_2 i
  \rceil$,
  consists of an $S$-chain $S(a_1, a_2), \ldots,
  S(a_{k-1}, a_k)$, and facts $S_{b_j}(a_{j+1},
  a'_{j+1})$ for $1\leq j\leq k-1$ where $a'_{j+1}$ is a fresh element and $b_j$ is the $j$-th bit in the binary expression of $i$ (padding
  the most significant bits with 0). We now
  define the query family $Q_i$: each $Q_i$ is formed by picking a root variable
  $x$ and gluing $2^i$ chains to~$x$; for $0\leq j\leq 2^i-1$,
  we have one chain that is the 
  concatenation of a chain of~$C$ of length $i$ and the coding of
  $j$ at length~$(i+1)$ using a gadget. Clearly the size of $Q_i$ is $\Theta(i \times
  2^i)$.
  Now, fix $\epsilon>0$.
  As $|Q_i|$ is in $O(i \times 2^i)$, there exist $N, \beta > 0$ such that $\forall i \geq N,~|Q_i| \leq \beta \times i \times 2^i$. 
  But clearly, there exists $M > 0$ such that $\forall i \geq M,~(\beta \times i \times 2^i)^{1-\epsilon} \leq 2^i - i/2$.
  Hence for $i \geq \max(N,M)$ we have $2^i - i/2 \geq |Q_i|^{1 - \epsilon}$.

  Fix $i > 0$. Let $A$ be a bNTA testing $Q_i$ on instances of treewidth
  $1$. We will show that $A$~must have at least $\binom{2^{i}}{2^{i-1}}=\Omega\left(2^{2^{i}-\frac{i}{2}}\right)$ states (the lower
  bound is obtained from Stirling's formula), from which
  the claim follows since for $i \geq \max(N,M)$ we have $2^{2^{i}-\frac{i}{2}}\geq 2^{|Q_i|^{1-\epsilon}}$.
  In fact, we will consider a specific subset $\calI$ of the instances
  of treewidth $\leq 1$, and a specific set $\calE$ of tree encodings of
  instances of $\calI$, and show the claim on $\calE$, which suffices to
  conclude.

  To define $\calI$, let $\calS_i$ be the set of subsets of
  $\{0, \ldots, 2^i-1\}$ of cardinality $2^{i-1}$, so that $\card{\calS_i}$ is
  $\binom{2^{i}}{2^{i-1}}$. 
  We will first define a family~$\calI'$ of instances indexed by $\calS_i$
  as follows. Given $S \in \calS_i$, the
  instance $I'_S$ of $\calI'$ is obtained by constructing a full binary tree of
  the $C$-relation of height $i-1$, and identifying, for all $j$, the $j$-th
  leaf node with element $a_1$ of the length-$(i+1)$ coding of the $j$-th
  smallest number in~$S$.
  We now define the instances of $\calI$ to consist of a root element
  with two $C$-children, each of which are the root element of an instance of
  $\calI'$ (we call the two the \emph{child instances}). It is clear
  that instances of $\calI$ have treewidth $1$, and we can check quite easily
  that an instance of $\calI$ satisfies $Q_i$ iff the child instances $I'_{S_1}$ and
  $I'_{S_2}$ are such that $S_1 \cup S_2 = \{1, \ldots, 2^i\}$.
  
  We now define $\calE$ to be tree encodings of instances of $\calI$. First,
  define $\calE'$ to consist of tree encodings of instances of~$\calI'$, which
  we will also index with $\calS_i$, i.e., $E_S$ is a tree encoding of $I'_S$. We
  now define $\calE$ as the tree encodings $E$ constructed as follows:
  given an instance $I \in \calI$, we encode it as a root bag with
  domain $\{r\}$, where $r$ is the root of the tree $I$, and no fact, the
  first child $n_1$ of the root bag having
  domain $\{r, r_1\}$ and fact $C(r, r_1)$, the second child $n_2$
  of the root being defined in the same way. Now, $n_1$ has one dummy
  child with empty domain and no fact, and one child
  which is the root of some tree encoding in~$\calE$ of one child instance of
  $I$. We define $n_2$ analogously with the other child instance.

  For each $S \in \calS_i$, letting $\bar S$ be the complement of~$S$ relative
  to $\{0, \ldots, 2^i-1\}$, we call $I_S \in \calI$ the instance where the
  first child instance is $I'_S$ and the second child instance is $I'_{\bar S}$,
  and we call $E_S \in \calE$ the tree encoding of~$I_S$ according to the
  definition above. We then call $\calQ_S$ the set of states $q$ of~$A$ such that
  there exists a run of~$A$ on~$E_S$ where the root of the encoding of the first child
  instance is mapped to~$q$. As each $I_S$ satisfies~$Q$, each $E_S$
  should be accepted by the automaton, so each $\calQ_S$ is non-empty.
  
  Further,
  we show that the $\calQ_S$ are pairwise disjoint: for any $S_1 \neq S_2$ of~$\calS_i$,
  we show that $\calQ_{S_1} \cap \calQ_{S_2} = \emptyset$. Assume to the
  contrary the existence of $q$ in the intersection, and let $\rho_{S_1}$ and
  $\rho_{S_2}$ be runs of~$A$ respectively on $I_{S_1}$ and $I_{S_2}$ that
  witness respectively that $q \in \calQ_{S_1}$ and $q \in \calQ_{S_2}$. Now,
  consider the instance $I \in \calI$ where the first child instance is~$I_1$,
  and the second child instance is $\bar{I_2}$, and let $E \in \calE$ be the tree
  encoding of~$I$. We can construct a run~$\rho$ of~$A$ on~$E$ by defining
  $\rho$ according to~$\rho_{S_2}$ except that, on the subtree of~$E$ rooted at
  the root $r'$ of the tree encoding of the first child instance, $\rho$ is defined
  according to~$\rho_{S_1}$: this is possible because $\rho_{S_1}$ and
  $\rho_{S_2}$ agree on~$r_1'$ as they both map $r'$ to~$q$. Hence, $\rho$
  witnesses that $A$ accepts $E$. Yet, as $I_1 \neq I_2$, we know that $I$ does
  not satisfy~$Q$, so that, letting $E \in \calE$ be its tree encoding, $A$
  rejects~$E$. We have reached a contradiction, so indeed the $\calQ_S$ are
  pairwise disjoint.

  As the $\calQ_S$ are non-empty, 
  we can construct a mapping from $\calS_i$ to the state
  set of $A$ by mapping each $S \in \calS_i$ to some state of~$\calQ_S$: as the
  $\calQ_S$ are pairwise disjoint, this mapping is injective.
  We deduce that the state set of~$A$ has size at least $\card{\calS_i}$, which
  concludes from the bound on the size of~$\calS_i$ that we showed previously.
\end{proof}

Faced by this, we propose to use different tree automata formalisms,
which are generally more concise than bNTAs.
There are two classical generalizations of nondeterministic automata, on words~\cite{birget1993state}
and on trees~\cite{tata}: one goes from the inherent existential
quantification of nondeterminism to
\emph{quantifier alternation}; the other allows
\emph{two-way} navigation instead of
imposing a left-to-right (on
words) or bottom-up (on trees) traversal. On words, both of these extensions
independently allow for exponentially more compact
automata~\cite{birget1993state}. In this work, we combine both extensions
and use
\emph{alternating two-way
tree automata}~\cite{tata,cachat2002two}, formally introduced in
Section~\ref{sec:compilation}, which leads to tractable 
combined complexity for evaluation.
Our general results in the next section will then imply:

\begin{proposition}\label{prp:alphaatwa}
  For any treewidth bound $\ki \in \NN$, given an $\alpha$-acyclic CQ $Q$, we can
  compute in FPT-linear time in $O(\card{Q})$ (parameterized by $\ki$) an alternating two-way tree automaton that tests
  it on instances of treewidth $\leq \ki$.

  Hence, if we are additionally given a relational instance $I$ of treewidth
  $\leq \ki$, one can determine whether $I\models Q$ in FPT-bilinear time in
  $\card{I}\cdot\card{Q}$ (parameterized by $\ki$).
\end{proposition}

\begin{proof}
This proof depends on notions and
results that are given in the rest of the paper, and should be read after studying the rest of this paper.

  Given the $\alpha$-acyclic CQ $Q$, we can compute in linear time in~$Q$ a
  chordal decomposition~$T$ (also called a join tree) of~$Q$ by using Theorem~5.6
	of~\cite{flum2002query} (attributed to~\cite{tarjan1984simple}). 
  We recall that a chordal decomposition of a CQ $Q$ is a 
	tree decomposition $T$ of $Q$ such that for every bag $b$ of $T$, there
	exists an atom of $Q$ whose variables is exactly the variables in $b$.
As $T$ is in
  particular a simplicial decomposition of~$Q$ of width $\leq \arity{\sigma} -
  1$, i.e., of constant width, we use Proposition~\ref{prp:simplicial} to
  obtain in linear time in $\card{Q}$ a
  CFG-Datalog program $P$ equivalent to $Q$ with body size bounded by a constant $\kp$.
  
  We now use Theorem~\ref{thm:maintheorem} to construct, in FPT-linear time in
  $\card{P}$ (hence, in $\card{Q}$), parameterized by $\ki$ and the constant
  $\kp$, an automaton $A$ testing $P$ on instances of treewidth $\leq \ki$; specifically, a 
  stratified isotropic alternating two-way automata or SATWA (to be introduced
  in Definition~\ref{def:satwa}).

  We now observe that, thanks to the fact that $Q$ is monotone, the SATWA $A$
  does not actually feature any negation: the translation in the proof of
  Proposition~\ref{prp:simplicial} does not produce any negated atom, and
  the translation in the proof of Theorem~\ref{thm:maintheorem} only
  produces a negated state within a Boolean formula when there is a corresponding
  negated atom in the Datalog program. Hence, $A$ is
  actually an alternating two-way tree automaton, which proves the first part of
  the claim.

  For the second part of the claim, we use Theorem~\ref{thm:main} to evaluate
  $P$ on $I$ in FPT-bilinear time in $\card{I} \cdot \card{P}$, parameterized by
  the constant $\kp$ and $\ki$. This proves the claim.
\end{proof}

\subparagraph*{Bounded-treewidth queries.}
Having re-proven the combined tractability of $\alpha$-acyclic queries (on
bounded-treewidth instances),
we naturally try to extend to \emph{bounded-treewidth} CQs. Recall from
Section~\ref{sec:tractableall} that these queries have PTIME combined
complexity on all instances, but are unlikely to be FPT when parameterized by
the query treewidth (unless $\text{W[1]}=\text{FPT}$). Can they be efficiently evaluated on
\emph{treelike} instances by translating them to automata? We answer in the negative:
that bounded-treewidth CQs \emph{cannot} be
efficiently translated to automata to test them, even when using the expressive
formalism of alternating two-way tree automata:

\begin{toappendix}
  \section{Proof of Theorem~\ref{thm:nocompile}}
\end{toappendix}

\begin{theoremrep}\label{thm:nocompile}
  There is an arity-two signature $\sigma$ for which there is no
  algorithm $\calA$ 
  with exponential running time and polynomial output size for the 
  following task:
  given a conjunctive query $Q$ of treewidth $\leq 2$,
  produce an alternating two-way tree automaton $A_Q$ on $\Gamma^5_\sigma$-trees
  that tests $Q$ on
  $\sigma$-instances of treewidth~$\leq 5$.
\end{theoremrep}

\begin{toappendix}
  \label{app:nocompile}
To prove this theorem, we need some notions and lemmas from
  \cite{unpublishedbenediktmonadic}, an extended version
  of~\cite{benedikt2012monadic}. Since \cite{unpublishedbenediktmonadic}
  is currently unpublished, relevant results are reproduced as Appendix~F
  of~\cite{amarilli2017combinedb}, in particular Lemma~68, Theorem~69,
  and their proofs.
\end{toappendix}

\begin{toappendix}
\begin{proof}[Proof of Theorem~\ref{thm:nocompile}]
  Let $\sigma$ be $\schildself$ as in Theorem~69
  of~\cite{amarilli2017combinedb}. We
  pose $c=3$, $\ki=2\times 3-1=5$.
  Assume by way of
  contradiction that there exists an algorithm $\calA$ satisfying the prescribed
  properties.
  We will describe an algorithm to solve any instance of the
  containment problem of Theorem~69 of~\cite{amarilli2017combinedb} in singly exponential
  time. As Theorem~69 of~\cite{amarilli2017combinedb} states that it is 2EXPTIME-hard, this
  yields a contradiction by the time hierarchy theorem.

  Let $P$ and $Q$ be an instance of the containment problem of
  Theorem~69 of~\cite{amarilli2017combinedb}, where $P$ is a monadic Datalog program of
  var-size $\leq 3$,
  and $Q$ is a CQ of treewidth $\leq 2$. We will show how to solve the
  containment problem, that is, decide whether there exists some instance $I$
  satisfying $P \land \neg Q$.
  
  Using
  Lemma~68 of~\cite{amarilli2017combinedb}, compute in singly exponential time the
  $\Gamma_\sigma^{\ki}$-bNTA
  $A_P$. Using the putative algorithm $\calA$ on~$Q$, compute in singly
  exponential time an alternating two-way automaton $A_Q$ of polynomial size. As $A_P$ describes a family $\calI$ of canonical
  instances for~$P$, there is an instance satisfying $P \wedge \neg Q$ iff there is an
  instance in $\calI$ satisfying $P \wedge \neg Q$. Now, as $\calI$ is described
  as the decodings of the language of~$A_P$, all instances in $\calI$ have
  treewidth~$\leq \ki$. Furthermore, the instances in~$\calI$ satisfy $P$ by definition of $\calI$.
  Hence, there is an instance satisfying $P\wedge\neg Q$ iff there is an encoding $E$ in the
  language of $A_P$ whose decoding satisfies $\neg Q$. Now, as $A_Q$ tests $Q$ on
  instances of treewidth $\ki$, this is the case iff there is an encoding $E$ in the
  language of $A_P$ which is not accepted by $A_Q$. Hence, our problem is
  equivalent to the problem of deciding whether there is a tree accepted by
  $A_P$ but not by $A_Q$.
  
  We now use Theorem~A.1 of~\cite{cosmadakis1988decidable} to compute in EXPTIME
  in~$A_Q$ a bNTA $A'_Q$ recognizing the complement of the language
  of~$A_Q$. Remember that $A_Q$ was computed in EXPTIME and is of polynomial
  size, so the entire process so far is EXPTIME. Now we know that we can solve
  the containment problem by testing whether $A_P$ and $A'_Q$ have non-trivial
  intersection, which can be done in PTIME by computing the product automaton
  and testing emptiness~\cite{tata}. This solves the containment problem in
  EXPTIME. As we explained initially, we have reached a contradiction, because
  it is 2EXPTIME-hard.
\end{proof}
\end{toappendix}

This result is obtained from a variant of the 2EXPTIME-hardness of monadic
Datalog containment \cite{benedikt2012monadic}. As this result heavily
relies on \cite{unpublishedbenediktmonadic}, an unpublished extension
of~\cite{benedikt2012monadic} whose relevant results are reproduced
in~\cite{amarilli2017combinedb}, we deport its proof to
Appendix~\ref{app:nocompile}.
Briefly, we show that efficient translation of bounded-treewidth CQs to
automata would yield an EXPTIME containment test, and conclude by the
time hierarchy theorem.

\subparagraph*{Bounded simplicial width.}
We have shown that we cannot translate bounded-treewidth queries to automata
efficiently.
We now show that efficient translation can be ensured with an additional
requirement on tree decompositions.
As it turns out, the resulting decomposition notion
has been independently introduced for graphs:

\begin{definition}[\cite{diestel1989simplicial}]
  A \emph{simplicial decomposition of a graph $G$} is a tree decomposition $T$ of~$G$
  such that, for any bag $b$ of $T$ and child bag $b'$ of~$b$, if $S$ is
  the intersection of the domains of $b$ and $b'$, then the subgraph of $G$ induced by $S$ is a complete
  subgraph of~$G$.
\end{definition}

We extend this notion to CQs, and introduce the \emph{simplicial width} measure:

\begin{definition}
  \label{def:simplicialcq}
  A \emph{simplicial decomposition of a CQ $Q$} is a simplicial
  decomposition of its primal graph.
  Note that any
  CQ has a simplicial
  decomposition (e.g., the trivial one that puts all variables in one
  bag).
  The \emph{simplicial width} of~$Q$
  is the minimum, over all simplicial tree decompositions, of the size of
  the largest bag minus~$1$.
\end{definition}

Bounding the simplicial width of CQs is of course more restrictive than
bounding their treewidth, and
this containment relation is strict: cycles have
treewidth $\leq 2$ but 
have unbounded simplicial width. This
being said, bounding the simplicial width is less restrictive than imposing
$\alpha$-acyclicity:
the join tree of an $\alpha$-acyclic CQ is
in particular a simplicial decomposition, so $\alpha$-acyclic CQs have
simplicial width at most $\arity{\sigma}-1$, which is constant as $\sigma$ is
fixed. Again, the containment is strict: a triangle has simplicial width $2$ but is not $\alpha$-acyclic.

To our knowledge, 
simplicial width for CQs has not been studied before.
Yet, we show that bounding the
simplicial width ensures that CQs can be efficiently translated to
automata. 
This is
in contrast to bounding the
treewidth, which we have shown in~Theorem~\ref{thm:nocompile} not to be sufficient to ensure efficient translatability to tree automata. Hence:

\begin{theorem}\label{thm:simplicial}
  For any $\ki, \kq \in \NN$, given a CQ $Q$ and a simplicial decomposition $T$
  of simplicial width $\kq$
  of~$Q$,
  we can compute in FPT-linear in~$\card{Q}$ (parameterized by~$\ki$ and~$\kq$)
  an alternating two-way tree automaton that tests
  $Q$ on instances of treewidth $\leq \ki$.

  Hence, if we are additionally given a relational instance~$I$ of treewidth
  $\leq \ki$,
  one can determine
  whether $I\models Q$ in FPT-bilinear time in $|I|\cdot(|Q|+|T|)$ 
  (parameterized by~$\ki$ and~$\kq$).
\end{theorem}

\begin{proof}
This proof depends on notions and
results that are given in the rest of the paper, and should be read after studying the rest of this paper.

  We use Proposition~\ref{prp:simplicial} to transform
  the CQ $Q$ to a
  CFG-Datalog program $P$ with body size at most $\kp \defeq f_\sigma(\kq)$,
  in FPT-linear time in $\card{Q} + \card{T}$ parameterized by~$\kq$. 
  
  We now use Theorem~\ref{thm:maintheorem} to construct, in FPT-linear time in
  $\card{P}$ (hence, in $\card{Q}$), parameterized by $\ki$ and~$\kp$, hence
  in~$\ki$ and~$\kq$, a SATWA $A$ testing $P$ on instances of treewidth $\leq \ki$ (see
  Definition~\ref{def:satwa}). For the same
  reasons as in the proof of Proposition~\ref{prp:alphaatwa}, it is actually a
  two-way alternating tree automaton, so we have shown the first part of the result.

  To prove the second part of the result,
  we now use Theorem~\ref{thm:main} to evaluate $P$ on $I$ in FPT-bilinear time
  in $\card{I} \cdot \card{P}$, parameterized by~$\kp$ and~$\ki$, hence again
  by~$\kq$ and~$\ki$. This proves the claim.
\end{proof}

Notice the technicality that the simplicial decomposition $T$ must be provided
as input to the procedure, because it is not known
to be computable in FPT-linear time, unlike tree decompositions.
While we are not aware of results on the
complexity of this specific task,
quadratic-time algorithms are known for the related problem of computing the
\emph{clique-minimal separator decomposition}
\cite{leimer1993optimal,berry2010introduction}.

The intuition for the efficient translation of bounded-simplicial-width CQs
is as follows. The \emph{interface} variables shared between any bag and its parent must
be ``clique-guarded'' (each pair is covered by an atom). Hence, consider
any subquery rooted at a bag of the query decomposition, and see it as a
non-Boolean CQ
with the interface variables as free variables. Each result of this
CQ must then be covered by a clique of facts of the instance, which ensures
\cite{gavril1974intersection} that it occurs in some bag of the instance tree decomposition and can be
``seen'' by a tree automaton.
This intuition can be generalized, beyond conjunctive queries, to design an
expressive query language featuring disjunction, negation, and fixpoint, with
the same properties of efficient translation to automata and FPT-linear
combined complexity of evaluation on treelike instances. We introduce such a
Datalog variant in the next section.

\section{CFG-Datalog on Treelike Instances}
\label{sec:CFG}
To design a Datalog fragment with efficient translation to automata, we must of
course impose some limitations, as we did for CQs. In fact, we can even show
that the full Datalog language (even without negation) \emph{cannot} be translated to automata, no matter
the complexity:

\begin{proposition}\label{prp:dlnotftar}
  There is a signature $\sigma$ and Datalog program $P$ 
  such that the language of $\Gamma_\sigma^1$-trees that encode instances
  satisfying $P$ is not a regular tree language.
  \end{proposition}

\begin{proof}
Let $\sigma$ be the signature containing two binary relations $Y$ and $Z$ and
  two unary relations $\mathrm{Begin}$ and $\mathrm{End}$.
Consider the following program $P$:
\begin{align*}
  \mathrm{Goal}() &\leftarrow S(x,y), \mathrm{Begin}(x), \mathrm{End}(y)\\
                      S(x,y) &\leftarrow Y(x,w), S(w,u), Z(u,y) \\ 
                      S(x,y) &\leftarrow Y(x,w), Z(w,y)
\end{align*}
Let $L$ be the language of the tree encodings of instances of treewidth~$1$ that
satisfy~$P$. We will show that $L$ is not a regular tree language, which clearly
  implies the second claim, as a bNTA or an alternating two-way tree
  automaton can
  only recognize regular tree languages \cite{tata}. To show this, let us assume by
  contradiction that $L$ is a regular tree language, so that there exists a
  $\Gamma^1_\sigma$-bNTA $A$ that accepts $L$, i.e., that tests $P$.

We consider instances that are chains of facts which are either $Y$- or
$Z$-facts, and where the first end is the only node labeled $\mathrm{Begin}$ and
the other end is the only node labeled $\mathrm{End}$. This condition on instances
  can clearly be expressed in MSO, so that by Theorem~\ref{thm:courcelle} there exists
a bNTA $A_{\mathrm{chain}}$ on~$\Gamma_\sigma^1$ that tests this property. In particular, we can build the
bNTA $A'$ which is the intersection of $A$ and $A_{\mathrm{chain}}$, which tests whether
instances are of the prescribed form and are accepted by the program $P$.

We now observe that such instances must be the instance
\begin{align*}
  I_k = {} & \{
  \mathrm{Begin}(a_1), \allowbreak
  Y(a_1, a_2), \ldots, \allowbreak
  Y(a_{k-1}, a_k), \allowbreak
  Y(a_k, a_{k+1}), \\
  & \quad Z(a_{k+1}, a_{k+2}), \ldots, \allowbreak
  Z(a_{2k-1}, a_{2k}), \allowbreak
  Z(a_{2k}, a_{2k+1}), \allowbreak
\mathrm{End}(a_{2k+1})\}\end{align*} for
some $k \in \mathbb{N}$. Indeed, it is clear that $I_k$ satisfies $P$ for all $k
\in \mathbb{N}$, as we derive the facts \[S(a_k, a_{k+2}), S(a_{k-1}, a_{k+3}), \ldots,
S(a_{k-(k-1)}, a_{k+2+(k-1)})\text{, that is, }S(a_1, a_{2k+1}),\] and finally
$\mathrm{Goal}()$. Conversely, for any instance $I$ of the prescribed shape that
satisfies~$P$, it is easily seen that the derivation of $\mathrm{Goal}$
justifies the existence of a chain in~$I$ of the form~$I_k$, which by the
restrictions on the shape of $I$ means that $I = I_k$.

We further restrict our attention to tree encodings
that consist of a single branch of a specific form, namely, their contents are as
  follows (given from leaf to root) for some integer $n\geq 0$:
$(\{a_1\}, \mathrm{Begin}(a_1))$, $(\{a_1, a_2\}, X(a_1, a_2))$, $(\{a_2, a_3\},
X(a_2, a_3))$, $(\{a_3, a_1\}, X(a_3, a_1))$, 
\dots, $(\{a_{n}, a_{n+1}\}, X(a_n, a_{n+1}))$, $(\{a_{n+1}\},
\mathrm{End}(a_{n+1}))$, 
  where we write $X$ to mean that we may match either $Y$ or $Z$,
  where addition is modulo $3$, and where we add dummy nodes $(\bot,
\bot)$ as left children of all nodes, and as right children of the leaf node
  $(\{a_1\},\mathrm{Begin}(a_1))$,
to ensure that the tree is full.
It is clear that we can design a bNTA $A_{\mathrm{encode}}$ which recognizes
  tree encodings of this form, and we define $A''$ to be the intersection of
  $A'$ and $A_{\mathrm{encode}}$. In other words,
  $A''$ further enforces that the $\Gamma^1_\sigma$-tree
encodes the input instance as a chain of consecutive facts with a certain
prescribed alternation pattern for elements, with the $\mathrm{Begin}$ end of
the chain at the top and the $\mathrm{End}$ end at the bottom.

Now, it is easily seen that there is exactly one tree encoding of every $I_k$ which is
accepted by~$A''$, namely, the one of the form tested by $A_{\mathrm{encode}}$
where $n=2k$, the first $k$ $X$ are matched to $Y$ and the last $k$ $X$
are matched to~$Z$. 

Now, we observe that as $A''$ is a bNTA which is forced to operate on
chains (completed to full binary trees by a specific
addition of binary nodes). Thus, we can translate it to a deterministic automaton
$A'''$ on words on the alphabet $\Sigma = \{B, Y, Z, E\}$,
by looking at
its behavior in terms of the $X$-facts. Formally, $A'''$ has same state
space as $A''$, same final states, initial state
$\delta(\iota((\bot, \bot)), \iota((\bot, \bot)))$ and
transition function $\delta(q, x) = \delta(\iota((\bot, \bot)), q,
(s, f))$ for every domain $s$, where $f$ is a fact corresponding to the letter $x
\in \Sigma$ ($B$ stands here for $\mathrm{Begin}$, and $E$ for
$\mathrm{End}$).
By definition of $A''$, the automaton $A'''$ on words recognizes the language $\{BY^kZ^kE \mid k \in
\mathbb{N}\}$. 
However, this language is not regular.
This contradicts our hypothesis about the existence of automaton $A$, which
establishes the desired result.
\end{proof}

Hence, there is no bNTA or alternating two-way tree automaton that tests $P$ for
  treewidth~$1$.
To work around this problem and ensure that translation is possible and efficient,
the key condition that we impose on Datalog programs, pursuant to the
intuition of simplicial decompositions, is that
the rules must be \emph{clique-frontier-guarded}, i.e., the
variables in the head must co-occur in \emph{positive} predicates of the rule body. 
We can then use the \emph{body size} of the program rules as a parameter,
and will show that the fragment can then be translated to automata in FPT-linear
time. Remember that we assume that the arity of the extensional signature is fixed.

\begin{definition}
  \label{def:CFG}
  Let $P$ be a stratified Datalog program.
  A rule $r$ of $P$ is \emph{clique-frontier-guarded} if for any two variables $x_i \neq x_j$ in the head of $r$, we have that $x_i$ and
  $x_j$ co-occur in some positive (extensional or intensional) predicate of the body of $r$.
  $P$ is \emph{clique-frontier-guarded} (CFG) if all its rules are clique-frontier-guarded.
  The \emph{body size} of~$P$ is the maximal number of atoms in
  the body of its rules, multiplied by its arity.
\end{definition}

\begin{example}
 Let $P$ be the stratified Datalog program from Example~\ref{expl:datalog}. 
	We recall that $\sigma = \{R\}$, and that $P$ tests if there are two elements that are not
	connected by a directed $R$-path. 
	Then $P$ is not a CFG-Datalog program, since the rule $T(x,y) \leftarrow R(x,z) \land T(z,y)$ is not 
	clique-frontier-guarded.
	In fact, it is easy to show that no CFG-Datalog program can contain a binary intensional predicate $T$ computing the transitive closure of an extensional binary relation $R$.

  However, it is possible to express a similar query in
        CFG-Datalog. Let $\sigma$ be $\{R,A,B\}$, with $R$ being binary, $A$ and $B$ being unary.
	Let $\sigmai = \{T\}$, with $T$ unary. 
	Consider the stratified Datalog program $P'$ with two strata $P'_1$ and $P'_2$. 
	The stratum $P'_1$ contains the following two rules:
	\begin{itemize}
		\item $T(x) \leftarrow A(x)$
		\item $T(y) \leftarrow T(x) \land R(x,y)$
	\end{itemize}
	And $P'_2$ contains the rule:
	\begin{itemize}
		\item $\mathrm{Goal}() \leftarrow A(x) \land B(y) \land \lnot T(y)$
	\end{itemize}
	Then $P'$ is a CFG-Datalog program (of body size $3\times 2=6$).
	Moreover, $P'$ tests if there exist two distinct elements $a \neq b$ such
	that $A(a)$ and $B(b)$ hold, and such that $a$ and $b$ are not connected by a directed $R$-path.
\end{example}

We will see later in this section what interesting query languages CFG-Datalog capture: (Boolean) CQs, (Boolean) 2RPQs, (Boolean) SAC2RPQs, guarded negation logics, monadic Datalog, etc.

The main result of this paper is that evaluation of CFG-Datalog is
\emph{FPT-bilinear} in \emph{combined complexity}, when parameterized by the body
size of the program and the instance treewidth.

\begin{theorem}
  \label{thm:main}
  Given a CFG-Datalog program~$P$ of body size
  $\kp$ and a relational instance~$I$ of treewidth $\ki$, checking
  if $I\models P$ is FPT-bilinear time in $|I|\cdot|P|$ 
  (parameterized by~$\kp$ and~$\ki$).
\end{theorem}

We will show this result in the next section by translating CFG-Datalog programs in FPT-linear time to
a special kind of tree automata (Theorem~\ref{thm:maintheorem}), and showing
in Section~\ref{sec:provenance} that we can efficiently evaluate such automata
and even compute \emph{provenance representations}.
The rest of this section
presents consequences of our main result for various languages.

\subparagraph*{Conjunctive queries.}
Our tractability result for
bounded-simplicial-width CQs (Theorem~\ref{thm:simplicial}), including $\alpha$-acyclic
CQs, is shown by rewriting
to CFG-Datalog of bounded body size:

\begin{proposition}\label{prp:simplicial}
  There is a function $f_\sigma$ (depending only on~$\sigma$)
  such that for all $k \in \mathbb{N}$,
  for any conjunctive query $Q$ and simplicial tree decomposition $T$ 
  of $Q$ of
  width at most~$k$, we can compute in $O(\card{Q} + \card{T})$ an equivalent
  CFG-Datalog program with body size at most $f_{\sigma}(k)$.
\end{proposition}

To prove Proposition~\ref{prp:simplicial}, we first prove the following lemma about simplicial tree decompositions:

  \begin{lemma}
    \label{lem:rewritesimplicial}
    For any simplicial decomposition $T$ of width~$k$ of a query $Q$, we can compute in
    linear time a simplicial decomposition $T_{\mathrm{bounded}}$ of $Q$ such
    that each bag has degree at most~$2^{k+1}$.
  \end{lemma}

  \begin{proof}
	Fix $Q$ and $T$.
        We construct the simplicial decomposition $T_{\mathrm{bounded}}$ of $Q$
        in a process which shares some similarity with the routine rewriting of
        tree decompositions to make them binary, by creating copies of bags.
        However, the process is more intricate because we need to preserve the fact that
        we have a \emph{simplicial} tree decomposition, where interfaces are
        guarded.
        
        We go over $T$ bottom-up: for each bag $b$ of $T$, we create a bag
        $b'$ of $T_{\mathrm{bounded}}$ with same domain as $b$.
        Now, we partition the children of $b$ depending on their intersection
        with $b$: for every subset $S$ of the domain of~$b$ such that $b$
        has some children whose intersection with $b$ is equal to $S$, we
        write these children
        $b_{S,1},\ldots,b_{S,n_S}$ (so we have $S = \dom(b) \cap
        \dom(b_{S,j})$ for all $1 \leq j \leq n_S$), and we write $b'_{S,1},
        \ldots, b'_{S, n_S}$ for the copies that we already created for these bags in
        $T_{\mathrm{bounded}}$. Now, for each~$S$, we create
        $n_S$ fresh bags $b'_{=S,j}$ in $T_{\mathrm{bounded}}$ (for $1 \leq
        j \leq n_S$) with domain equal to~$S$,
        and we set $b'_{=S,1}$ to be a child of~$b'$,
        $b'_{=S,j+1}$ to be a child of $b'_{=S,j}$ for all $1 \leq j < n_S$,
        and we set each $b'_{S,i}$ to be a child of $b'_{=S,i}$.

        This process can clearly be performed in linear time. Now, the degree of
        the fresh bags in $T_{\mathrm{bounded}}$ is at most~$2$, and the degree
        of the copies of the original bags is at most~$2^{k+1}$, as stated.
        Further, it is clear that the result is still a tree decomposition (each
        fact is still covered, the occurrences of each element still form a
        connected subtree because they are as in~$T$ with the addition of some
        paths of the fresh bags), and the interfaces in~$T_{\mathrm{bounded}}$
        are the same as in~$T$, so they still satisfy the requirement of
        simplicial decompositions.
  \end{proof}

    We can now prove Proposition~\ref{prp:simplicial}. In fact, as will be easy
    to notice from the proof, our construction further ensures that the
    equivalent CFG-Datalog program is positive, nonrecursive, and
    conjunctive.
Recall that
  a Datalog program is \emph{positive} if it contains no negated
  atoms. It is \emph{nonrecursive} if there is no cycle in the directed graph on
  $\sigmai$ having an edge from $R$ to $S$ whenever a rule contains $R$ in its
  head and $S$ in its body.
  It is \emph{conjunctive}~\cite{benedikt2010impact} if each intensional relation
  $R$ occurs in the head of at most one rule.

  \begin{proof}[Proof of Proposition~\ref{prp:simplicial}]
        Using Lemma~\ref{lem:rewritesimplicial}, we can start by rewriting in
        linear time the
        input simplicial decomposition to ensure that each bag has degree at
        most $2^{k+1}$. Hence, let us assume without loss of generality that $T$
        has this property. We further add an empty root bag if necessary to
        ensure that the root bag of~$T$ is empty and has exactly one child.

	We start by using Lemma~3.1 of \cite{flum2002query} to annotate in linear time each node $b$ of $T$ by
	the set of atoms $\calA_b$ of $Q$ whose free variables are in the domain of $b$ and such that 
	for each atom $A$ of $\calA_b$, $b$ is the topmost bag of $T$ which contains all the variables of~$A$.
        As the signature $\sigma$ is fixed, note that we have
        $\card{\calA_b} \leq g_\sigma(k)$ for some function~$g_\sigma$ depending only on~$\sigma$.

	We now perform a process similar to Lemma~3.1 of \cite{flum2002query}. 
	We start by precomputing in linear time a mapping $\mu$ that associates, to each
        pair $\{x, y\}$ of variables of~$Q$, the set of all atoms in~$Q$ where
        $\{x, y\}$ co-occur. We can compute $\mu$ in linear time by
        processing all atoms of~$Q$ and adding each atom as an image of~$\mu$ for
        each pair of variables that it contains (remember that the arity of~$\sigma$ is
        constant). Now, we do the following computation:
        for each bag $b$ which is not the root of~$T$, letting
        $S$ be its interface with its parent bag, we annotate $b$ by a set of
        atoms $\calA^\guard_b$ defined as follows:
        for all $x,y \in S$ with $x \neq y$, letting $A(\mathbf{z})$ be an atom
        of~$Q$ where $x$ and $y$ appear (which must exist, by the requirement on
        simplicial decompositions, and which we retrieve from~$\mu$), we
        add $A(\mathbf{w})$ to $\calA^\guard_b$, where, for $1 \leq i \leq
        \card{\mathbf{z}}$, we set $w_i \defeq z_i$ if $z_i \in \{x, y\}$, and
        $w_i$ to be a fresh variable otherwise. In other words,
        $\calA^\guard_b$ is a set of atoms that ensures that the interface $S$
        of $b$ with its parent is covered by a clique, and we construct it by
        picking atoms
        of~$Q$ that witness the fact that it is guarded (which it is, because
        $T$ is a simplicial decomposition), and replacing their irrelevant
        variables to be fresh. Note that $\calA^\guard_b$
        consists of at most $k\times(k+1)/2$ atoms, but the domain of these
        atoms is not a subset of $\dom(b)$ (because they include fresh
        variables). This entire computation is performed in linear time.

        We now define the function $f_\sigma(k)$ as follows, remembering that
        $\arity{\sigma}$ denotes the arity of the \emph{extensional} signature:
        \[
          f_\sigma(k) \defeq (k+1) \times \left( g_\sigma(k)
          + 2^{k+1} + k(k+1)/2 \right).
        \]

	We now build our CFG-Datalog program $P$ of body size $f_\sigma(k)$
        which is equivalent to $Q$. We define the intensional signature
        $\sigmai$ by creating one intensional predicate
        $P_b$ for each non-root bag $b$ of~$T$,
        whose arity is the size of the intersection of $b$ with its parent.
        As we
        ensured that the root bag $b_\r$ of~$T$ is empty and has exactly one
        child $b_\r'$, we use $P_{b_\r'}$ as our
        0-ary $\goal$ predicate (because its interface with its parent $b_\r$
        is necessarily empty).
        We now define
        the rules of~$P$ by
        processing $T$ bottom-up: for each bag $b$ of~$T$, we add one rule
        $\rho_b$ with head $P_b(\mathbf{x})$, defined as follows:

	\begin{itemize}
		\item If $b$ is a leaf, then
                  $\rho_b$ is $P_b \leftarrow \bigwedge\calA^{\text{guard}}_{b} \land \bigwedge \calA_b$.
		\item If $b$ is an internal node with children $b_1,\ldots,b_m$
                  (remember that $m \leq 2^{k+1}$),
			then $\rho_b$ is $P_b \leftarrow
			\bigwedge\calA^{\text{guard}}_{b} \land \bigwedge \calA_b \land
                        \bigwedge_{1 \leq i \leq m} P_{b_i}$.
	\end{itemize}
        
        We first check that $P$ is clique-frontier-guarded, but this is the
        case because by construction the conjunction of atoms
	$\bigwedge\calA^{\guard}_{b}$ is a suitable guard for $\mathbf{x}$: for each
	$\{x, y\} \in \mathbf{x}$, it contains an atom where both $x$ and $y$ occur.

        Second, we check that the body size of $P$ is indeed
        $f_\sigma(k)$. It is clear that $\arity{P} = \arity{\sigmai \cup \sigma} \leq k+1$.
	Further, 
        the maximal number
	of atoms in the body of a rule is $g_\sigma(k) + 2^{k+1} + k(k+1)/2$,
        so we obtain the desired bound.

	What is left to check is that $P$ is equivalent to $Q$. It will be
        helpful to reason about $P$ by seeing it as the conjunctive query $Q'$
        obtained by recursively inlining the definition of rules: observe that
        this a conjunctive query, because $P$ is conjunctive, i.e., for each
        intensional atom $P_b$, the rule $\rho_b$ is the only one where
        $P_b$ occurs as head atom. It is clear that $P$ and~$Q'$ are
        equivalent, so we must prove that $Q$ and $Q'$ are equivalent.

        For the forward direction, it is obvious that $Q'$ implies $Q$, because
        $Q'$ contains every atom of~$Q$ by construction of the $\calA_b$. For
        the backward direction, noting that the only atoms of $Q'$ that are not
        in $Q$ are those added in the sets $\calA^\guard_b$, we observe that
        there is a homomorphism from $Q'$ to $Q$ defined by mapping each atom
        $A(\mathbf{w})$ occurring in some $\calA^\guard_b$ to the atom
        $A(\mathbf{z})$ of~$Q$ used to create it;
        this mapping is the identity on the two variables $x$ and $y$ used to
        create $A(\mathbf{w})$, and maps each fresh variables $w_i$ to
        $z_i$: the fact that these variables are fresh ensures that this
        homomorphism is well-defined. This shows $Q$ and $Q'$, hence $P$, to be
        equivalent, which concludes the proof.
\end{proof}

This implies that CFG-Datalog can express any CQ up to
increasing the body size parameter (since any CQ has a simplicial decomposition), unlike, e.g., $\mu$CGF~\cite{berwanger2001games}.
Conversely, we can show that bounded-simplicial-width CQs
\emph{characterize}
the queries expressible in CFG-Datalog when disallowing negation,
recursion,
and disjunction.  

\begin{proposition}
	\label{prp:onlysimplicial}
  For any positive, conjunctive, nonrecursive CFG-Datalog program~$P$
  with body size $k$,
  there is a CQ $Q$ of
  simplicial width $\leq k$ that is equivalent to~$P$.
\end{proposition}

  To prove Proposition~\ref{prp:onlysimplicial}, we will use the notion of
  \emph{call graph} of a Datalog program. This is the graph $G$ on the relations of
  $\sigmai$ which has an edge from $R$ to $S$ whenever a rule contains relation
  $R$ in its head and $S$ in its body. From the requirement that $P$ is
  nonrecursive,
  we know that this graph $G$ is a DAG.

  \begin{proof}[Proof of Proposition~\ref{prp:onlysimplicial}]
  We first check that every intensional relation reachable from $\text{Goal}$ in
  the call graph~$G$ of~$P$ appears in the head of a rule of $P$ (as $P$ is
    conjunctive, this
  rule is then unique). Otherwise, it is clear that $P$ is not satisfiable (it
  has no derivation tree), so we can simply rewrite $P$ to the query False.
  We also assume without loss of generality that each intensional relation
  except $\goal$ occurs in the body of some rule, as otherwise we can simply
  drop them and drop all rules where they appear as the head relation.

  In the rest of the proof we will consider the rules of~$P$ in some order, and
    create an equivalent CFG-Datalog program $P'$ 
        with rules
  $r'_0, \ldots, r'_{m}$. We will ensure that $P'$ is also positive,
    conjunctive, and
  nonrecursive, and that it further satisfies the following additional properties:
  \begin{enumerate}
	  \item Every intensional relation other than $\text{Goal}$ appears in
            the body of exactly one rule of~$P'$, and appears there exactly
            once;
	  \item For every $0 \leq i \leq m$, for every variable $z$ in the body
            of rule $r'_i$ that does not occur in its head,
		  then for every $0 \leq j < i$, $z$ does not occur in $r'_j$.
  \end{enumerate}

    We initialize a queue that contains only the one rule that defines
    $\text{Goal}$ in~$P$, and we do the following until the queue is empty:
  \begin{itemize}
	  \item Pop a rule $r$ from the queue.
		  Let $r'$ be defined from $r$ as follows:
                  for every intensional relation $R$ that occurs in the body of
                  $r$, letting $R(\mathbf{x^1}), \ldots, R(\mathbf{x^n})$ be
                  its occurrences, rewrite these atoms to $R^1_r(\mathbf{x^1}),
                  \ldots, R^n_r(\mathbf{x^n})$, where the $R^i_r$ are \emph{fresh}
                  intensional relations.
	  \item Add $r'$ to $P'$.
	  \item For each intensional atom $R^i_r(\mathbf{x})$ of $r'$, letting $R$
            be the relation from which $R^i_r$ was created, 
            let $r_R$ be the rule of $P$ that has $R$ in its head (by our
            initial considerations, there is one such rule, and as the program
            is
	    conjunctive there is exactly one such rule). Define
            $r'_{R^i}$ from~$r_R$ by changing its head relation to be $R^i_r$
            instead of~$R$, and by renaming
            its head and body variables such that the head is exactly
            $R^i_r(\mathbf{x})$. Further rename all variables that occur in the
            body but not in the head, to replace them by fresh new variables.
            Add $r'_{R^i}$ to the queue.
  \end{itemize}

  We first argue that this process terminates. Indeed, considering the graph
  $G$, whenever we pop from the queue a rule with head relation $R$ (or a fresh
  relation created from a relation~$R$), we add to the queue a finite number of
  rules for head relations created from relations $R'$ such that the edge $(R,
  R')$ is in the graph $G$. The fact that $G$ is acyclic ensures that the
  process terminates (but note that its running time may generally be
  exponential in the input).
  Second, we observe that, by construction, $P$ satisfies the first property,
  because each occurrence of an intensional relation in a body of $P'$ is fresh,
  and satisfies the second property, because each variable which is in the body of
  a rule but not in its head is fresh, so it cannot occur in a previous rule

  Last, we verify that $P$ and $P'$ are equivalent, but this is immediate,
  because any derivation tree for $P$ can be rewritten to a derivation tree for $P'$ (by
  renaming relations and variables), and vice-versa.

  We define $Q$ to be the conjunction of all extensional atoms occurring
  in~$P'$. To show that it is equivalent to~$P'$, the fact that $Q$ implies $P'$
  is immediate as the leaves are sufficient to construct a derivation tree, and the
  fact that $P'$ implies $Q$ is because, 
  letting $G'$ be the call graph of~$P'$, by the first property of~$P'$ we can easily observe that it is a
  tree, so the structure of derivation trees of $G'$ also corresponds to $P$, and by
  the second property of~$P'$ we know that two variables are equal in two extensional
  atoms iff they have to be equal in any derivation tree. Hence, $P'$ and $Q$ are
  indeed equivalent.

 We now justify that $Q$ has simplicial width at most $k$.
 We do so by building from~$P'$ a simplicial decomposition $T$ of $Q$ of width $\leq k$. 
 The structure of $T$ is the same as $G'$ (which is actually a tree). For each
 bag $b$ of $T$ corresponding to a node of $G'$ standing for a rule $r$ of $P'$, we
 set the domain of $b$ to be the variables occurring in~$r$. It is clear that
 $T$ is a tree decomposition of~$Q$, because each atom of~$Q$ is covered by a
 bag of~$T$ (namely, the one for the rule whose body contained that atom) and the
 occurrences of each variable form a connected subtree (whose root is the node
 of $G'$ standing for the rule where it was introduced, using the second
 condition of~$P'$). Further, $T$ is a simplicial decomposition because 
  $P'$ is clique-frontier-guarded; further, from the second condition, the
  variables shared between one bag and its child are precisely the head
  variables of the child rule. The width is $\leq k$ because the
  body size of a CFG-Datalog program is an upper bound on the maximal number of
  variables in a rule body.
\end{proof}

However, our CFG-Datalog fragment is still exponentially more \emph{concise}
than such CQs:

\begin{proposition}\label{prp:simplicialconcise}
  There is a signature $\sigma$ and a family $(P_n)_{n\in\NN}$ of CFG-Datalog
  programs with body size at most $6$ 
  which are positive, conjunctive, and nonrecursive, such that
  $|P_n|=O(n)$
  and any conjunctive query $Q_n$ equivalent to $P_n$ has size $\Omega(2^n)$.
\end{proposition}

  To prove Proposition~\ref{prp:simplicialconcise}, we recall the
  following classical notion:

  \begin{definition}
    \label{def:match}
  A \emph{match} of a conjunctive query $Q$
  in an instance $I$
  is a subinstance $M$ of~$I$ which is an image of a homomorphism from
  the canonical instance of $Q$ to~$I$, i.e., $M$ witnesses that $I \models Q$,
  in particular $M \models Q$.
  \end{definition}

  Our proof will rely on the following elementary observation:
  \begin{lemma}
    \label{lem:match}
  If a CQ $Q$ has a match $M$
  in an instance $I$, then necessarily $\card{Q} \geq \card{M}$.
  \end{lemma}
  
  \begin{proof}
    As $M$ is
  the image of $Q$ by a homomorphism, it cannot have more facts than~$Q$
    has atoms.
  \end{proof}

  We are now ready to prove Proposition~\ref{prp:simplicialconcise}:

  \begin{proof}[Proof of Proposition~\ref{prp:simplicialconcise}]
  Fix $\sigma$ to contain a binary relation $R$ and a binary relation $G$.
  Consider the rule $\rho_0: R_0(x, y) \leftarrow R(x, y)$ and define
  the following rules, for all $i > 0$:
  \[
    \rho_i: R_i(x, y) \leftarrow G(x, y), R_{i-1}(x, z), R_{i-1}(z, y)
  \]
  For each $i > 0$, we let $P_i$ consist of the rules $\rho_j$ for $0 \leq j
  \leq i$, as well as the rule $\text{Goal}() \leftarrow R_i(x, y)$.
  It is clear that each $P_i$ is positive, conjunctive, and
  nonrecursive; further, the predicate $G$ ensures that it is a CFG-Datalog
  program. The arity is~$2$ and the maximum number of atoms is the body is~$3$,
  so the body size is indeed~$6$.

  We first prove by an immediate induction that, for each $i \geq 0$, considering
  the rules of $P_i$ and the intensional predicate $R_i$, whenever an instance $I$
  satisfies $R_i(a, b)$ for two elements $a, b \in \dom(I)$ then there is an
  $R$-path of length $2^{i}$ from~$a$ to~$b$.
  Now, fixing $i \geq 0$, this
  clearly implies there is
  an instance $I_i$
  of size (number of
  facts) $\geq 2^{i}$, namely, an $R$-path of this length
  with the right set of additional $G$-facts,
  such that $I_i \models P_i$ but any strict subset of~$I_i$ does not satisfy
  $P_i$.

  Now, let us consider a CQ $Q_i$ which is equivalent to $P_i$, and let us show
  the desired size bound. By equivalence, we know that $I_i \models Q_i$, hence
  $Q_i$ has a match $M_i$ in $I_i$, but any strict subset of~$I_i$ does not
  satisfy $Q_i$, which implies that, necessarily, $M_i = I_i$ (indeed,
  otherwise $M_i$ would survive as a match in some strict subset of~$I_i$).
  Now, by Lemma~\ref{lem:match}, we deduce that $\card{Q_i} \geq \card{M_i}$,
  and as $\card{M_i} = \card{I_i} \geq 2^{i}$, we obtain the desired size
  bound, which concludes the proof.
\end{proof}

\subparagraph*{Guarded negation fragments.}
Having explained the connections between CFG-Datalog and CQs, we now study its
connections to the more expressive languages of guarded logics, specifically,
the \emph{guarded negation fragment} (GNF), a fragment of first-order logic
\cite{barany2015guarded}. Indeed, when putting GNF formulas in 
\emph{GN-normal form} \cite{barany2015guarded} or even \emph{weak GN-normal form}
\cite{benedikt2014effective}, we can translate them to
CFG-Datalog, and we can use the \emph{CQ-rank}
parameter \cite{benedikt2014effective} (that measures the maximal number of
atoms in conjunctions) to control the body size parameter.
We first recall from \cite{benedikt2014effective}, Appendix~B.1, the definitions of a weak GN-normal form formulas and of CQ-rank:

\begin{definition}
A formula is in weak GN-normal form if it is a $\phi$-formula in the inductive definition below:
  
  \begin{itemize}
    \item A disjunction of existentially quantified conjunctions of
      $\psi$-formulas is a $\phi$-formula;
    \item An atom is a $\psi$-formula;
    \item The conjunction of a $\phi$-formula and of a guard is a
      $\psi$-formula;
    \item The conjunction of the negation of a $\phi$-formula and of a guard is
      a $\psi$-formula.
  \end{itemize}
The \emph{CQ-rank} of a $\phi$-formula is the overall
  number of conjuncts occurring in the disjunction of existentially quantified
  conjunctions that defines this subformula.
\end{definition}

We can then show:

\begin{proposition}\label{prp:gnf-to-icg2}
  There is a function $f_\sigma$ (depending only on~$\sigma$) such that, for any
  weak GN-normal form GNF query $Q$ of CQ-rank $r$, we can compute in time $O(\card{Q})$ an
  equivalent nonrecursive CFG-Datalog program $P$ of body size $f_\sigma(r)$.
\end{proposition}

\begin{proof}
  We define $f_\sigma : n \mapsto \arity{\sigma} \times n$.

  We consider an input Boolean GN-normal form formula $Q$ of CQ-rank~$r$, and call $T$ its
  abstract syntax tree. We rewrite $T$ in linear time to inline in
  $\phi$-formulas the definition of their $\psi$-formulas, so all nodes of~$T$
  consist of $\phi$-formulas, in which all subformulas are guarded (but they can
  be used positively or negatively).
 
  We now process $T$ bottom-up. We introduce one intensional Datalog predicate~$R_n$ per node $n$ in~$T$:
  its arity is the number of variables that are free at~$n$.
  We then introduce one rule $\rho_{n,\delta}$ for each disjunct $\delta$ of the disjunction that defines
  $n$ in~$T$: the head of~$\rho_{n, \delta}$ is an $R_n$-atom whose free variables are the variables that are
  free in~$n$, and the
  body of $\rho_{n, \delta}$ is the conjunction that defines~$\delta$, with each subformula replaced by the
  intensional relation that codes it. Of course, we use the predicate $R_r$ for
  the root $r$ of~$T$ as our goal predicate; note that it must be $0$-ary, as
  $Q$ is Boolean so there are no free variables at the root of~$T$. This process defines
  our CFG-Datalog program~$P$: it is clear that this process runs in linear
  time.
 
  We first observe that body size for an intensional predicate $R_n$
  is less than the CQ-rank of the corresponding subformula.
	Hence, as the arity of~$\sigma$ is
  bounded, clearly~$P$ has body size $\leq
  f_\sigma(r)$. We next observe that intentional predicates in the bodies of
  rules of~$P$ are
  always guarded, thanks to the guardedness requirement on~$Q$.
  Further, it is obvious that $P$ is nonrecursive, as it is
  computed from the abstract syntax tree~$T$. Last, it is clear that $P$ is equivalent to the
  original formula $Q$, as we can obtain $Q$ back simply by inlining the
  definition of the intensional predicates.
\end{proof}

In fact, the efficient translation of bounded-CQ-rank normal-form GNF programs
(using the fact that subformulas are ``answer-guarded'', like our guardedness
requirements)
has been used recently (e.g., in~\cite{benedikt2016step}),
to give efficient procedures for GNF \emph{satisfiability}.
The \emph{satisfiability} problem for a logic formally asks, given a sentence in
this logic, whether it is satisfiable (i.e., there is an instance that satisfies
it), and two variants of the problem exist: \emph{finite satisfiability}, where
we ask for the existence of a \emph{finite} instance (as we defined them in this work),
and \emph{unrestricted satisfiability}, where we also allow the satisfying
instance to be infinite.
The decidability of both finite and unrestricted satisfiability for GNF is shown by translating GNF to
automata (for a treewidth which is not fixed, unlike in our context, but
depends on the formula). CFG-Datalog further
allows clique guards (similar to CGNFO~\cite{barany2015guarded}), can reuse
subformulas (similar to the idea of DAG-representations in
\cite{benedikt2014effective}), and supports recursion (similar to
GNFP~\cite{barany2015guarded}, or GN-Datalog~\cite{barany2012queries} but
whose combined complexity is intractable --- P$^\text{NP}$-complete). 
CFG-Datalog also resembles $\mu$CGF~\cite{berwanger2001games}, but recall that
$\mu$CGF
is not a \emph{guarded negation} logic, so, e.g., $\mu$CGF cannot express all CQs, unlike CFG-Datalog or GNF.

Hence, the design of CFG-Datalog, and its translation to automata, has
similarities with guarded logics. However, to our knowledge, the idea of applying
it to query evaluation is new, and CFG-Datalog is designed to support all
relevant features to capture interesting query languages (e.g.,
clique guards are necessary to capture bounded-simplicial-width queries).
Moreover CFG-Datalog is intrinsically more expressive than guarded
negation logics as its satisfiability is undecidable, in contrast with
GNF~\cite{barany2015guarded}, CGNFO~\cite{barany2015guarded},
GNFP~\cite{barany2015guarded}, GN-Datalog~\cite{barany2012queries},
$\mu$CGF~\cite{Gradel02}, the satisfiability of all of which is decidable.

\begin{proposition}
  \label{prop:undecidable}
  Given a signature~$\sigma$ and a CFG-Datalog~$P$ over~$\sigma$,
  determining if~$P$ is satisfiable
  is undecidable, in both the finite and
  unrestricted cases.
\end{proposition}

\begin{proof}
  \newcommand\Eq{\mathord{\mathrm{Eq}}}

We reduce from the implication problem for
functional dependencies and inclusion dependencies, a problem
known to be
undecidable~\cite{mitchell1983implication,chandra1985implication} over
both finite and unrestricted instances.
See also~\cite{abiteboul1995foundations} for a general presentation of the
problem and formal definitions and notation for functional dependencies and inclusion
dependencies.

Let $\sigma$ be a relational signature,
  let $d$ be a functional dependency or an inclusion dependency over~$\sigma$, and let $\Delta$ be a set of
functional  dependencies and inclusion dependencies over~$\sigma$.
The problem is to determine if $\Delta$ implies $d$.

We construct a CFG-Datalog program~$P$ over $\sigma$ which is satisfiable over finite
(resp., unrestricted) instances iff $\Delta$ implies $d$ over finite
(resp., unrestricted) instances, which establishes that
CFG-Datalog satisfiability is undecidable.

The intensional signature of the program $P$ is made of:
\begin{itemize}
  \item a binary
relation $\Eq$;
\item a nullary relation $P_{\neg \delta}$ for every dependency
  $\delta \in \Delta \cup \{d\}$;
\item a relation $P_{\Pi_Z(S)}$ whose arity is $|Z|$ whenever there is at least one inclusion
  dependency $R[Y]\subseteq S[Z] \in \Delta \cup \{d\}$;
  \item the nullary relation $\mathrm{Goal}$.
  \end{itemize}

  For every extensional relation~$R$ and for every $1\leq i\leq \arity{R}$,
  we add rules of the form:
  \begin{align*}\Eq(x_i,x_i)\leftarrow
  R(\mathbf x).\end{align*}
Consequently, for
  every instance~$I$ over~$\sigma$, 
  the $\Eq$-facts of~$P(I)$ will be exactly $\{\Eq(v,v) \mid v \in \dom(I)\}$.

  For every functional dependency $\delta$ in $\Delta\cup\{d\}$ with
  $\delta = R[Y] \rightarrow R[Z]$, we add the following rules, for $1 \leq j
  \leq \card{Z}$:
  \begin{align*}
    P_{\neg \delta}() \leftarrow
    R(\mathbf x), R(\mathbf x'), {}&\Eq(y_{1}
    ,y'_{1} ), \dots, \Eq(y_{|Y|},y'_{|Y|}) ,
    \neg \Eq(z_j,z'_j)
\end{align*}
where for each $1\leq i \leq |Y|$, the variables $y_i$~and~$y'_i$ are those 
at the $Y_i$-th position in $R(\mathbf x)$
and $R(\mathbf x')$, respectively; and where 
the variables $z_j$~and~$z'_j$ are those
at the $Z_j$-th position in $R(\mathbf x)$
and $R(\mathbf x')$, respectively.

  For every inclusion dependency $\delta \in \Delta \cup \{d\}$, with $\delta  =
R[Y] \subseteq S[Z]$ we add two rules:
\[
  P_{\Pi_Z(S)}(\mathbf{z}) \leftarrow S(\mathbf{x})\qquad P_{\neg
\delta}() \leftarrow R(\mathbf{x}), \neg P_{\Pi_Z(S)}(\mathbf{y})\] where
$\mathbf z$ are the variables at positions $Z$ within  $S(\mathbf{x})$ and 
$\mathbf y$ are the variables at positions $Y$ within $R(\mathbf{x})$.

Finally, we add one rule for the goal predicate:
\[\mathrm{Goal}() \leftarrow  P_{\neg d}()
,\neg P_{\neg \delta_1}(), \cdots, \neg P_{\neg \delta_k}()\]
where $\Delta=\{\,\delta_1,\dots, \delta_k\,\}$.

Note that all the rules that we have written are clearly in CFG-Datalog.
Now, let $I$ be some instance.
It is clear that for each functional dependency $\delta$,
$P_{\neg \delta}()$ is in $P(I)$ iff $I$ does not satisfy $\delta$.
Similarly,
for each inclusion dependency $\delta$, $P_{\neg
\delta}()$ is in $P(I)$ iff $I$ does not satisfy $\delta$. Therefore, for
each instance $I$, $\mathrm{Goal}()$ is in $P(I)$ iff $I$ satisfies~$\Delta$ and
$I$ does not satisfy $d$. Thus $P$ is satisfiable over finite instances (resp.,
unrestricted instances) iff there exists a finite
instance (resp., a finite or infinite instance) that satisfies $\Delta$ and does not satisfy $d$, i.e., iff $\Delta$
does imply~$d$ over finite instances (resp., over unrestricted instances).
\end{proof}

We point out that the extensional signature is not fixed in this proof,
unlike in the rest of the article. This is simply to establish the
expressiveness of CFG-Datalog, it has no impact on our study of the
combined complexity of query evaluation. 

\subparagraph*{Recursive languages.}
The use of fixpoints in CFG-Datalog, in particular, allows us
to capture the combined tractability of interesting recursive languages. First,
observe that our
guardedness requirement becomes trivial when all intensional predicates are
monadic (arity-one), so our main result implies that 
\emph{monadic Datalog} of bounded body size is tractable in combined
complexity on treelike instances. This is reminiscent of the results of
\cite{gottlob2010monadic}. We show:

\begin{proposition}\label{prp:monadicdl}
  The combined complexity of monadic Datalog query evaluation on
  bounded-treewidth instances is FPT when parameterized by instance treewidth
  and body size (as in Definition~\ref{def:CFG}) of the monadic Datalog program.
\end{proposition}

\begin{proof}
  This is simply by observing that any monadic Datalog program is a CFG-Datalog
  program with the same body size, so we can simply apply
  Theorem~\ref{thm:main}.
\end{proof}

Second, CFG-Datalog can capture \emph{two-way regular path 
queries} (2RPQs), and even \emph{strongly acyclic conjunctions of 2RPQs} (SAC2RPQs) \cite{calvanese2000containment,barcelo2013querying}, a 
well-known query language in the context of  
graph databases and knowledge bases:

\begin{definition}
  We assume that the signature $\sigma$ contains only binary relations.
  A (non-Boolean) \emph{regular path query} (RPQ) $q_L(x,y)$ is defined by a regular language $L$ on the
  alphabet $\Sigma$ of the relation symbols of~$\sigma$. Its semantics is that
  $q_L$ has two free variables $x$ and $y$, and $q_L(a, b)$ holds on an instance
  $I$ for $a, b \in \dom(I)$ precisely when there is a directed path $\pi$ 
  of relations of~$\sigma$ 
  from $a$ to
  $b$ such that the label of~$\pi$ is in $L$. A \emph{two-way regular
  path query} (2RPQ) is an RPQ on the alphabet $\Sigma^\pm \defeq \Sigma \sqcup
  \{R^- \mid R \in \Sigma\}$, which holds whenever 
  there is a path from~$a$ to~$b$ with label in~$L$, with $R^-$ meaning that we traverse an
  $R$-fact in the reverse direction. 
  A C2RPQ $q = \bigwedge_{i=1}^n q_i(z_i,z'_i)$ is a conjunction of 2RPQs, i.e., a conjunctive query made from atoms $q_i(z_i,z'_i)$ that are 2RPQs ($z_i$ and $z'_i$ are not necessarily distinct). 
The \emph{graph} of $q$ is the unlabeled undirected graph having as vertices the variables of $q$ and whose set of edges is $\{\{z_i,z'_i\} \mid 1 \leq i \leq n,~z_i \neq z'_i\}$.
A C2RPQ is \emph{acyclic} if its graph is acyclic. 
A \emph{strongly acyclic C2RPQ} (SAC2RPQ) is an
acyclic C2RPQ that further satisfies: 1) for $1 \leq i \leq n$, we have $z_i \neq z_i'$ (no self-loops); and 2) for $1 \leq i < j \leq n$, we have $\{z_i, z'_i\} \neq \{z_j,z'_j\}$ (no multi-edges). 
  A \emph{Boolean 2RPQ} (resp., \emph{Boolean C2RPQ}) is a 2RPQ (resp., C2RPQ) which is existentially quantified on all its free variables.
\end{definition}

\begin{proposition}[\cite{mendelzon1989finding,barcelo2013querying}]
  \label{prp:2rpq-to-icg}
  2RPQ query evaluation (on arbitrary instances) has linear time combined complexity.
  \end{proposition}

CFG-Datalog allows us to capture this result for Boolean SAC2RPQs on treelike
instances. We can prove the following result, for Boolean 2RPQs and SAC2RPQs, which 
further implies
translatability to automata (and efficient computation of provenance
representations). We do not know whether this extends to the more general classes studied 
in~\cite{barcelo2014does}.

\begin{proposition}\label{prp:rpqcompile}
	Given a Boolean SAC2RPQ $Q$ (where each 2RPQ is given as a regular expression), we can compute in time
  $O(\card{Q})$ an equivalent CFG-Datalog program $P$ of body size~$4$.
\end{proposition}

\begin{proof}
  We first show the result for 2RPQs, and then explain how to extend it to
  SAC2RPQs.

  We first use Thompson's
  construction~\cite{thompson1968programming}
  to compute in linear time an equivalent NFA $A$ (with $\epsilon$-transitions) on
  the alphabet $\Sigma^\pm$. Note that the result of Thompson's
  construction has
  exactly one final state, so we may assume that $A$ has exactly one
  final state.

  We now define the intensional signature of the CFG-Datalog program to consist
  of one unary predicate $P_q$ for each state $q$ of the automaton, in addition
  to $\text{Goal}()$. We add the rule $\text{Goal}() \leftarrow P_{q_\f}(x)$ for the
  final state $q_\f$, and for each extensional relation $R(x, y)$, 
  we add the rules $P_{q_0}(x)
  \leftarrow R(x, y)$ and $P_{q_0}(y)
  \leftarrow R(x, y)$, where $q_0$ is the initial state. We then add
  rules corresponding to automaton transitions:
  \begin{itemize}
    \item for each transition from $q$ to $q'$ labeled with a letter
      $R$, we add the rule $P_{q'}(y) \leftarrow P_q(x), R(x, y)$;
    \item for each transition from $q$ to $q'$ labeled with a negative letter
      $R^-$, we add the rule $P_{q'}(y) \leftarrow P_q(x), R(y, x)$;
    \item for each $\epsilon$-transition from $q$ to $q'$ we add the rule
      $P_{q'}(x) \leftarrow P_q(x)$
  \end{itemize}

  This transformation is clearly in linear time,
  and the result clearly satisfies the desired body size
  bound. Further, as the result is a monadic Datalog program, it is clearly a
  CFG-Datalog program. Now, it is clear that, in any instance~$I$ where $Q$
  holds, from two witnessing elements $a$ and $b$ and a path $\pi: a = c_0, c_1, \ldots,
  c_n = b$ from $a$ to $b$ satisfying $Q$, we can build a derivation tree of the
  Datalog program by deriving $P_{q_0}(a), P_{q_1}(c_1), \ldots, P_{q_n}(c_n)$,
  where $q_0$ is the initial state and $q_n$ is final, to match the accepting
  path in the automaton $A$ that witnesses that $\pi$ is a match of~$Q$.
  Conversely, any derivation tree of the Datalog program $P$ that witnesses that an
  instance satisfies $P$ can clearly be used to extract a path of relations
  which corresponds to an accepting run in the automaton.

  \medskip

  We now extend this argument to SAC2RPQs. 
  We start with a preliminary observation on CFG-Datalog programs: any rule
  of the form (*) $A(x) \leftarrow A_1(x), \ldots, A_n(x)$, where $A$ and each
  $A_i$ is a unary atom, can be rewritten in linear time to rules with bounded
  body size, by creating unary intensional predicates $A_i'$ for $1 \leq i \leq
  n$, writing the rule
  $A_n'(x) \leftarrow A_n(x)$, writing the rule $A_i'(x) \leftarrow A_{i+1}'(x), A_i(x)$ for each
  $1 \leq i < n$, and writing the rule $A(x) \leftarrow A_1'(x)$. Hence, we will write rules of
  the form (*) in the transformation, with unbounded body size, being understood
  that we can finish the process by rewriting out each rule of this form to
  rules of bounded body size.

  Given a SAC2RPQ $Q$, we compute in linear time the undirected graph $G$ on
  variables, and its connected components. Clearly we can rewrite each connected
  component separately, by defining one $\text{Goal}_i()$ 0-ary predicate for
  each connected component $i$, and adding the rule $\text{Goal}() \leftarrow
  \text{Goal}_1(), \ldots, \text{Goal}_n()$: this is a rule of form (*), which we can
  rewrite. Hence, it suffices to consider each connected component separately.

  Hence, assuming that the graph $G$ is connected, we root it at an arbitrary
  vertex to obtain a tree $T$. For each node $n$ of~$T$ (corresponding to a variable of the
  SAC2RPQ), we define a unary intensional predicate $P'_n$ which will intuitively
  hold on elements where there is a match of the sub-SAC2RPQ defined by the
  subtree of~$T$ rooted at~$n$, and one unary intensional predicate $P''_{n,n'}$ for
  all non-root~$n$ and children $n'$ of~$n$ in~$T$ which will hold whenever
  there is a match of the sub-SAC2RPQ rooted at $n$ which removes all children
  of~$n$ except~$n'$. Of course we add the rule $\text{Goal}()
  \leftarrow P'_{n_\r}(x)$, where $n_\r$ is the root of~$T$.
  
  Now, we rewrite the SAC2RPQ to monadic Datalog
  by rewriting each edge of $T$ independently, as in the argument for 2RPQs
  above. Specifically, we assume that the edge when read from bottom to top
  corresponds to a 2RPQ; otherwise, if the edge is oriented in the wrong
  direction, we can clearly compute an automaton for the
  reverse language in linear time from the Thompson automaton, by reversing the
  direction of transitions in the automaton, and swapping the initial state and
  the final state. We modify the
  previous construction by replacing the rule for the initial state $P_{q_0}$ by
  $P_{q_0}(x) \leftarrow P'_{n'}(x)$ where $n'$ is the lower node of the 
  edge that we are rewriting, and the rule for the goal predicate in the head is replaced by a rule
  $P''_{n,n'}(x) \leftarrow P_{q_\f}(x)$, where $n$ is the upper node of the
  edge, and $q_\f$ is the final state of the
  automaton for the edge: this is the rule that defines the $P''_{n,n'}$.
  
  Now, we define each $P'_n$ as follows:

  \begin{itemize}
    \item If $n$ is a leaf node of~$T$, we define $P'_n$ by the same rules that
      we used to define $P_{q_0}$ in the previous construction, so that $P'_n$ holds
      of all elements in the active domain of an input instance.
    \item If $n$ is an internal node of~$T$, we define $P'_n(x) \leftarrow
      P''_{n,n_1}(x), \ldots, P''_{n,n_m}(x)$, where $n_1, \ldots, n_m$ are the
      children of~$n$ in~$T$: this is a rule of form (*).
  \end{itemize}

  Now, given an instance $I$ satisfying the SAC2RPQ, from a match of the SAC2RPQ
  as a rooted tree of paths, it is easy to see by bottom-up induction on the
  tree that we derive $P_v$ with the desired semantics, using the correctness of
  the rewriting of each edge. Conversely, a derivation tree for the rewriting can be
  used to obtain a rooted tree of paths with the correct structure where each
  path satisfies the RPQ corresponding to this edge.
\end{proof}

The rest of the article presents the tools needed for our tractability results (alternating two-way automata and cyclic provenance circuits) and their technical proofs.

\section{Translation to Automata}
\label{sec:compilation}
In this section, we study how we can translate CFG-Datalog queries on treelike
instances to tree automata, to be able to evaluate them efficiently.
As we showed with Propositions~\ref{prp:bntalower} and~\ref{prp:simplicial} (remembering that $\alpha$-acyclic queries have bounded simplicial width), we need more expressive
automata than bNTAs.
Hence, we use instead the formalism of 
\emph{alternating two-way automata} \cite{tata}, i.e., automata that can
navigate in trees in any direction, and can express transitions using Boolean
formulas on states.
Specifically, we introduce for our purposes a variant of these automata, which
are \emph{stratified} (i.e., allow a form of
stratified negation), and \emph{isotropic} (i.e., no direction is privileged,
in particular order is ignored).

As in Section~\ref{sec:treelike}, 
we will define tree automata that run on \emph{$\Gamma$-trees} for some alphabet
$\Gamma$: a $\Gamma$-tree $\la T,
\lambda \ra$ is a
finite rooted ordered
tree with a labeling function $\lambda$ from the nodes of~$T$ to
$\Gamma$.
The \emph{neighborhood} $\neigh(n)$ of a node $n \in T$ is the set which
contains~$n$, all children of~$n$, and the parent of~$n$ if it exists.
\subparagraph*{Stratified isotropic alternating two-way automata.}
To define the transitions of our alternating automata, we write 
$\calB(X)$ the set of propositional formulas (not necessarily monotone)
over a set~$X$ of variables: we will assume without loss of generality that \emph{negations are only applied to variables}, 
which we can always enforce using De Morgan's laws.
A \emph{literal} is a propositional variable $x\in X$ (\emph{positive} literal) or the
negation of a propositional variable~$\lnot x$ (\emph{negative} literal).

A \emph{satisfying assignment} of $\phi
\in \calB(X)$ consists of two \emph{disjoint} sets $P, N \subseteq X$
(for ``positive'' and ``negative'') such that $\phi$ is a tautology when
substituting the variables of $P$ with~$1$ and those of $N$ with $0$,
i.e., when we have $\nu(\phi)=1$ for every valuation $\nu$ of $X$ such that $\nu(x) = 1$ for all $x \in P$ and
$\nu(x) = 0$ for all $x \in N$.
Note that we allow satisfying assignments with $P \sqcup N
\subsetneq X$, which will be useful for our technical results.
We now define our automata:

\begin{definition}
  \label{def:satwa}
A \emph{stratified isotropic alternating two-way automaton} on $\Gamma$-trees
  ($\Gamma$-SATWA) is a tuple $A=(\calQ,q_{\I}, \Delta, \strat)$ with $\calQ$ a finite set of \emph{states}, $q_{\I}$ the \emph{initial state}, 
$\Delta$ the \emph{transition function} from $\calQ \times \Gamma$ to $\calB(\calQ)$, and 
  $\strat$ a \emph{stratification function}, i.e., a surjective function from $\calQ$
  to~$\{0,\ldots,m\}$ for some $m \in \NN$,
  such that for any $q, q' \in \calQ$ and $f \in \Gamma$,
  if $\Delta(q,f)$ contains~$q'$ as a
  positive literal (resp., negative literal), then 
  $\strat(q') \leq \strat(q)$ (resp., $\strat(q') < \strat(q)$).

  We define by induction on $0 \leq i \leq m$ an \emph{$i$-run} of~$A$
  on a $\Gamma$-tree $\la T , \lambda \ra$ as a finite tree $\la T_\r,
  \lambda_\r \ra$, with labels of the
  form $(q,w)$ or $\lnot (q,w)$ for $w\in T$ and $q \in \calQ$ with $\strat(q)
  \leq i$, by the
  following (nested) inductive definition on $T_r$:
  \begin{enumerate}
  \item For $q \in \calQ$ such that $\strat(q) < i$,
    the singleton tree $\la T_\r, \lambda_\r \ra$
    with one node labeled by $(q, w)$ (resp., by~$\neg (q,
    w)$) is an $i$-run if
      there is a $\strat(q)$-run of~$A$ on $\la T , \lambda \ra$ whose root is
      labeled by~$(q, w)$ (resp., if there is no such run);
  \item For $q \in \calQ$ such that $\strat(q) = i$,
    if $\Delta(q,\lambda(w))$ has a satisfying assignment $(P,N)$, if we have an $i$-run
          $T_{q^-}$ for each $q^- \in N$ with root labeled by $\neg (q^-, w)$, and an
          $i$-run $T_{q^+}$ for each $q^+ \in P$ with root labeled by $(q^+, w_{q^+})$ for
          some $w_{q^+}$ in~$\neigh(w)$, then the tree
          $\la T_\r, \lambda_\r \ra$
          whose root is labeled $(q,
          w)$ and has as children all the $T_{q^-}$ and $T_{q^+}$ is an
          $i$-run.
  \end{enumerate}
  A \emph{run} of $A$ starting in a state $q \in \calQ$ at a node $w \in T$ is an
  $m$-run whose root is labeled $(q, w)$.
  We say that $A$ \emph{accepts} $\la T, \lambda \ra$ (written $\la T, \lambda \ra \models A$) if there
  exists a run of $A$ on $\la T, \lambda \ra$ starting in
  the initial state $q_{\I}$ at the root of~$T$.
\end{definition}

Observe that the internal nodes of a run starting in some state $q$ are labeled
by states~$q'$ in the same stratum as~$q$. The leaves of the run may be labeled
by states of a strictly lower stratum or negations thereof, or by states of the
same stratum whose transition function is tautological, i.e., by some $(q', w)$
such that $\Delta(q', \lambda(w))$ has $\emptyset, \emptyset$ as a satisfying
assignment. Intuitively, if we disallow negation in transitions, our automata
amount to
the alternating two-way automata used by \cite{cachat2002two}, with the
simplification that they do not need parity acceptance conditions (because we
only work with finite trees), and that they are \emph{isotropic}:
the run for each positive child state of an internal node may
start indifferently on \emph{any} neighbor of~$w$ in the tree
(its parent, a child, or $w$ itself), no matter the direction. (Note, however,
that the run for negated child states must start on~$w$ itself.)

We will soon explain how the translation of CFG-Datalog is performed, but we
first note that evaluation of $\Gamma$-SATWAs is in linear time.
In fact, this result follows from the definition of provenance cycluits
for SATWAs in the next section, and the claim that these cycluits can be
evaluated in linear time.

\begin{proposition}\label{prp:satwaeval}
  For any alphabet $\Gamma$, given a $\Gamma$-tree $\la T,\lambda\ra$ and a $\Gamma$-SATWA $A$,
  we can determine whether $\la T,\lambda\ra\models A$ in time $O(\card{T} \cdot \card{A})$.
\end{proposition}

\begin{proof}
  This proof depends on notions and
  results that are given in the rest of the paper, hence can be skipped at first reading of Section~\ref{sec:compilation}.

  We use Theorem~\ref{thm:satwaprov} to compute a provenance cycluit~$C$ of the
  SATWA (modified to be a $\overline{\Gamma}$-SATWA by simply ignoring the
  second component of the alphabet) in time $O(\card{T}\cdot\card{A})$.
  Then we conclude by evaluating the
  resulting provenance cycluit (for an arbitrary valuation of that
  circuit) in time $O(\card{C})$ using
  Proposition~\ref{prp:cycluitlinear}.

  Note that, intuitively, the fixpoint evaluation of the cycluit can be understood as a
  least fixpoint computation to determine which pairs of states and tree nodes
  (of which there are $O(\card{T} \cdot \card{A})$) are reachable.
\end{proof}

We now give our main translation result: we can efficiently translate any CFG-Datalog program of
bounded body size into a stratified alternating two-way automaton that \emph{tests} it (in the same sense as for
bNTAs). 
For pedagogical purposes, we present the translation for a subclass of CFG-Datalog, namely, \emph{CFG-Datalog with guarded negations} (CFG$^{\text{GN}}$-Datalog), in which
invocations of negative intensional predicates are guarded in rule bodies:

\begin{definition}
  \label{def:CFG-GN}
  Let $P$ be a stratified Datalog program.
  A negative intensional literal $\lnot A(\mathbf{x})$ in a 
  rule body~$\psi$ of~$P$
  is \emph{clique-guarded} if, for any two variables
  $x_i \neq x_j$ of~$\mathbf{x}$,
  it is the case that $x_i$ and $x_j$
  co-occur in some positive atom of~$\psi$.
  A CFG$^\text{GN}$-Datalog program is a CFG-Datalog program such that 
  for any rule
  $R(\mathbf{x}) \leftarrow \psi(\mathbf{x}, \mathbf{y})$, every
  negative intensional literal in $\psi$ is clique-guarded in
  $\psi$.
\end{definition}

We will then prove in Section~\ref{sec:proof} the following translation result, and explain at the end of Section~\ref{sec:proof} how it can be extended to full CFG-Datalog:
\begin{theorem}
\label{thm:maintheorem}
	Given a CFG$^\text{GN}$-Datalog program~$P$ of body size
  $\kp$ and $\ki \in \NN$,
  we can build in FPT-linear time in~$|P|$
  (parameterized by~$\kp, \ki$)
  a SATWA $A_P$ testing~$P$ on instances of treewidth~$\leq \ki$.
\end{theorem}
\begin{proof}[Proof sketch]
For every relational symbol $R$, we introduce states of the form 
$q_{R(\mathbf{x})}^{\nu}$, where $\nu$ is a partial valuation of $\mathbf{x}$.
  The semantics is that we can start a run at state $q_{R(\mathbf{x})}^{\nu}$
  iff
        we can navigate the tree encoding
        to build a total valuation $\nu'$ that extends $\nu$ and
	such that $R(\nu'(\mathbf{x}))$ holds. Once we have built~$\nu'$, if~$R$
        is an extensional relation, we just check that $R(\nu'(\mathbf{x}))$ appears in the
        tree encoding. If $R$ 
        is intensional, we use the clique-guardedness condition to argue that
        the elements of $\nu'(\mathbf{x})$ can be found together in a bag. We
        then choose 
        a rule~$r$ with head relation~$R$, instantiate its head variables
        according to~$\nu'$, and inductively check all literals of the body
        of~$r$.
	The fact that the automaton is isotropic relieves us from the syntactic
        burden of dealing with directions in the tree, as one usually has to do with alternating two-way automata.
\end{proof}

\section{Provenance Cycluits}
\label{sec:provenance}
In the previous section, we have seen how CFG-Datalog programs could be translated
efficiently to tree automata that test them on treelike instances. To
show that SATWAs can be evaluated in linear time
(stated earlier as Proposition~\ref{prp:satwaeval}), we will
introduce an operational semantics for SATWAs based on the notion of
\emph{cyclic circuits}, or \emph{cycluits} for short.

We will also use
these cycluits as a new powerful tool to 
compute (Boolean) \emph{provenance
information}, i.e., a representation of how the query result depends on the
input data:

\begin{definition}
  \label{def:provenance}
  A (Boolean) \emph{valuation} of a set $S$ is a function $\nu: S \to \{0, 1\}$.
  A \emph{Boolean function} $\phi$ on variables~$S$ is a mapping that associates
  to each valuation $\nu$ of~$S$ a Boolean value in $\{0, 1\}$ called the
  \emph{evaluation} of~$\phi$ according to~$\nu$; for consistency with further
  notation, we write it $\nu(\phi)$.
  The \emph{provenance} of a query $Q$ on an
  instance~$I$ is the Boolean function $\phi$ whose variables are the facts
  of~$I$, which is defined as follows:
  for any valuation $\nu$ of the facts of $I$, we have $\nu(\phi)
  = 1$ iff the subinstance $\{F \in I \mid \nu(F) = 1\}$ satisfies~$Q$.
\end{definition}

We can represent Boolean provenance as
Boolean formulas
\cite{ImielinskiL84,green2007provenance}, or (more recently) as
Boolean circuits \cite{deutch2014circuits,amarilli2015provenance}.
In this section, we first introduce \emph{monotone cycluits} 
(monotone Boolean circuits with cycles), for which we define a semantics
(in terms of the
Boolean function that they express);
we also show that cycluits can be
evaluated in linear time, given a valuation. 
Second, we extend them to \emph{stratified cycluits}, allowing a form of stratified
negation.
We conclude the section by showing how to construct the \emph{provenance}
of a SATWA as a cycluit, in FPT-bilinear time. Together with
Theorem~\ref{thm:maintheorem}, this claim implies our main provenance result:

\begin{theorem}
  \label{thm:mainprov}
  Given a CFG-Datalog program~$P$ of body size
  $\kp$ and a relational instance~$I$ of treewidth $\ki$,
  we can construct in FPT-bilinear time in~$|I|\cdot|P|$
  (parameterized by~$\kp$ and~$\ki$)
  a representation of the provenance
  of $P$ on $I$ as a stratified cycluit. 
\end{theorem}

Of course, this result implies the analogous claims for query languages that are
captured by CFG-Datalog parameterized by the body size, as we studied in
Section~\ref{sec:CFG}. When combined with the fact that cycluits can
be tractably evaluated, it yields our main result,
Theorem~\ref{thm:main}. The rest of this section formally introduces
cycluits and proves Theorem~\ref{thm:mainprov}.

\subparagraph*{Cycluits.}
We coin the term \emph{cycluits} for Boolean circuits without the
acyclicity requirement. This is the same kind of objects studied 
in~\cite{riedel2012cyclic}. To avoid the problem of feedback loops, however,
we first study \emph{monotone cycluits}, and then cycluits with stratified
negation.

\begin{definition}
  \label{def:cycluits}
    A \emph{monotone Boolean cycluit} is a directed graph $C = (G,W,g_0,\mu)$ 
    where $G$ is the set of \emph{gates},
    $W\subseteq G^2$ is the set of directed edges called \emph{wires} (and written $g
    \rightarrow g'$),
    $g_0 \in G$ is the \emph{output gate}, and
    $\mu$ is the \emph{type} function mapping each gate
    $g \in G$ to one of $\inp$ (input gate, with no incoming wire in $W$), $\land$
  (AND gate) or $\lor$ (OR gate).
\end{definition}

We now define the semantics of monotone cycluits. 
A (Boolean) \emph{valuation} of~$C$ is a function $\nu : C_{\inp} \to \{0,1\}$ 
indicating the value of the input gates.
As for standard monotone circuits, a valuation yields 
an \emph{evaluation} $\nu' : C \to \{0,1\}$, that we will define shortly,
indicating the value of each gate under the valuation~$\nu$: we abuse notation
and
write
$\nu(C) \in \{0, 1\}$ for the \emph{evaluation result}, i.e.,
$\nu'(g_0)$ where 
$g_0$ is the
output gate of~$C$. The Boolean function \emph{captured} by
a cycluit $C$ is thus 
the Boolean
function $\phi$ on~$C_\inp$ defined by $\nu(\phi) \colonequals \nu(C)$ for each
valuation
$\nu$ of~$C_\inp$. 
We define the evaluation~$\nu'$ from~$\nu$ by a least fixed-point computation: we
set all input gates to their value by~$\nu$, and other gates to $0$.
We then iterate until the evaluation no longer changes, by evaluating OR-gates
to~$1$ whenever some input evaluates to~$1$, and AND-gates to~$1$
whenever all their inputs evaluate to~$1$.
Formally, the semantics of monotone cycluits is defined by
Algorithm~\ref{alg:semantics-monotone}.
\begin{algorithm}
\DontPrintSemicolon
\KwIn{Monotone cycluit $C=(G,W,g_0,\mu)$, valuation $\nu: C_{\inp} \to
  \{0,1\}$}
        \KwOut{$\{g \in C \mid \nu'(g)=1\}$}
	$S_0 \defeq \{ g \in C_\inp \mid \nu(g)=1 \}$\;
	$i \defeq 0$\;
	\Do{$S_i \neq S_{i-1}$}{
		$i$++\;
                $S_{i} \defeq S_{i-1} \cup\Big\{g \in C \mid (\mu(g) =
        {\lor}), \exists g^\prime \in S_{i-1}, g^\prime \rightarrow
      g\in W\Big\} \cup{}$\\$\quad \Big\{g \in C \mid (\mu(g) =
        {\land}), \{g^\prime \mid g^\prime \rightarrow g\in W\} \subseteq
    S_{i-1}\Big\}$}
	\Return $S_i$\;
	\caption{Semantics of monotone cycluits}
	\label{alg:semantics-monotone}
\end{algorithm}

The Knaster--Tarski
theorem~\cite{tarski1955lattice} gives an equivalent characterization:

\begin{proposition}\label{prp:altersem}
  For any monotone cycluit $C$ and Boolean valuation $\nu$ of $C$,
	letting $\nu'$ be the evaluation (as defined by Algorithm~\ref{alg:semantics-monotone}),
  the set $S \defeq \{g \in C \mid \nu'(g) = 1\}$ is \emph{the} minimal set of
  gates (under inclusion) such that:
	\begin{compactenum}[(i)]
    \item $S$ contains the true input gates, i.e., it contains $\{g \in C_\inp \mid
  \nu(g) = 1\}$;
\item for any $g$ such that $\mu(g) = \lor$, if some input gate of $g$ is
  in $S$, then $g$ is in~$S$;
\item for any $g$ such that $\mu(g) = \land$, if all input gates of $g$ are
  in $S$, then $g$ is in~$S$.
\end{compactenum}
\end{proposition}
\begin{proof}
  The operator used in Algorithm~\ref{alg:semantics-monotone} is clearly
  monotone, so by the Knaster--Tarski theorem, the outcome of the
  computation is the intersection of all sets of gates satisfying the
  conditions in Proposition~\ref{prp:altersem}.
\end{proof}

Algorithm~\ref{alg:semantics-monotone} is a naive fixpoint algorithm running in quadratic time, but we show that the same output can be computed in linear time with Algorithm~\ref{alg:linear-time-cycluits}.
\begin{algorithm}
\DontPrintSemicolon
\KwIn{Monotone cycluit $C=(G,W,g_0,\mu)$, valuation $\nu: C_{\inp} \to \{0,1\}$}
        \KwOut{$\{g \in C \mid \nu'(g)=1\}$}
        ~\tcc{Precompute the in-degree of $\land$ gates}
        \For{$g \in C$ s.t.\ $\mu(g) = \land$}{
          $M[g] \defeq \card{\{g' \in C \mid g' \rightarrow g\}}$\;
        }
        $Q\defeq\{g \in C_\inp \mid \nu(g)=1\} \cup \{g \in C \mid (\mu(C) =
        \land) \land M[g] = 0\}$ \tcc*{as a stack} $S\defeq \emptyset$ \tcc*{as a bit array}
	\While{$Q \neq \emptyset$}{
          pop $g$ from $Q$\;
                        \If{$g \notin S$}{
                        add $g$ to~$S$\;
                        \For{$g^\prime \in C \mid g \to g^\prime$}{
                                \If{$\mu(g^\prime)=\lor$}{
                                        push $g^\prime$ into $Q$\;
                                }
                                \If{$\mu(g^\prime)=\land$}{
                                        $M[g^\prime]\defeq M[g^\prime]-1$\;
                                        \If{$M[g^\prime]=0$}{
                                        push $g^\prime$ into $Q$\;
                                        }
                                }
                        }
                      }
                        
        }
	\Return $S$
        \caption{Linear-time evaluation of monotone cycluits}
	\label{alg:linear-time-cycluits}
\end{algorithm}

\begin{proposition}\label{prp:moncycluitlinear}
  Given any monotone cycluit $C$ and Boolean valuation $\nu\!$ of~$C$, we can compute the
  evaluation $\nu'\!$ of~$C$ in linear time.
\end{proposition}

\begin{proof}
	We use Algorithm~\ref{alg:linear-time-cycluits}.
  We first prove the claim about the running time. The preprocessing to compute
  $M$ is in linear-time in~$C$ (we enumerate at most once every wire), and the rest of the algorithm is clearly
  in linear time as it is a variant of a DFS traversal of the graph, with the added
  refinement that we only visit nodes that evaluate to~$1$ (i.e., OR-gates
  with some input that evaluates to~$1$, and AND-gates where all inputs
  evaluate to~$1$).

  We now prove correctness. We use the characterization of
  Proposition~\ref{prp:altersem}. We first check that $S$ satisfies the
  properties:
  
	\begin{compactenum}[(i)]
    \item $S$ contains the true input gates by construction.
    \item Whenever an OR-gate $g'$ has an input gate $g$ in~$S$, then, when
      we added $g$ to~$S$, we have necessarily followed the wire $g \rightarrow
      g'$ and added $g'$ to~$Q$, and later added it to~$S$. 
    \item Whenever an AND-gate $g'$ has
  all its input gates $g$ in~$S$, there are two cases. The first case is when
      $g$ has no input gates at all, in which case $S$ contains it by
      construction. The second case is when $g'$ has input gates: in this
      case, observe that $M[g']$ was initially equal to the fan-in of~$g'$, and
      that we decrement it for each input gate $g$ of~$g'$ that we add to~$S$.
      Hence, considering the last input gate $g$ of~$g'$ that we add to~$S$, it
      must be the case that $M[g']$ reaches zero when we decrement it,
      and then we add $g'$ to~$Q$, and later to~$S$.
  \end{compactenum}
  
  Second, we check that $S$ is minimal. Assume by contradiction that it
  is not the case, and consider the first gate $g$ which is
  added to $S$ while not being in the minimal Boolean valuation $S'$.  It cannot
  be the case that $g$ was added when initializing $S$, as we initialize $S$ to
  contain true input gates and AND-gates with no inputs, which must be true
  also in~$S'$ by the characterization of Proposition~\ref{prp:altersem}. Hence,
  we added $g$ to~$S$ in a later step of the algorithm. However,
  we notice that we must have added $g$ to $S$ because of the value of its input
  gates. By minimality of~$g$, these input gates have the same value in~$S$ and
  in~$S'$. This yields a contradiction, because the gates that we add to~$S$ are
  added following the characterization of Proposition~\ref{prp:altersem}.
  This concludes the proof.
\end{proof}

Another way to evaluate cycluits in linear time is by a
rewriting of the circuit to a Horn formula, whose minimal model can be
computed in linear time~\cite{dowling1984linear} and corresponds to the cycluit
evaluation.

\subparagraph*{Stratified cycluits.}
We now move from monotone cycluits to general cycluits featuring negation.
However, allowing arbitrary negation would make it difficult to define a
proper
semantics, because of possible cycles of negations. Hence, we focus on
\emph{stratified cycluits}:

\begin{definition}
  A \emph{Boolean cycluit} $C$ is defined like a \emph{monotone cycluit},
  but further allows NOT-gates ($\mu(g) = \neg$), which are required to have a
  single input. It is 
  \emph{stratified} if there exists a surjective
  \emph{stratification function}~$\strat$ mapping its gates
  to $\{0,\ldots,m\}$ for some $m \in \NN$ such that $\strat(g) = 0$ iff $g
  \in C_\inp$, and
  $\strat(g)
  \leq \strat(g')$ for each wire $g \rightarrow g'$, the inequality being strict
  if $\mu(g') = \neg$.
\end{definition}

This notion of stratification is similar to that of stratification of
Datalog programs or that of
stratification of Horn formulas \cite{dantsin2001complexity}.

Equivalently, we can show that $C$ is stratified if and only if it contains no cycle of gates involving a $\neg$-gate. Moreover if $C$ is
stratified we can
compute a stratification function in linear time, from a topological sort of its strongly connected components:

\begin{definition}
A \emph{strongly connected component} (SCC) of a directed graph $G=(V,E)$ is a subset $S \subseteq V$ that is maximal by inclusion and which ensures that for any $x,y \in S$ with $x \neq y$, there is a directed path from $x$ to $y$ in $G$. Observe that the SCCs of $G$ are disjoint.
A \emph{topological sort} of the SCCs of $(G,W)$ is a linear ordering $(S_1,\ldots,S_k)$ of all the SCCs of $G$ such that for any $1 \leq i < j \leq k$ and $x \in S_i$ and $y \in S_j$, there is no directed path
from $y$ to $x$ in $G$.
\end{definition}

Such a topological sort always exists and can be computed in linear time from $G$~\cite{tarjan1972depth}. We can then show:

\begin{proposition}\label{prp:stratifun}
  Any Boolean cycluit $C$ is stratified iff it it contains no cycle of gates involving a $\neg$-gate.
Moreover, a stratification function can be computed in linear time from~$C$.
\end{proposition}

\begin{proof}
  To see why a stratified Boolean cycluit $C$ cannot contain a cycle of gates involving a $\neg$-gate,
  assume by contradiction that it has such a cycle $g_1 \rightarrow g_2
  \rightarrow \cdots \rightarrow g_n \rightarrow g_1$. As $C$ is stratified, there exists a stratification function $\strat$.
  From the properties of a stratification function, we know that $\strat(g_1) \leq \strat(g_2) \leq \cdots \leq
  \strat(g_1)$, so that we must have $\strat(g_1) = \cdots = \strat(g_n)$.
  However, letting $g_i$ be such that $\mu(g_i) = \neg$, we know that
  $\strat(g_{i-1}) < \strat(g_i)$ (or, if $i = 1$, $\strat(g_n) < \strat(g_1)$),
  so we have a contradiction.

  We now prove the converse direction of the claim, i.e., that any 
  Boolean cycluit which does not contain a cycle of gates involving a $\neg$-gate must have a stratification function, 
  and show how to compute such a function in linear time.
  Compute in linear time the strongly connected components (SCCs) of~$C$, and a
  topological sort of the SCCs. As the input gates of~$C$ do not themselves have
  inputs, each of them must have their own SCC, and each such SCC must be a
  leaf, so we can modify the topological sort by merging these SCCs
  corresponding to input gates, and putting them first in the topological sort.
  We define the function
  $\strat$ to map each gate
  of~$C$ to the index number of its SCC in the topological sort, which ensures
  in particular that the input gates of~$C$ are exactly the gates assigned
  to~$0$ by~$\strat$.
  This can be performed in linear time. Let us show that the result $\strat$ is
  a stratification function:
  \begin{itemize}
    \item For any edge $g \rightarrow g'$, we have $\strat(g) \leq \strat(g')$.
      Indeed, either $g$ and $g'$ are in the same strongly connected component and we
      have $\strat(g) = \strat(g')$, or they are not and in this case the edge
      $g \rightarrow g'$ witnesses that the SCC of $g$ precedes that of~$g'$,
      whence, by definition of a topological sort, it follows that $\strat(g) <
      \strat(g')$.
    \item For any edge $g \rightarrow g'$ where $\mu(g') = \lnot$, we have $\strat(g) <
      \strat(g')$. Indeed, by adapting the reasoning of the previous bullet
      point, it suffices to show that $g$ and $g'$ cannot be in the same
      SCC. Indeed, assuming by contradiction that they are, by
      definition of a SCC, there must be a path from $g'$ to $g$, and combining
      this with the edge $g \rightarrow g'$ yields a cycle involving a
      $\lnot$-gate,
      contradicting our assumption on~$C$.\qedhere
  \end{itemize}
\end{proof}

We can then use any stratification function to define
the evaluation of~$C$ (which will be independent of the choice of stratification function):

\begin{definition}
  \label{def:strateval}
  Let $C$ be a stratified cycluit with stratification function $\strat: C
  \rightarrow \{0,\ldots,m\}$,
  and let $\nu$ be a Boolean valuation of~$C$. We inductively define the
  \emph{$i$-th stratum evaluation} $\nu_i$, for $i$ in the range of~$\strat$,
  by setting $\nu_0 \colonequals \nu$, and letting $\nu_{i}$ extend the $\nu_j$ ($j < i$) as
  follows:
  \begin{enumerate}
  \item For $g$ such that $\strat(g) = i$ with $\mu(g) = \neg$, set $\nu_i(g)
    \defeq \neg \nu_{\strat(g')}(g')$ for its one input~$g'$.
  \item Evaluate all other $g$ with $\strat(g) = i$ as for monotone cycluits,
    considering the $\neg$-gates of point 
    1.\ and all gates of 
    stratum $< i$ as input gates fixed to their value in~$\nu_{i-1}$.
  \end{enumerate}
  Letting $g_0$ be the output gate of~$C$,
  the Boolean function $\phi$ \emph{captured} by~$C$ is then defined
  as $\nu(\phi) \colonequals \nu_m(g_0)$
  for each valuation $\nu$ of $C_\inp$.
\end{definition}

\begin{proposition}\label{prp:cycluitlinear}
  We can compute $\nu(C)$ in linear time in the stratified cycluit $C$ and in~$\nu$.
  Moreover, the result is independent of the chosen stratification function.
\end{proposition}

\begin{proof}
	Compute in linear time a stratification function $\strat$ of $C$ 
	using Proposition~\ref{prp:stratifun}, and compute the evaluation following Definition~\ref{def:strateval}.
	This can be performed in linear time.
        To see why this evaluation is independent from the choice of
        stratification, observe that any stratification function must clearly assign the
        same value to all gates in an SCC. Hence,
        choosing a stratification function amounts to choosing
        the stratum that we assign to each SCC. Further,
        when an SCC $S$ precedes another SCC $S'$, the stratum of $S$ must be no
        higher than the stratum of~$S'$. So in fact the only freedom that we
        have is to choose a topological sort of the SCCs, and optionally to
        assign the same stratum to consecutive SCCs in the topological sort:
        this amounts to ``merging'' some SCCs, and is only possible when there
        are no $\lnot$-gates between them. Now, in the evaluation, it is clear
        that the order in which we evaluate the SCCs makes no difference, nor
        does it matter if some SCCs are evaluated simultaneously. Hence, the
        evaluation of a stratified cycluit is well-defined.
\end{proof}

\subparagraph*{Building provenance cycluits.}
Having defined cycluits as our provenance representation, we compute 
the provenance of a query on an instance as the 
\emph{provenance} of its SATWA on a tree encoding.
To do so, we must give a general definition of the provenance of SATWAs.
Consider a $\Gamma$-tree $\calT \defeq \la T, \lambda \ra$ for some
alphabet $\Gamma$, as in 
Section~\ref{sec:compilation}.
We define
a (Boolean) \emph{valuation}~$\nu$ of $\calT$ 
as a mapping from the nodes of~$T$ to $\{0, 1\}$.
Writing $\overline{\Gamma} \colonequals \Gamma \times \{0, 1\}$, 
each valuation~$\nu$ then defines a $\overline{\Gamma}$-tree
$\nu(\calT) \defeq \la T, (\lambda \times \nu) \ra$, obtained by annotating each
node of~$\calT$ by its $\nu$-image.
As in~\cite{amarilli2015provenance},
we define the provenance of a $\overline{\Gamma}$-SATWA $A$ on~$\calT$, which
intuitively captures all possible results of evaluating $A$ on possible
valuations of~$\calT$:

\begin{definition}
  \label{def:satwaprov}
  The \emph{provenance} of a $\overline{\Gamma}$-SATWA $A$ on a $\Gamma$-tree $\calT$
  is the Boolean function $\phi$ defined on the nodes of~$T$ such
  that, for any valuation $\nu$ of~$\calT$, 
  $\nu(\phi) = 1$ iff $A$
  accepts $\nu(\calT)$.
\end{definition}

We then show that we can efficiently build provenance representations of
SATWAs on trees as stratified cycluits:

\begin{theorem}\label{thm:satwaprov}
  For any fixed alphabet $\Gamma$, given a $\overline{\Gamma}$-SATWA $A$ and a
  $\Gamma$-tree $\calT = \la T, \lambda \ra$, we can build 
  a stratified cycluit capturing the provenance of~$A$
  on~$\calT$
  in time $O(\card{A} \cdot \card{\calT})$.
\end{theorem}

The construction generalizes Proposition~3.1
of~\cite{amarilli2015provenance} from bNTAs and circuits to SATWAs and cycluits. 
The reason why we need cycluits rather than circuits is because two-way automata may loop
back on previously visited nodes.
To prove Theorem~\ref{thm:satwaprov}, we construct a cycluit $C^A_{\calT}$ as follows. 
For each node $w$ of $T$, we create an input node $g_w^{\gi}$, a $\lnot$-gate
$g_w^{\lnot \gi}$ defined as~$\NOT(g_w^{\gi})$, and an OR-gate $g_w^q$ for each state $q \in
Q$. Now for each $g_w^q$, for $b \in \{0, 1\}$, we consider the propositional
formula $\Delta(q, (\lambda(w),b))$, and we express it as a circuit that
captures this formula:
we let
$g_w^{q,b}$ be the output gate of that circuit, we
replace each variable~$q'$ occurring positively by an
OR-gate $\bigvee_{w^\prime \in \neigh(w)} g_{w^\prime}^{q'}$, 
and we replace each variable $q'$ occurring negatively by the gate~$g_w^{q'}$.
We then define $g^q_w$ as
  $\OR(\AND(g_w^{\gi},g_w^{q,1}), \AND(g_w^{\lnot
  \gi},g_w^{q,0}))$. Finally, we let the output gate of $C$ be $ g_r^{q_{\I}}$, where $r$ is the root of $T$, and $q_{\I}$ is the initial state of $A$.

It is clear that this process runs in linear time in $\card{A} \cdot \card{\calT}$.
The proof of Theorem~\ref{thm:satwaprov} then results from the following claim:

 \begin{lemma}\label{lem:goodprov}
	  The cycluit $C^A_{\calT}$ is a stratified cycluit capturing the provenance of $A$
  on~$\calT$.
 \end{lemma}

\begin{proof}
	We first show that $C \colonequals C^A_{\calT}$ is a stratified cycluit.
  Let $\strat$ be the stratification function
  of the $\overline{\Gamma}$-SATWA $A$ and let $\{0, \ldots, m\}$ be its range.
  We use $\strat$ to define $\strat'$ as the following function from the gates
  of~$C$ to $\{0, \ldots, m+1\}$:
\begin{itemize}
	\item For any input gate $g_w^\gi$, we set
          $\strat^\prime(g_w^\gi)\colonequals 0$ and
          $\strat^\prime(g_w^{\lnot \gi})\colonequals 1$.
        \item For an OR gate $g \colonequals \bigvee_{w^\prime \in \neigh(w)} g_{w^\prime}^{q'}$, we set
          $\strat'(g) \colonequals \strat(q') + 1$.
	\item For any gate $g_w^q$, we set $\strat^\prime(g_w^q) \colonequals \strat(q) +
          1$, and we set $\strat^\prime$ to the same value for the intermediate
          AND-gates used in the
          definition of~$g_w^q$, as well as for the gates in the two circuits that capture the
          transitions $\Delta(q, (\lambda(w), b))$ for $b \in \{0, 1\}$, except
          for the input gates of that circuit (i.e., gates of the form $\bigvee_{w^\prime \in \neigh(w)}
          g_{w^\prime}^{q'}$, which are covered by the previous point, or the
          gates of the form
          $g^{q'}_w$, which are covered by another application of that point).
\end{itemize}

Let us show that $\strat'$ is indeed a stratification function for~$C$. We first
  observe that it is the case that the gates in stratum zero are precisely the
  input gates. We then check the condition for the various possible wires:

\begin{itemize}
    \item $g^\gi_w \rightarrow g^{\lnot\gi}_w$: by construction, we have
      $\strat(g^\gi_w) < \strat'(g^{\lnot\gi}_w)$.
    \item $g \rightarrow g'$ where $g'$ is a gate of the form $g_w^q$ and $g$ is an
      intermediate AND-gate in the definition of a gate of the form $g_w^q$: by construction we
      have $\strat'(g) = \strat'(g')$, so in particular $\strat'(g) \leq
      \strat'(g')$.
    \item $g \rightarrow g'$ where $g'$ is an intermediate AND-gate in the
      definition of a gate of the form $g_w^q$, and $g$ is $g^\gi_w$ or
      $g^{\lnot\gi}_w$: by construction we have $\strat'(g) \in \{0, 1\}$ and
      $\strat'(g') \geq 1$, so $\strat'(g) \leq \strat'(g')$.
    \item $g \rightarrow g'$ where $g$ is a gate in a circuit capturing
      the propositional formula of some transition of $\Delta(q, \cdot)$ without
      being an input gate or a NOT-gate of this circuit, and $g'$ is also such a
      gate, or is an intermediate AND-gate in the definition of~$g_w^q$: then
      $g'$ cannot be a NOT-gate (remembering
      that the propositional formulas of transitions only have negations on
      literals), and
      by construction we have $\strat'(g) = \strat'(g')$.
    \item $g \rightarrow g'$ where $g$ is of the
	 form $\bigvee_{w^\prime \in \neigh(w)} g_{w^\prime}^{q}$, and $g'$ is
         a gate in a circuit describing $\Delta(q', \cdot)$ or an intermediate
         gate in the definition of $g^{q'}_w$.
         Then we have $\strat'(g) = \strat(q)$ and $\strat'(g') = \strat(q')$,
         and as $q$ occurs as a positive literal in a transition of~$q'$, by
         definition of $\strat$ being a transition function, we have $\strat(q)
         \leq \strat(q')$. Now we have $\strat'(g) = \strat(q)$ and $\strat'(g')
         = \strat'(q')$ by definition of~$\strat'$, so we deduce that
         $\strat'(g) \leq \strat'(g')$.
    \item $g \rightarrow g'$ where $g'$ is of the
	 form $\bigvee_{w^\prime \in \neigh(w)} g_{w^\prime}^{q'}$, and $g$ is
         one of the $g_{w'}^{q'}$. Then by definition of~$\strat'$ we have
         $\strat'(g) = \strat(q')$ and $\strat'(g') = \strat(q')$, so in
         particular $\strat'(g) \leq \strat'(g')$.
    \item $g \rightarrow g'$ where $g$ is a NOT-gate in a circuit
      capturing a propositional formula $\Delta(q',
      (\lambda(w), b))$, and $g$ is then necessarily a gate of the form
      $g^{q}_w$: then clearly $q'$ was
      negated in $\phi$ so we had $\strat(q) < \strat(q')$, and as by
      construction we have $\strat'(g) = \strat(q)$ and $\strat'(g') =
      \strat(q')$, 
      we deduce that $\strat'(g) < \strat'(g')$.
\end{itemize}

We now show that $C$ indeed captures the provenance of~$A$ on~$\la T,\lambda\ra$.
Let $\nu : T \to \{0,1\}$ be a Boolean valuation of the inputs of $C$, that we
extend to an evaluation $\nu' : C \to \{0,1\}$ of~$C$.
We claim the following \textbf{equivalence}: for all $q$ and $w$, there exists a run $\rho$ of $A$ on $\nu(T)$ \emph{starting at $w$} in state $q$ if and only if $\nu'(g_w^q)=1$.

We prove this claim by induction on the stratum $i = \strat(q)$ of~$q$. Up to adding
  an empty first stratum, we can make sure that the base case is vacuous.
For the induction step, we prove each implication separately.

  \subparagraph*{\qquad Forward direction.}
First, suppose that there exists a run $\rho = \la T_r, \lambda_r \ra$ starting at $w$ in state $q$, and
let us show that $\nu'(g^q_w) = 1$. We show by induction on the run (from
bottom to top) that for each node $y$ of the run labeled by a \emph{positive} state
$(q', w')$ we have $\nu'(g^{q'}_{w'}) = 1$, and for every node $y$ of the run
labeled by a \emph{negative} state $\lnot (q', w')$ we have $\nu'(g^{q'}_{w'})
= 0$. The base case concerns the leaves, where there are three possible
subcases:

\begin{itemize}
  \item We may have $\lambda_r(y)=(q^\prime, w')$ with $\strat(q^\prime) = i$, so
  that $\Delta(q^\prime,(\lambda(w'),\nu(w')))$ is tautological. In
    this case, $g_{w'}^{q^\prime}$ is defined as
    $\OR(\AND(g_{w'}^{\gi},g_{w'}^{q',1}),\AND(g_{w'}^{\lnot
    \gi},g_{w'}^{q',0}))$. Hence, we know that $\nu(g_{w'}^{q^\prime,\nu(w)}) =
    1$ because the circuit is also tautological, and depending on whether
    $\nu(w)$ is~$0$ or~$1$ 
    we know that $\nu(g_{w'}^{\lnot\gi})=1$ or $\nu(g_{w'}^{\gi})=1$, so this
    proves the claim.
		
  \item We may have $\lambda_r(y)=(q^\prime, {w'})$ with $\strat(q^\prime) = j$ for $j < i$. 
		By definition of the run $\rho$, this implies that there exists
                a run starting at ${w'}$ in state $q^\prime$. 
		But then, by the induction on the strata (using the forward
                direction of the \textbf{equivalence}), we must have $\nu(g_{w'}^{q^\prime})=1$.

  \item We may have $\lambda_r(y)=\lnot (q^\prime, {w'})$ with $\strat(q^\prime)
                = j$ for $j < i$. 
                Then by definition there exists no run starting at ${w'}$ in state $q^\prime$.
                Hence again by induction on the strata (using the backward
                direction of the \textbf{equivalence}), we have that $\nu(g_{w'}^{q^\prime})=0$.
\end{itemize}

For the induction case on the run, where $y$ is an internal node,
by definition of a run there is a subset $S = \{q_{P_1},\cdots,q_{P_n}\}$ of
positive literals and a subset $N = \{ \lnot q_{N_1}, \cdots, \lnot q_{N_m}\}$
of negative literals that satisfy $\phi_{\nu(w')} \defeq \Delta(q^\prime,(\lambda(w'),\nu(w')))$ such
that:

\begin{itemize}
  \item For all $q_{P_k} \in P$, there exists a child $y_k$ of~$y$ with $\lambda_r(y_k)=(q_{P_k}, {w'}_k)$
    where ${w'}_k \in \neigh({w'})$;
  \item For all $\lnot q_{N_k}
\in N$ there is a child $y'_k$
of $y$ with
$\lambda_r(y'_k)=\lnot (q_{N_k}, {w'})$.
\end{itemize}
    
    Then, by induction on the run, we
know that for all $q_{P_k}$ we have $\nu(g_{{w'}_k}^{q_{P_k}})=1$ and for all
$\lnot q_{N_k}$ we have
$\nu(g_{{w'}}^{q_{N_k}})=0$. 
Let us show that we have
$\nu(g_{w'}^{q'})=1$, which would finish the induction case on the run.
There are two cases:
either $\nu(w') = 1$ or $\nu(w') = 0$. In the
first case, remember that the first input of the OR-gate $g_{w'}^{q'}$ is an AND-gate of
$g^{\gi}_{w'}$ and the output gate $g_{w'}^{q',1}$ of a circuit coding $\phi_{1}$ on inputs including
the $g^{q_{P_k}}_{w'_k}$ and $g^{q_{N_k}}_{w'}$. We have $\nu(g^{\gi}_{w'}) = 1$ because $\nu(w') = 1$,
and the second gate ($g_{w'}^{q',1}$) evaluates to~$1$ by construction of the circuit, as
witnessed by the Boolean valuation of the $g^{q_{P_k}}_{w'_k}$ and $g^{q_{N_k}}_{w'}$.
In the second case we follow the same reasoning but with the second input
of~$g^{q'}_{w'}$ instead, which is an AND-gate on
$g^{\lnot \gi}_{w'}$ and a circuit coding $\phi_{0}$.

By induction on the run, the claim is proven, and applying it to the root of the
run concludes the proof of the first direction of the \textbf{equivalence} (for the
induction step of the induction on strata).

  \subparagraph*{\qquad Backward direction.}
We now prove the converse implication for the induction step of the induction on
strata, i.e., letting $i$ be the current stratum, for every node $w$ and state
$q$ with $\strat(q) = i$, if $\nu(g^q_w) = 1$ then there exists a run $\rho$
of~$A$ starting at~$w$. From the definition of the stratification
function~$\strat'$ of the cycluit from~$\strat$, we have $\strat'(g^q_w) =
\strat(q)+1$, so as $\nu(g^q_w) = 1$ we know that $\nu_{i+1}(g^q_w) = 1$, where
$\nu_{i+1}$ is the $i+1$-th
stratum evaluation of~$C$ (remember Definition~\ref{def:strateval}).
By induction hypothesis on the strata, we know from the \textbf{equivalence} that, for any $j \leq i$, for any
gate $g^{q''}_{w''}$ of $C$ with $\strat(g^{q''}_{w''}) = j$, we have $\nu_j(g^{q''}_{w''}) = 1$ iff 
there exists a run $\rho$ of $A$ on $\nu(T)$ \emph{starting at $w''$}
in state $q''$.

Recall that the definition of $\nu_{i+1}$ according to
Definition~\ref{def:strateval} proceeds in three steps. Initially, we
fix the value in $\nu_{i+1}$ of gates of
lower strata, so we can then conclude by induction hypothesis on the strata. We
then set the value of all NOT-gates in~$\nu_{i+1}$, but these cannot be of the form
$g^{q'}_{w'}$ so there is nothing to show. Last, we evaluate all other gates
with Algorithm~\ref{alg:semantics-monotone}. We then show our claim by an
induction on the iteration in the application of
Algorithm~\ref{alg:semantics-monotone} for $\nu_{i+1}$ where the gate $g^q_w$ was
set to~$1$. The base case, where $g^q_w$ was initially true, was covered in the
beginning of this paragraph.

For the induction step on the application of
Algorithm~\ref{alg:semantics-monotone},
when a gate $g^{q'}_{w'}$ is set to true by $\nu_{i+1}$, as $g^{q'}_{w'}$ is an
OR-gate by construction, from the workings of Algorithm~\ref{alg:semantics-monotone}, there are two possibilities:
either its input AND-gate that includes $g^{\gi}_{w'}$ was true,
or its input AND-gate that includes $g^{\neg\gi}_{w'}$ was true. We prove the
first case, the second being analogous. From the fact that $g^{\gi}_{w'}$ is
true, we know that $\nu(w') = 1$. Consider the other input gate to that AND
gate, which is the output gate of a circuit $C'$ reflecting 
$\phi \defeq \Delta(q^\prime,(\lambda(w'),\nu(w')))$, with the input gates
adequately substituted. We consider the value by $\nu_{i+1}$ of the gates that
are used as input gates of $C'$ in the construction of~$C$ (i.e., OR-gates, in
the case of variables that occur positively, or directly $g^{q''}_{w'}$-gates,
in the case of variables that occur negatively). By construction of
$C'$, the corresponding Boolean valuation $\nu'$ is a witness to the satisfaction of $\phi$. By induction hypothesis on
the strata (for the negated inputs to $C'$; and for the non-negated inputs to
$C'$ which are in a lower stratum) and on the step at which the gate was
set to true by Algorithm~\ref{alg:semantics-monotone} (for the inputs in the same stratum, which must be positive), the
valuation of these inputs reflects the existence of the corresponding runs.
Hence, we can assemble these (i.e., a leaf node in the first two cases, a run in
the third case) to obtain a run starting at $w'$ for state $q'$
using the Boolean valuation $\nu'$ of the variables of~$\phi$; this valuation satisfies
$\phi$ as we have argued.

This concludes the two inductions of the proof of the \textbf{equivalence} for the
induction step of the induction on strata, which concludes the proof of Theorem~\ref{thm:satwaprov}.
\end{proof}

Note that the proof can be easily modified to make it work for standard
alternating two-way automata rather than our isotropic automata.

\subparagraph*{Proving Theorem~\ref{thm:mainprov}.}
We are now ready to conclude the proof of our main provenance construction result, i.e., Theorem~\ref{thm:mainprov}.
We do so by explaining how our
provenance construction for $\overline{\Gamma}$-SATWAs can be used to compute
the provenance of a CFG-Datalog query on a treelike instance. This is again
similar to~\cite{amarilli2015provenance}.

Recall the definition of tree encodings from
Section~\ref{sec:existing}, and the
definition of the alphabet~$\Gamma^k_\sigma$. To represent the dependency of
automaton runs on the presence of individual facts, we will be working with
$\overline{\Gamma^k_\sigma}$-trees, where the Boolean annotation on a node~$n$ indicates
whether the fact coded by~$n$ (if any) is present or absent. The semantics
is that we map back the result to $\Gamma^k_\sigma$ as follows:

\begin{definition}
  We define the mapping $\epsilon$ from $\overline{\Gamma^k_\sigma}$ to
  $\Gamma^k_\sigma$ by:
  \begin{itemize}
    \item $\epsilon((d, s), 1)$ is just $(d, s)$, indicating that the fact
      of~$s$ (if any) is kept;
    \item $\epsilon((d, s), 0)$ is $(d, \emptyset)$, indicating that the fact of
      $s$ (if any) is removed.
  \end{itemize}

  We abuse notation and also see $\epsilon$ as a mapping from
  $\overline{\Gamma^k_\sigma}$-trees to $\Gamma^k_\sigma$-trees by applying it
  to each node of the tree.
\end{definition}

As our construction of provenance applies to automata on
$\overline{\Gamma^k_\sigma}$, we show the following easy \emph{lifting lemma}
(generalizing Lemma~3.3.4 of~\cite{amarilli2016leveraging}):

\begin{lemma}\label{lem:lifting}
  For any $\Gamma^k_\sigma$-SATWA $A$, we can compute in linear time a
  $\overline{\Gamma^k_\sigma}$-SATWA $A'$ such that, for any
  $\overline{\Gamma^k_\sigma}$-tree $E$, we have that $A'$ accepts $E$ iff $A$
  accepts $\epsilon(E)$.
\end{lemma}

\begin{proof}
  The proof is exactly analogous to that of  Lemma~3.3.4
  of~\cite{amarilli2016leveraging}.
\end{proof}

We are now ready to conclude the proof of our main provenance result
(Theorem~\ref{thm:mainprov}):

\begin{proof}[Proof of Theorem~\ref{thm:mainprov}]
  Given the program $P$ and instance $I$, use
Theorem~\ref{thm:maintheorem} to compute in
FPT-linear time in $\card{P}$ a $\Gamma^k_\sigma$-SATWA $A$ that tests $P$ on instances of treewidth $\leq k_{\text{I}}$, for $k_{\text{I}}$ the treewidth bound.
Compute also in FPT-linear time a tree encoding $\la E, \lambda \ra$ of the instance $I$ (i.e., a
$\Gamma^k_\sigma$-tree), using Theorem~\ref{lem:getencoding}.
Lift the  $\Gamma^k_\sigma$-SATWA $A$ in linear time using
Lemma~\ref{lem:lifting} to a $\overline{\Gamma^k_\sigma}$-SATWA~$A'$, and use Theorem~\ref{thm:satwaprov} on $A'$ and $\la E, \lambda \ra$ to compute in
FPT-bilinear time a stratified cycluit $C'$ that captures the provenance of~$A'$
on~$\la E, \lambda \ra$: the inputs of~$C'$ correspond to the nodes of~$E$. 
  Let $C$ be obtained
from $C'$ in linear time by changing the inputs of $C'$ as follows: those which
correspond to nodes $n$ of~$\la E, \lambda \ra$ containing a fact (i.e., with label $(d, s)$ for
$\card{s} = 1$) are renamed to be an input gate that stands for the fact of~$I$
coded in this node; the nodes $n$ of~$\la E, \lambda \ra$ containing no fact are replaced by a
0-gate, i.e., an OR-gate with no inputs. Clearly, $C$ is still a stratified
Boolean cycluit,
and $C_\inp$ is exactly the set of facts of~$I$.

All that remains to show is that $C$ captures the provenance of $P$ on~$I$ in
  the sense of Definition~\ref{def:provenance}. 
To see why this is the case, consider an arbitrary
Boolean valuation $\nu$ mapping the facts of~$I$ to $\{0, 1\}$, and call $\nu(I) \defeq
\{F \in I \mid \nu(F) = 1\}$. We must show that $\nu(I)$ satisfies $P$ iff
$\nu(C) = 1$. By construction of~$C$, it is obvious that $\nu(C) = 1$ iff
$\nu'(C') = 1$, where $\nu'$ is the Boolean valuation of $C_\inp$ defined by $\nu'(n) =
\nu(F)$ when $n$ codes some fact $F$ in $\la E, \lambda \ra$, and $\nu'(n) = 0$ otherwise. By
definition of the provenance of $A'$ on~$\la E, \lambda \ra$, we have $\nu'(C') = 1$ iff $A'$
accepts $\nu'(\la E, \lambda \ra)$, that is, by definition of lifting, iff $A$ accepts
$\epsilon(\nu'(\la E, \lambda \ra))$.
Now all that remains to observe is that $\epsilon(\nu'(\la E, \lambda \ra))$ is precisely a tree
encoding of the instance $\nu(I)$: this is by definition of $\nu'$ from $\nu$,
and by definition of our tree encoding scheme.
Hence, by definition of $A$ testing $P$, the tree
$\epsilon(\nu'(\la E, \lambda \ra))$ is accepted by $A$ iff $\nu(I)$ satisfies $P$. This finishes
the chain of equivalences, and concludes the proof of
Theorem~\ref{thm:mainprov}.
\end{proof}

This concludes the presentation of our provenance results.

\section{Proof of Translation}
\label{sec:proof}
In this section, we prove our main technical theorem, Theorem~\ref{thm:maintheorem}, which we recall here:
\makeatletter
\begin{axp@theoremrp}[\ref{thm:maintheorem}]
	Given a \cfggnd program~$P$ of body size
  $\kp$ and $\ki \in \NN$,
  we can build in FPT-linear time in~$|P|$
  (parameterized by~$\kp, \ki$)
  a SATWA $A_P$ testing $P$ on instances of treewidth $\leq \ki$.
\end{axp@theoremrp}
\makeatother
We then explain at the end of the section how this can be extended to full CFG-Datalog (i.e., with negative intensional predicates not being necessarily guarded in rule bodies).

\subsection{Guarded-Negation Case}	
	First, we introduce some useful notations to deal with
        valuations of variables as constants of the encoding alphabet.
        Recall that $\calD_{\ki}$ is the domain of elements
        for treewidth $\ki$, used to define the alphabet $\Gamma_{\sigmae}^{\ki}$ of tree
        encodings of width~$\ki$.

	\begin{definition}
          \label{def:partialval}
          Given a tuple $\mathbf{x}$ of variables, a \emph{partial valuation} of~$\mathbf{x}$ is a function $\nu$ from~$\mathbf{x}$ to $\calD_{\ki} \sqcup \{?\}$.
		The set of \emph{undefined} variables of $\nu$ is
                $U(\nu)=\{x_j \mid \nu(x_j)=\mathord{?}\}$: we say that the variables of~$U(\nu)$ are \emph{not defined} by~$\nu$, and the other variables are \emph{defined} by~$\nu$.

		A \emph{total valuation} of $\mathbf{x}$ is a partial valuation $\nu$ of $\mathbf{x}$ such that $U(\nu) = \emptyset$.
                We say that a valuation $\nu'$ \emph{extends} another
                valuation~$\nu$ if the domain of~$\nu'$ is a superset of that
                of~$\nu$, and if all variables defined by $\nu$ are defined by~$\nu'$
                and are mapped to the same value.
                For $\mathbf{y} \subseteq \mathbf{x}$, we say that $\nu$ is \emph{total on~$\mathbf{y}$} if its restriction to~$\mathbf{y}$ is a total valuation.

                For any two partial valuations $\nu$ of $\mathbf{x}$ and $\nu'$
                of $\mathbf{y}$,
                if we have $\nu(z) = \nu'(z)$ for all $z$ in $(\mathbf{x} \cap
                \mathbf{y})
                \setminus (U(\nu) \cup U(\nu'))$, then we write $\nu \cup \nu'$ for
                the valuation on $\mathbf{x} \cup \mathbf{y}$ that maps every $z$ to $\nu(z)$ or
                $\nu'(z)$ if one is defined, and to ``?'' otherwise.

		When $\nu$ is a partial valuation of $\mathbf{x}$ with $\mathbf{x} \subseteq \mathbf{x^\prime}$ 
		and we define a partial valuation $\nu^\prime$ of
                $\mathbf{x^\prime}$ with $\nu^\prime \colonequals \nu$, we mean that $\nu^\prime$
		is defined like~$\nu$ on $\mathbf{x}$ and is undefined on $\mathbf{x^\prime} \setminus \mathbf{x}$.
	\end{definition}
	\begin{definition}
                \label{def:hom}
		Let $\mathbf{x}$ and $\mathbf{y}$ be two tuples of
                variables of same arity (note that some variables
                of~$\mathbf{x}$ may be repeated, and likewise for~$\mathbf{y}$).
                Let $\nu : \mathbf{x} \to \calD_{\ki}$
		be a total valuation of~$\mathbf{x}$. We define $\mathrm{Hom}_{\mathbf{y},\mathbf{x}}(\nu)$ to be 
                the (unique) homomorphism from the tuple
                $\mathbf{y}$ to the tuple $\nu(\mathbf{x})$, if
                such a homomorphism exists; otherwise,
                $\mathrm{Hom}_{\mathbf{y},\mathbf{x}}(\nu)$ is~$\mnull$.
	\end{definition}
	
  The rest of this section proves Theorem~\ref{thm:maintheorem}
    in two steps. First, we build a SATWA~$A'_P$ and we prove that $A'_P$~tests~$P$ on instances of treewidth $\leq \ki$; however,
    the construction of $A'_P$ that we present is not FPT-linear. Second, we
    explain how to modify the construction to construct an equivalent
    SATWA~$A_P$ while respecting the FPT-linear time bound.

  \subparagraph*{Construction of ${A_P^\prime}$.} We formally construct the SATWA $A'_P$
    by describing its states and transitions.
	First, for every extensional atom $S(\mathbf{x})$ appearing in the body
        of a rule of $P$ and for every partial valuation~$\nu$ of 
	$\mathbf{x}$, we introduce a state $q_{S(\mathbf{x})}^{\nu}$.
        For every node~$n$, we want
        $A_P^\prime$ to have a run starting at node~$n$ in state
        $q_{S(\mathbf{x})}^{\nu}$
        iff
        we can start at node~$n$, navigate the tree encoding 
        while building a total valuation $\nu'$ that extends $\nu$, and reach a
        node~$n'$ where $S(\nu'(\mathbf{x}))$ holds. However, remember that the
        same element name in the tree encoding may refer to different elements
        in the instance. Hence, we must ensure that the elements in the image
        of~$\nu$ still decode to the same element in~$n'$ as they did in~$n$. To
        ensure this, we forbid $A^\prime_P$ from leaving the occurrence
	subtree of the values in the image of~$\nu$, which we call the \emph{allowed subtree}.
        We now define the transitions
        needed to implement this.

	Let $(d,s) \in \Gamma_{\sigmae}^{\ki}$ be a symbol; we have the
        following transitions:

	\begin{itemize}
		\item If there is an $x_j \in \mathbf{x}$ such that $\nu(x_j)\neq\mathord{?}$ (i.e., $x_j$ is defined by $\nu$) and $\nu(x_j) \notin d$, then $\Delta(q_{S(\mathbf{x})}^{\nu},(d,s)) \colonequals \false$.
			This is to prevent the automaton from leaving the
                        allowed subtree.
		\item Else if $\nu$ is not total, then
		$\Delta(q_{S(\mathbf{x})}^{\nu},(d,s)) \colonequals
		q_{S(\mathbf{x})}^{\nu} \lor \bigvee\limits_{a \in d, x_j \in U(\nu)} q_{S(\mathbf{x})}^{\nu \cup \{x_j\mapsto a\}}$.
		That is, either we continue navigating in the same state (but remember that the automaton may move to any neighbor node), or we guess a value for some undefined variable.
              \item Else if $\nu$ is total but $s \neq S(\nu(\mathbf{x}))$,
                  then
		$\Delta(q_{S(\mathbf{x})}^{\nu},(d,s)) \colonequals q_{S(\mathbf{x})}^{\nu}$: if the fact $s$ of the node is not a match, then we continue searching. 
              \item Else, the only remaining possibility is that $\nu$ is total and that $s = S(\nu(\mathbf{x}))$, in which case we set  $\Delta(q_{S(\mathbf{x})}^{\nu},(d,s)) \colonequals \true$, i.e., we have found a node containing the desired fact.
	\end{itemize}
	
Let $r$ be a rule of $P$ and $\calA$ be a subset of the literals in the body of $r$. 
	We write $\vars(\calA)$ the set of variables that appear in some atom of $\calA$.
	For every rule $r$ of $P$, for every subset $\calA$ of the literals in the
	body of $r$, and for every partial valuation $\nu$ of $\vars(\calA)$ that defines all the variables that are also in the head of $r$,
	we introduce a state $q_r^{\nu,\calA}$. 
	This state is intended to verify the literals in $\calA$ with the partial valuation $\nu$.
	We will describe the transitions for those states later.
	
	For every intensional predicate $R(\mathbf{x})$ appearing in a rule of $P$ and partial valuation $\nu$ of $\mathbf{x}$,
	we have a state $q_{R(\mathbf{x})}^{\nu}$. This state is intended to
        verify $R(\mathbf{x})$ with a total extension of $\nu$.
	Let $(d,s) \in \Gamma_{\sigmae}^{\ki}$ be a symbol; we have
        the following transitions:

	\begin{itemize}
		\item If there is a $j$ such that $x_j$ is defined by $\nu$ and $\nu(x_j) \notin d$,
                  then $\Delta(q_{R(\mathbf{x})}^{\nu},(d,s)) \colonequals \false$.
                  This is again in order to prevent the automaton from leaving the
                  allowed subtree.
		\item Else if $\nu$ is not total, then
		$\Delta(q_{R(\mathbf{x})}^{\nu},(d,s)) \colonequals
		q_{R(\mathbf{x})}^{\nu} \lor \bigvee\limits_{a \in d, x_j \in U(\nu)} q_{R(\mathbf{x})}^{\nu \cup \{x_j\mapsto a\}}$.
		Again, either we continue navigating in the same state, or we guess a value for some undefined variable.
		\item Else (in this case $\nu$ is total), $\Delta(q_{R(\mathbf{x})}^{\nu},(d,s))$ is defined as the
                  disjunction of all the $q_r^{\nu^\prime,\calA}$ for
                  each rule~$r$ such that
		the head of~$r$ is $R(\mathbf{y})$, $\nu^\prime \colonequals
                \mathrm{Hom}_{\mathbf{y},\mathbf{x}}(\nu)$ is not~$\mnull$ and
		$\calA$ is the set of all literals in the body of $r$. Notice that because $\nu$ is total on~$\mathbf{x}$, $\nu'$ is also total on~$\mathbf{y}$.
		This transition simply means that we need to chose an
                appropriate rule to derive $R(\mathbf{x})$.
		We point out here that these transitions are the ones that make the construction quadratic instead of linear in $\card{P}$, but this will
		be handled later.
	\end{itemize}

	It is now time to describe transitions for the states
        $q_r^{\nu,\calA}$. Let $(d,s) \in \Gamma_{\sigmae}^{\ki}$, then:

	\begin{itemize}
		\item If there is a variable $z$ in $\calA$ such that $z$ is defined by $\nu$ and $\nu(z) \notin d$, then $\Delta(q_r^{\nu,\calA},(d,s)) \colonequals \false$.
			Again, this is to prevent the automaton from leaving the
                	allowed subtree.
		\item Else, if $\calA$ contains at least two literals, then $\Delta(q_r^{\nu,\calA},(d,s))$ is defined as a disjunction of $q_r^{\nu,\calA}$ and of $\Bigg[$~a disjunction over
		all the non-empty sets $\calA_1,\calA_2$ that partition $\calA$ of $\bigg[$a disjunction over all the total valuations $\nu^\prime$
                of $U(\nu)\cap \vars(\calA_1) \cap \vars(\calA_2)$ with values in $d$ of
		$\big[q_r^{\nu \cup \nu^\prime,\calA_1} \land q_r^{\nu \cup \nu^\prime,\calA_2}\big]\bigg]\Bigg]$.
		This transition means that we allow to partition in two
                the literals that need to be verified, and for each class of the partition we launch one run
		that will have to prove the literals of that class. 
		In doing so, we
                have to take care that the two runs will build valuations
                that are consistent.
		This is why we fix the value of the variables that they have in common with a total valuation $\nu'$.
		\item Else, if $\calA = \{T(\mathbf{y})\}$ where $T$ is an
                  extensional or an intensional relation, then $\Delta(q_r^{\nu,\calA},(d,s)) \colonequals q_{T(\mathbf{y})}^{\nu}$.
		\item Else, if $\calA = \{\lnot
                  R^\prime(\mathbf{y})\}$ where $R^\prime$ is an intensional
                  relation, and if $|\mathbf{y}| = 1$, and if $\nu(y)$ is undefined (where we write $y$ the one element of~$\mathbf{y}$), then 
			$\Delta(q_r^{\nu,\calA},(d,s)) \colonequals q_r^{\nu,\calA} \lor \bigvee_{a \in d} q_r^{\nu \cup \{y\mapsto a\},\calA}$.
		\item Else, if $\calA = \{\lnot R^\prime(\mathbf{y})\}$ where $R^\prime$ is an intensional relation, 
		then we will only define the transitions in the case where $\nu$ is total on~$\mathbf{y}$, in which case we set
		$\Delta(q_r^{\nu,\calA},(d,s)) \colonequals \lnot q_{R^\prime(\mathbf{y})}^{\nu}$.
                It is sufficient to define the transitions in this case, because $q_r^{\nu,\{\lnot R^\prime(\mathbf{y})\}}$ can only be reached
                if $\nu$ is total on~$\mathbf{y}$. Indeed, if $|\mathbf{y}| = 1$, then $\nu$ must be
                total on~$\mathbf{y}$ because we would have applied the previous bullet
                point otherwise. If $|\mathbf{y}| > 1$, the only way we could have reached the state $q_r^{\nu,\{\lnot R^\prime(\mathbf{y})\}}$
                is by a sequence of transitions involving $q_r^{\nu_0,\calA_0},\ldots,q_r^{\nu_m,\calA_m}$, where $\calA_0$ are all the literals in the body of $r$,
		$\calA_m$ is $\calA$ and $\nu_m$ is $\nu$.
		We can then see that, during the partitioning process,
		$\lnot R^\prime(\mathbf{y})$ must have been separated from all
                the (positive) atoms that formed its guard (recall the
                definition of \cfggnd), hence
		all its variables have been assigned a valuation.
	\end{itemize}

Finally, the initial state of $A_P^\prime$ is
$q_{\text{Goal}}^{\emptyset}$.

\medskip

We describe the stratification function $\strat^\prime$ of $A^\prime_P$. Let $\strat$ be that of $P$.
Observe that we can assume without loss of generality that the first stratum of $\strat$ (i.e., relations $R$ with $\strat(R)=1$) contains exactly all the extensional relations.
For any state $q$ of the form $q_{T(\mathbf{x})}^\nu$ or
$q_r^{\nu,\calA}$ with~$r$ having as head relation $T$ ($T$ begin extensional or intensional), then $\strat^\prime(q)$ is 
defined to be $\strat(T) -1$.
Notice that this definition ensures that only the states corresponding to extensional relations are in the first stratum of $\strat'$.
It is then clear from the transitions that $\strat^\prime$ is a valid stratification function for $A^\prime_P$.

\medskip

    As previously mentioned, the construction of $A'_P$ is not
    FPT-linear, but we will explain at the end of the proof how to construct in FPT-linear time a SATWA
    $A_P$ equivalent to~$A'_P$.

\subparagraph*{${A^\prime_P}$ tests ${P}$ on instances of treewidth
${\leq \ki}$.}
To show this claim, let $\la T, \lambda_E \ra$ be a $(\sigmae,\ki)$-tree encoding. 
Let $I$ be the instance obtained by decoding $\la T, \lambda_E \ra$;
we know that $I$ has treewidth $\leq \ki$ and that we can define from $\la T, \lambda_E \ra$ a tree 
decomposition $\la T, \dom \ra$ of~$I$ whose underlying tree is also~$T$.
For each node $n \in T$, let $\dec_n : \calD_{\ki} \to \dom(n)$ be the function that decodes the 
elements in node~$n$ of the encoding to the elements of~$I$ that are in the corresponding bag of the tree decomposition,
and let $\enc_n : \dom(n) \to \calD_{\ki}$ be the inverse function that encodes back the elements,
so that we have $\dec_n \circ \enc_n = \enc_n \circ \dec_n =
\mathrm{Id}$.
We will denote elements of $\calD_{\ki}$ by $a$ and elements in the domain of~$I$ by $c$.

We recall some properties of tree decompositions and tree encodings:
\begin{property}
\label{propTreeEnc1}
Let $n_1,n_2$ be nodes of $T$ and $a \in \calD_{\ki}$ be an (encoded)
  element that appears in the $\lambda_E$-image of $n_1$ and $n_2$. Then the element $a$ appears in the $\lambda_E$-image of every node in the path from $n_1$ to $n_2$ if and only if $\dec_{n_1}(a) = \dec_{n_2}(a)$.
\end{property}

\begin{property}
\label{propTreeEnc2}
	Let $n_1,n_2$ be nodes of $T$ and $c$ be an element of $I$ that appears
        in $\dom(n_1) \cap \dom(n_2)$.
	Then for every node $n'$ on the path from $n_1$ to $n_2$, $c$ is also in
        $\dom(n')$, and moreover $\enc_{n'}(c)=\enc_{n_1}(c)$.
\end{property}

We start with the following lemma about extensional facts:
\begin{lemma}
\label{lem:extensional}
  For every extensional relation $S$, node $n \in T$, variables $\mathbf{y}$, and partial valuation $\nu$ of $\mathbf{y}$, there exists a run $\rho$ of $A^\prime_P$ starting at node $n$
in state~$q_{S(\mathbf{y})}^{\nu}$ if and only if there exists
a fact $S(\mathbf{c})$ in~$I$ 
such that we have $\dec_n(\nu(y_j))=c_j$
  for every $y_j$ defined by $\nu$.
  We call this a match $\mathbf{c}$ of~$S(\mathbf{y})$ in~$I$ that is \emph{compatible with $\nu$ at node $n$}.
		
\end{lemma}
\begin{proof}
  We prove each direction in turn.

  \subparagraph*{\qquad Forward direction.}
	Suppose there exists a run $\rho$ of $A^\prime_P$ starting at node $n$
in state~$q_{S(\mathbf{y})}^{\nu}$. First, notice that by design of the
  transitions starting in a state of that form, 
  states appearing in the labeling of the run can only be of the form
  $q_{S(\mathbf{y})}^{\nu\prime}$ for an extension $\nu'$ of~$\nu$.
We will show by induction on the run that for every node $\pi$ of the run labeled by $(q_{S(\mathbf{y})}^{\nu\prime},m)$, 	
there exists $\mathbf{c'}$ such that $S(\mathbf{c'}) \in I$ and
  $\mathbf{c'}$ is compatible with $\nu'$ at node $m$. This will conclude
  the proof of the forward part of the lemma, by taking $m=n$.

The base case is when $\pi$ is a leaf of $\rho$.
	The node $\pi$ is then labeled by $(q_{S(\mathbf{y})}^{\nu^\prime}, m)$ such that $\Delta(q_{S(\mathbf{y})}^{\nu^\prime},\lambda_E(m)) = \true$.
	Let $(d,s) = \lambda_E(m)$. 
	By construction of the automaton we have that $\nu^\prime$ is total and $s = S(\nu^\prime(\mathbf{y}))$.
	We take $\mathbf{c'}$ to be $\dec_m(\nu^\prime(\mathbf{y}))$, which satisfies the compatibility condition by definition
	and is such that $S(\mathbf{c'}) =
        S(\dec_m(\nu^\prime(\mathbf{y}))) = \dec_m(s) \in I$.

	When $\pi$ is an internal node of $\rho$, its label is written,
        again, as 
        $(q_{S(\mathbf{y})}^{\nu^\prime}, m)$.
        By definition of the
        transitions of the automaton, we have
	\[\Delta(q_{S(\mathbf{y})}^{\nu^\prime},(d,s)) =
	q_{S(\mathbf{y})}^{\nu^\prime} \lor \bigvee\limits_{a \in d, y_j \in U(\nu^\prime)} q_{S(\mathbf{y})}^{\nu^\prime \cup \{y_j\mapsto a\}}\]
        Hence, the node $\pi$ has at least one child $\pi'$, the second component
        of the label of $\pi'$ is some $m' \in \neigh(m)$, and we have two
        cases depending on the first component of its label (i.e., the state):

	\begin{itemize}
          \item $\pi'$ may be labeled by $(q_{S(\mathbf{y})}^{\nu^\prime}, m^\prime)$.
                Then by induction on the run there exists $\mathbf{c''}$ such that $S(\mathbf{c''}) \in I$ and
                $\mathbf{c''}$ is compatible with $\nu'$ at node $m'$.
                We take $\mathbf{c'}$ to be $\mathbf{c''}$, so
                that we only need to check the compatibility
                condition, i.e., that for every $y_j$ defined
                by $\nu'$, we have $\dec_m(\nu'(y_j)) = c_j = \dec_{m'}(\nu'(y_j))$.
                This is true by Property~\ref{propTreeEnc1}. Indeed, 
                for every $y_j$ defined by $\nu'$, we must have $\nu'(y_j) \in
                m'$, otherwise $\pi'$ would have a label that cannot occur in a
                run (because this would mean that we have escaped the allowed subtree).
        \item $\pi^\prime$ is labeled by $(q_{S(\mathbf{y})}^{\nu^\prime \cup \{y_j\mapsto a\}}, m^\prime)$
          for some $a \in d$ and for some $y_j \in U(\nu^\prime)$. 
                Then by induction on the run there exists $\mathbf{c''}$ such that $S(\mathbf{c''}) \in I$ and
                $\mathbf{c''}$ is compatible with $\nu^\prime \cup \{y_j\mapsto a\}$ at node $m'$.
                We take $\mathbf{c'}$ to be $\mathbf{c''}$, which again satisfies the compatibility condition thanks to Property~\ref{propTreeEnc1}.	
\end{itemize}

  \subparagraph*{\qquad Backward direction.}
Now, suppose that there exists $\mathbf{c}$ such that
$S(\mathbf{c}) \in I$ and $\mathbf{c}$ is compatible with $\nu$ at node~$n$.
The fact $S(\mathbf{c})$ is encoded somewhere in
   $\la T, \lambda_E \ra$, so there exists a node $m$ such that, letting $(d,s)$ be $\lambda_E(m)$, we have $\dec_m(s)=S(\mathbf{c})$.
Let $n = m_1, m_2,\ldots,m_p = m$ be the nodes on the path from $n$ to $m$, and $(d_i,s_i)$ be $\lambda_E(m_i)$ for $1 \leq i \leq p$.
By compatibility, for every $y_j$ defined by $\nu$ we have $\dec_n(\nu(y_j)) = c_j$. But $\dec_n(\nu(y_j)) \in \dom(n)$ and $c_j \in \dom(m)$ so by Property~\ref{propTreeEnc2}, for every $1 \leq i \leq p$ we have
$c_j \in \dom(m_i)$ and $\enc_{m_i}(c_j) =\enc_{n}(c_j) = \enc_{n}(\dec_n(\nu(y_j))) = \nu(y_j)$, so that $\nu(y_j) \in d_i$.
We can then construct a run $\rho$ starting at node $n$ in state
  $q_{S(\mathbf{y})}^{\nu}$ as follows. The root $\pi_1$ is labeled by $(q_{S(\mathbf{y})}^{\nu},n)$, and
for every $2 \leq i \leq p$, $\pi_i$ is the unique child of $\pi_{i-1}$ and is labeled by $(q_{S(\mathbf{y})}^{\nu}, m_i)$. 
This part is valid because we just proved that for every $i$, there is no $j$ such that $y_j$ is defined by $\nu$ and $\nu(y_j) \notin d_j$.
Now from $\pi_m$, we continue the run by staying at node $m$ and building up the
valuation, until we reach a total valuation $\nu_{\f}$ such that $\nu_\f(\mathbf{y}) = \enc_m(\mathbf{c})$.
Then we have $s = S(\nu_\f(\mathbf{y}))$ and the transition is $\true$, which
completes the definition of the run.
\end{proof}

The preceding lemma concerns the base case of extensional relations. We
now prove a similar \emph{equivalence lemma} for all relations (extensional or intensional).
This lemma allows us to conclude the correctness proof, by applying it to the
$\mathrm{Goal}()$ predicate and to the root of the tree-encoding.

\begin{lemma}
  \label{lem:correctness}
For every relation $R$, node $n \in T$ and partial valuation $\nu$ of $\mathbf{x}$, there exists a run $\rho$ of $A^\prime_P$ starting at node $n$
in state $q_{R(\mathbf{x})}^{\nu}$ if and only if there exists $\mathbf{c}$ such that $R(\mathbf{c}) \in P(I)$ and $\mathbf{c}$ is compatible
with $\nu$ at node $n$ (i.e., we have $\dec_n(\nu(x_j)) = c_j$ for every $x_j$ defined by $\nu$). 
\end{lemma}

\begin{proof}
We will prove this equivalence by induction on the stratum $\strat(R)$ of
  the relation $R$. 
The base case ($\strat(R)=0$, so $R$ is an extensional relation) was shown in Lemma~\ref{lem:extensional}.
For the inductive case, where $R$ is an intensional relation, we prove each direction separately.

\subparagraph*{\qquad Forward direction.}
First, suppose that there exists a run $\rho$ of $A^\prime_P$ starting at node $n$
in state $q_{R(\mathbf{x})}^{\nu}$. We show by induction on the run (from bottom to top) that for every node $\pi$ of the run the following implications hold:

	\begin{compactenum}[(i)]
\item \label{case:positive-atom} 
  If $\pi$ is labeled with $(q_{R^\prime(\mathbf{y})}^{\nu^\prime}, m)$, then there exists $\mathbf{c}$ such that
		$R^\prime(\mathbf{c}) \in P(I)$ and $\mathbf{c}$ is
                compatible with $\nu'$ at node~$m$.
              \item \label{case:negative-atom}
                If $\pi$ is labeled with $\lnot (q_{R^\prime(\mathbf{y})}^{\nu^\prime}, m)$, then
		$R^\prime(\dec_m(\nu^\prime(\mathbf{y}))) \notin P(I)$ (remembering that in this case $\nu'$ must be total, thanks to the fact that negations are guarded in rule bodies).
              \item \label{case:rule} If $\pi$ is labeled with $(q_r^{\nu^\prime,\calA},m)$, then
		      there exists a mapping $\mu : \vars(\calA) \to \mathrm{Dom}(I)$
		      that is compatible with $\restr{\nu'}{\vars(\calA)}$ at node~$m$ and such that:
		\begin{itemize}
			\item For every positive literal $S(\mathbf{z})$ in $\calA$, then $S(\mu(\mathbf{z})) \in P(I)$.
			\item For every negative literal $\lnot S(\mathbf{z})$ in $\calA$, then $S(\mu(\mathbf{z})) \notin P(I)$.
		\end{itemize}
\end{compactenum}

The base case is when $\pi$ is a leaf. Notice that in this case, and by
  construction of $A^\prime_P$, the node $\pi$ cannot be labeled by states corresponding
  to rules of~$P$: indeed, there
are no transitions for these states leading to a tautology, and all
  transitions to such a state are from a state in the same stratum, so $\pi$
  could not be a leaf. Thus,
  we have three subcases:
\begin{itemize}
	\item $\pi$ may be labeled by
          $(q_{R^\prime(\mathbf{y})}^{\nu^\prime}, m)$, where $R^\prime$ is
		extensional. We must show~(\ref{case:positive-atom}), but this follows
          from Lemma~\ref{lem:extensional}.
	\item $\pi$ may be labeled by $(q_{R^\prime(\mathbf{y})}^{\nu^\prime},
          m)$, where $R^\prime$ is intensional
	and verifies $\strat(R^\prime) < i$.
	Again we
		need to show~(\ref{case:positive-atom}).
	By definition of the run $\rho$, this implies that there exists a
        run of~$A^\prime_P$ starting at $m$ in state
        $q_{R^\prime(\mathbf{y})}^{\nu'}$.
		But then~(\ref{case:positive-atom}) follows from the induction hypothesis on the strata
        (using the forward direction of the equivalence lemma).
	\item $\pi$ may be labeled by $\lnot
          (q_{R^\prime(\mathbf{y})}^{\nu^\prime}, m)$, where $R^\prime$ is intensional 
		and verifies $\strat(R^\prime) < i$. Observe that by construction of the automaton, $\nu^\prime$ is total (because negations are guarded in rule bodies). We
		need to show~(\ref{case:negative-atom}).
		By definition of the run $\rho$ there exists no run of
                $A^\prime_P$ starting at $m$ in state
                $q_{R^\prime(\mathbf{y})}^{\nu'}$. 
		Hence by induction on the strata
                we have (using the backward direction of the
                equivalence lemma) that
                $R^\prime(\dec_m(\nu^\prime(\mathbf{y}))) \notin P(I)$, which is
                what we needed to show.
\end{itemize}

For the induction case, where $\pi$ is an internal node, we let $(d,s)$ be $\lambda_E(m)$ in what follows, and we distinguish five subcases:
\begin{itemize}
\item $\pi$ may be labeled by $(q_{R^\prime(\mathbf{y})}^{\nu^\prime}, m)$ with $R^\prime$ intensional. 
	We need to prove~(\ref{case:positive-atom}). We distinguish two subsubcases:
  \begin{itemize}
	\item $\nu'$ is not total. In that case, given the definition of $\Delta(q_{R^\prime(\mathbf{y})}^{\nu^\prime},(d,s))$ and of the run,
	there exists a child $\pi^\prime$ of~$\pi$ labeled by
	$(q_{R^\prime(\mathbf{y})}^{\nu''},m')$, where $m^\prime \in \neigh(m)$ and $\nu''$ is either $\nu'$ or is $\nu' \cup \{x_j \mapsto a\}$ for some $x_j$ undefined by $\nu'$ and
	$a \in d$. Hence by induction on the run there exists $\mathbf{c'}$ such that $R'(\mathbf{c'}) \in P(I)$ and $\mathbf{c'}$ is compatible with $\nu''$ at node $m'$. 
	We then take $\mathbf{c}$
	to be $\mathbf{c'}$, and one can check that the compatibility condition holds.
	\item $\nu'$ is total. In that case, given the definition of $\Delta(q_{R^\prime(\mathbf{y})}^{\nu^\prime},(d,s))$ and of the run,
        there exists a child $\pi^\prime$ of~$\pi$ labeled by
        $(q_r^{\nu'',\calA},m')$, where $m^\prime \in \neigh(m)$, where $r$ is a
        rule with head $R^\prime(\mathbf{z})$, where
	$\nu'' = \mathrm{Hom}_{\mathbf{z},\mathbf{y}}(\nu')$ is a partial
        valuation which is not~$\mnull$, and where
	$\calA$ is the set of literals of $r$. 
        Then, by induction on the run, there exists
	a mapping $\mu : \vars(\calA) \to \mathrm{Dom}(I)$ that
		  verifies~(\ref{case:rule}).
	Thus by definition of the semantics of $P$ we have that
        $R'(\mu(\mathbf{z})) \in P(I)$, and we take $\mathbf{c}$
	to be $\mu(\mathbf{z})$. What is left to check is that the compatibility condition holds.
	We need to prove that $\dec_m(\nu'(\mathbf{y})) = \mathbf{c}$,
        i.e., that $\dec_m(\nu'(\mathbf{y})) = \mu(\mathbf{z})$.
	We know, by definition of $\mu$, that $\dec_{m'}(\nu''(\mathbf{z})) = \mu(\mathbf{z})$. 
	So our goal is to prove $\dec_m(\nu'(\mathbf{y})) =
        \dec_{m'}(\nu''(\mathbf{z}))$, i.e., by definition of $\nu''$
	we want $\dec_m(\nu'(\mathbf{y})) = \dec_{m'}(\mathrm{Hom}_{\mathbf{z},\mathbf{y}}(\nu')(\mathbf{z}))$.
        By definition of
        $\mathrm{Hom}_{\mathbf{z},\mathbf{y}}(\nu')$,
        we know that $\nu'(\mathbf{y}) =
        \mathrm{Hom}_{\mathbf{z},\mathbf{y}}(\nu')(\mathbf{z})$, and this
        implies the desired equality by applying
        Property~\ref{propTreeEnc1} to~$m$ and~$m'$.
  \end{itemize}

\item $\pi$ may be labeled by $(q_r^{\nu',\calA},m)$, where $\calA =
   \{\lnot R''(\mathbf{y})\}$, where $|\mathbf{y}| = 1$,
    and where, writing $y$ the one element of~$\mathbf{y}$, $y$ is undefined by $\nu'$.
		We need to prove~(\ref{case:rule}).
	By construction we have $\Delta(q_r^{\nu',\calA},(d,s)) = q_r^{\nu',\calA} \lor \bigvee_{a \in d} q_r^{\nu' \cup \{y\mapsto a\},\calA}$.
	So by definition of a run there is $m' \in \neigh(m)$ and a child $\pi'$ of $\pi$ such that $\pi'$ is labeled by $( q_r^{\nu',\calA},m')$ or by
	$( q_r^{\nu' \cup \{y\mapsto a\},\calA},m')$ for some $a\in d$. 
	In both cases it is easily seen that we can define an appropriate $\mu$ from the mapping~$\mu'$ that we obtain
        by induction on the run (more details are given in the next bullet point).
\item $\pi$ may be labeled by $(q_r^{\nu',\calA},m)$ with $\calA =
	\{R''(\mathbf{y})\}$.
	We need to
		prove~(\ref{case:rule}).
	By construction we have $\Delta(q_r^{\nu',\calA},(d,s)) = q_{R''(\mathbf{y})}^{\nu^\prime}$,
	so that by definition of the run there is $m' \in \neigh(m)$ and a child $\pi'$ of $\pi$ such that $\pi'$
	is labeled by $(q_{R''(\mathbf{y})}^{\nu^\prime},m')$.
	Thus by induction on the run there exists $\mathbf{c}$ such that
        $R''(\mathbf{c}) \in P(I)$ and $\mathbf{c}$ compatible with~$\nu'$ at
        node~$m'$. By Property~\ref{propTreeEnc1}, $\mathbf{c}$ is also
        compatible with~$\nu'$ at node~$m$.
	We define $\mu$ by $\mu(\mathbf{y}) \defeq \mathbf{c}$, which effectively defines it because in this case $\vars(\calA)=\mathbf{y}$, and this choice satisfies the required properties.
\item $\pi$ may be labeled by $(q_r^{\nu',\calA},m)$, with $\calA =
	\{\lnot R''(\mathbf{y})\}$ and $\nu'$ total on~$\mathbf{y}$.
	We again need to prove~(\ref{case:rule}).
	By construction we have $\Delta(q_r^{\nu',\calA},(d,s)) = \lnot q_{R''(\mathbf{y})}^{\nu^\prime}$ and then by definition of the automaton
	there exists a child $\pi'$ of $\pi$ labeled by $\lnot (q_{R''(\mathbf{y})}^{\nu^\prime},m)$ with $\strat(R'') < i$ and
	there exists no run starting at node $m$ in state $q_{R''(\mathbf{y})}^{\nu^\prime}$.
	So by using~(\ref{case:negative-atom}) of the induction on the strata we have $R''(\dec_m(\nu'(\mathbf{y})))
        \notin P(I)$.
	We define $\mu$ by $\mu(\mathbf{y}) = \dec_m(\nu'(\mathbf{y}))$, which effectively defines it because $\vars(\calA)=\mathbf{y}$, and the compatibility conditions are satisfied.

\item $\pi$ may be labeled by $(q_r^{\nu',\calA},m)$, with $\card{\calA}
	\geq 2$. We need to prove~(\ref{case:rule}).
	Given the definition of $\Delta(q_r^{\nu',\calA},(d,s))$ and by
        definition of the run,
	one of the following holds:
	\begin{itemize}
		\item There exists $m' \in \neigh(m)$ and a child $\pi'$ of $\pi$ such that $\pi'$ is labeled by $(q_r^{\nu',\calA},m')$.
			By induction there exists $\mu' : \vars(\calA)
			\to \mathrm{Dom}(I)$ satisfying~(\ref{case:rule})
                        for node $m'$. We can take $\mu$ to be $\mu'$, which satisfies the required properties.
		\item 
	There exist $(m_1, m_2) \in \neigh(m) \times \neigh(m)$ and $\pi_1, \pi_2$ children of $\pi$ and non-empty sets $\calA_1,\calA_2$ 
	that partition $\calA$ and a total valuation $\nu''$ of
            $\vars(\calA_1) \cap \vars(\calA_2)$ with values in $d$
	such that $\pi_1$ is labeled by $(q_r^{\nu' \cup \nu'',\calA_1},m_1)$ 
	and $\pi_2$ is labeled by $(q_r^{\nu' \cup \nu'',\calA_2},m_2)$.
	By induction there exists $\mu_1 : \vars(\calA_1) \to
			\mathrm{Dom}(I)$ and similarly $\mu_2$ that satisfy~(\ref{case:rule}).
	Thanks to the compatibility conditions for $\mu_1$ and $\mu_2$ and to Property~\ref{propTreeEnc1} applied to~$m_1$ and~$m_2$ via~$m$, we can define $\mu : \vars(\calA) \to \mathrm{Dom}(I)$
	with $\mu \defeq \mu_1 \cup \mu_2$. One can check that $\mu$ satisfies the required properties.
	\end{itemize}
\end{itemize}

Hence, the forward direction of our equivalence lemma is proven.

\subparagraph*{\qquad Backward direction.}
We now prove the backward direction of the induction on strata of
our main equivalence lemma (Lemma~\ref{lem:correctness}).
From the induction hypothesis on strata,
we know that, for every relation $R$ with $\strat(R) \leq i-1$, for every node $n \in T$ and partial valuation $\nu$ of $\mathbf{x}$, 
there exists a run $\rho$ of $A^\prime_P$ starting at node $n$
in state $q_{R(\mathbf{x})}^{\nu}$ if and only if
there exists $\mathbf{c}$ such that $R(\mathbf{c}) \in P(I)$ and $\mathbf{c}$ is compatible
with $\nu$ at node $n$. 
Let~$R$ be a relation with $\strat(R)=i$, let $n \in T$ be a node and let $\nu$ be a partial valuation of $\mathbf{x}$ such that
there exists $\mathbf{c}$ such that $R(\mathbf{c}) \in P(I)$ and $\mathbf{c}$ is compatible
with $\nu$ at node $n$. 
We need to show that there exists a run $\rho$ of $A^\prime_P$ starting at node $n$
in state $q_{R(\mathbf{x})}^{\nu}$.
We will prove this by induction on the smallest $j \in \mathbb{N}$ such
that $R(\mathbf{c}) \in \Xi^j_{P}(P_{i-1}(I))$, where
$\Xi^j_{P}$ is the $j$-th
application of the immediate consequence operator for the program~$P$ and $P_{i-1}$ is the restriction of
$P$ with only the rules up to strata~$i-1$.
The base case, when $j=0$, is in fact vacuous since 
$R(\mathbf{c}) \in \Xi^0_{P}(P_{i-1}(I))=P_{i-1}(I)$ implies
that $\strat(R)\leq i-1$, whereas we assumed $\strat(R)=i$.
For the inductive case ($j \geq 1$), we have $R(\mathbf{c})
\in \Xi_P^j(P_{i-1}(I))$ so by definition of the semantics of $P$, there is a rule
$r$ of the form $R(\mathbf{z}) \leftarrow L_1(\mathbf{y}_1) \ldots
L_t(\mathbf{y}_t)$ of $P$ and a mapping $\mu :
\mathbf{y}_1\cup\dots\cup\mathbf{y}_t\to \mathrm{Dom}(I)$ such that
$\mu(\mathbf{z}) = \mathbf{c}$ and, for every literal $L_l$ in the body of $r$:
\begin{itemize}
  \item If $L_l(\mathbf{y}_l) = R_l(\mathbf{y}_l)$ is a positive
    literal, then $R_l(\mu(\mathbf{y}_l)) \in
          \Xi_P^{j-1}(P_{i-1}(I))$
        \item If $L_l(\mathbf{y}_l) = \lnot R_l(\mathbf{y}_l)$ is a
          negative literal, then $R_l(\mu(\mathbf{y}_l)) \notin
          P_{i-1}(I)$
\end{itemize}

Now, the definition of clique-guardedness ensures that
each pair of elements of~$\mathbf{c}$ co-occurs in some fact of~$I$, i.e., 
$\mathbf{c}$ induces a clique in~$I$. This ensures that there is a bag of the
tree decomposition that contains all elements of~$\mathbf{c}$ (see
Lemma~2 of~\cite{BodlaenderK10}, Lemma~1
of~\cite{gavril1974intersection}), i.e., 
there exists a node $n'$ such that $\mathbf{c} \subseteq \dom(n')$.
Let $n = n_1, n_2,\ldots,n_p = n'$ be the nodes on the path from $n$ to $n'$, and $(d_i,s_i)$ be $\lambda_E(n_i)$ for $1 \leq i \leq p$.
By compatibility, for every $x_j$ defined by $\nu$ we have $\dec_n(\nu(x_j)) = c_j$. But $\dec_n(\nu(x_j)) \in \dom(n)$ and $c_j \in \dom(m)$ so by Property~\ref{propTreeEnc2}, for every $1 \leq i \leq p$ we have
$c_j \in \dom(m_i)$ and $\enc_{m_i}(c_j) =\enc_{n}(c_j) = \enc_{n}(\dec_n(\nu(x_j))) = \nu(x_j)$, so that $\nu(x_j) \in d_i$.
We can then start to construct the run $\rho$ starting at node $n$ in state
  $q_{R(\mathbf{x})}^{\nu}$ as follows. The root $\pi_1$ is labeled by $(q_{R(\mathbf{x})}^{\nu},n)$, and
for every $2 \leq i \leq p$, $\pi_i$ is the unique child of $\pi_{i-1}$ and is labeled by $(q_{R(\mathbf{x})}^{\nu}, m_i)$. 
This part is valid because we just proved that for every $i$, there is no $j$ such that $y_j$ is defined by $\nu$ and $\nu(y_j) \notin d_j$.
Now from $\pi_{n'}$, we continue the run by staying at node $n'$ and building up the
valuation, until we reach a total valuation $\nu'$ such that $\nu'(\mathbf{x}) = \enc_{n'}(\mathbf{c})$.
Hence we now only need to build a run $\rho'$ starting at node $n'$ in state
$q_{R(\mathbf{x})}^{\nu'}$.

To achieve our goal of building a run starting at node $n'$ in state  $q_{R(\mathbf{x})}^{\nu'}$, it suffices to construct a run starting at node $n'$ in state $q_r^{\nu'',\{L_1,\ldots,L_t\}}$, 
with $\nu'' = \mathrm{Hom}_{\mathbf{z},\mathbf{x}}(\nu')$.
The first step is to take care of the literals of the rule and to prove that:

	\begin{compactenum}[(i)]
\item If $L_l(\mathbf{y}_l) = R_l(\mathbf{y}_l)$ is a positive literal, then there exists a node $m_l$ and a total valuation $\nu_l$
		of $\mathbf{y}_l$ with $\dec_{m_l}(\nu_l(\mathbf{y}_l)) = \mu(\mathbf{y}_l)$ such that there exists a run $\rho_l$
		starting at node $m_l$ in state
                $q_{R_l(\mathbf{y}_l)}^{\nu_l}$.
              \item If $L_l(\mathbf{y}_l) = \lnot R_l(\mathbf{y}_l)$ is a negative literal, then
		there exists a node $m_l$ and a total valuation $\nu_l$
		of $\mathbf{y}_l$ with $\dec_{m_l}(\nu_l(\mathbf{y}_l)) = \mu(\mathbf{y}_l)$ such that there exists a run $\rho_l$
		starting at node $m_l$ in state
                $\lnot q_{R_l(\mathbf{y}_l)}^{\nu_l}$.
\end{compactenum}

We first prove (i).
We have $R_l(\mu(\mathbf{y}_l)) \in \Xi_P^{j-1}(P_{i-1}(I))$, and because $P$ is clique-frontier-guarded, there exists a node $m$ such that $\mu(\mathbf{y}_l) \subseteq m$.
We take $m_l$ to be $m$ and $\nu_l$ to be such that $\nu_l(\mathbf{y}_l) = \enc_{m_l}(\mu(\mathbf{y}_l))$. We then directly obtain (i) by induction hypothesis (on $j$).
We then prove (ii). Because the negative literals are guarded in rule bodies, there exists a node $m$ such that $\mu(\mathbf{y}_l) \subseteq m$.
We take $m_l$ to be $m$ and $\nu_l$ to be such that $\nu_l(\mathbf{y}_l) = \enc_{m_l}(\mu(\mathbf{y}_l))$.
We straightforwardly get (ii) using the induction on the strata of our equivalence
lemma.

The second step is to use the runs $\rho_l$ that we just constructed and to construct from them a run starting at node $n'$ in state $q_r^{\nu'',\{L_1,\ldots,L_t\}}$.
We describe in a high-level manner how we build the run.
Starting at node~$n$, we partition the literals to verify (i.e., the atoms of the
body of the rule that we are applying), in the following way:

\begin{itemize}
  \item We create one class in the partition for each literal $R_l$
    (which can be intentional or extensional)
    such that $m_l$ is $n'$, which we verify directly at the current node.
    Specifically, we handle these literals one by one, by splitting the
    remaining literals in two using the transition formula corresponding to the rule and by staying at node $n'$ and building the valuations according to 
$\dec_n(\mu)$.
    \item For the remaining literals, considering all neighbors of~$n'$ in the
      tree encoding, we split the literals into one class per neighbor $n''$,
      where each literal $L_l$ is mapped to the neighbor that allows us to reach
      its node~$m_l$. We ignore the empty classes. If there is only one class,
      i.e., we must go in the same direction to verify all facts, we simply go to
      the right neighbor~$n''$, remaining in the same state. If there are
      multiple classes, we partition the facts and verify each class on the
      correct neighbor.
      
      One must then argue that, when we do so, we can indeed
      choose the image by~$\nu''$ of all elements that were shared between
      literals in two different classes and were not yet defined in~$\nu''$. The
      reason why this is possible is because we are working on a tree encoding:
      if two facts of the body share a variable $x$, and the two facts will be
      witnessed in two different directions, then the variable $x$ must be mapped
      to the same element in the two direction (namely, $\mu(x)$), which implies that it must occur
      in the node where we split. Hence, we can indeed choose the image
      of~$x$ at the moment when we split. \qedhere
\end{itemize}
\end{proof}

\subparagraph*{FPT-linear time construction.}
Finally, we justify that we can construct in FPT-linear time the automaton $A_P$ which recognizes the same language as $A_P^\prime$.
The size of $\Gamma_{\sigmae}^{\ki}$ only depends on $\ki$ and on the
extensional signature, which are fixed. As the number of states is linear
in $\card{P}$, the number of transitions is linear in~$\card{P}$. 
Most of the transitions are of constant size, and in fact one can check that the only problematic transitions
are those for states of the form $q^\nu_{R(\textbf{x})}$ 
with $R$ intensional, specifically the second bullet point. Indeed, we have
defined a transition from $q^\nu_{R(\textbf{x})}$, for each valuation $\nu$ of a
rule body, to the $q_r^{\nu^\prime,\calA}$ for linearly many rules, so in
general there are quadratically many transitions.

However, it is easy to fix this problem: instead of having one state $q^\nu_{R(\textbf{x})}$
for every occurrence of an intensional predicate $R(\textbf{x})$ in a rule body
of~$P$ and total valuation $\nu$ of this rule body,
we can instead have a constant number of states $q_{R(\textbf{a})}$ for $\textbf{a} \in \calD_{\ki}^{\arity{R}}$. 
In other words, when we have decided to verify a single intensional atom in the
body of a rule, instead of remembering the entire valuation of the rule body (as
we remember $\nu$ in $q^\nu_{R(\textbf{x})}$), we can simply forget all other
variable values, and just remember the tuple which is the image of~$\textbf{x}$
by~$\nu$, as in~$q_{R(\textbf{a})}$. Remember that the number of such states is
only a function of $\kp$ and $\ki$, because bounding $\kp$ implies that we bound
the arity of~$P$, and thus the arity of intensional predicates.

We now redefine the transitions for those states :

	\begin{itemize}
		\item If there is a $j$ such that $a_j \notin d$,
                  then $\Delta(q_{R(\mathbf{a})},(d,s)) = \false$.
		\item Else, $\Delta(q_{R(\mathbf{a})},(d,s))$ is a
                  disjunction of all the $q_r^{\nu^\prime,\calA}$ for
                  each rule~$r$ such that
		  the head of~$r$ is $R(\mathbf{y})$, $\nu^\prime(\mathbf{y}) = \mathbf{a}$ and
		$\calA$ is the set of all literals in the body of $r$.
	\end{itemize}

        The key point is that a given $q_r^{\nu',\calA}$ will only appear
        in rules for states of the form $q_{R(\mathbf{a})}$ where $R$ is
        the predicate of the head of $r$, and there is a constant number
        of such states.
        
We also redefine the transitions that used these states:
\begin{itemize}
\item Else, if $\calA = \{R^\prime(\mathbf{y})\}$ with $R^\prime$ intensional, 
		then $\Delta(q_r^{\nu,\calA},(d,s)) = q_{R^\prime(\nu(\mathbf{y}))}$.
\item Else, if $\calA = \{\lnot R^\prime(\mathbf{y})\}$ with $R^\prime$ intensional, 
		then $\Delta(q_r^{\nu,\calA},(d,s)) = \lnot q_{R^\prime(\nu(\mathbf{y}))}$. 
\end{itemize}

$A_P$ recognizes the same language as $A_P'$. Indeed, consider a run of $A'_P$, and replace every state $q^\nu_{R(\textbf{x})}$ with $R$ intensional by the state 
$q_{R(\nu(\textbf{x}))}$: we obtain a run of $A_P$. 
Conversely, being given a run of $A_P$, observe that every state $q_{R(\textbf{a})}$ comes from a state~$q_r^{\nu,\{R(\mathbf{y})\}}$ with 
$\nu(\mathbf{y}) = \mathbf{a}$.
We can then replace~$q_{R(\textbf{a})}$ by the state $q^\nu_{R(\textbf{x})}$ to obtain a run of~$A'_P$.

\subsection{Managing Unguarded Negations}
\label{sec:managing-unguarded}

We now explain how the translation can be extended to the full CFG-Datalog fragment.
We recall that the difference with \cfggnd is that negative
literals in rule bodies no longer need to be clique-guarded.
Remember that clique-frontier-guardedness was used 
in the translation of \cfggnd
to ensure the following property:
when the automaton is verifying a rule $r \defeq R(\mathbf{z}) \leftarrow L_1(\mathbf{y}_1) \ldots
L_t(\mathbf{y}_t)$ at some node $n$, i.e., when it is in a state $q_r^{\nu,\calA}$ at node $n$ 
for some subset $\calA$ of literals of the body of~$r$ and partial
valuation~$\nu$ of the variables in~$\vars(\calA)$, 
then, for each literal $L_l(\mathbf{y}_l)$ for $1 \leq l \leq t$,
the images of $\mathbf{y}_l$ all appear together in a bag.
More formally,
writing $\vars(r)$ for the variables of the body of~$r$,
let $\mu : \vars(r) \to \dom(I)$ be a mapping
with $\mu(\mathbf{z}) = \dec_n(\nu(\mathbf{z}))$ that witnesses that
$R(\mu(\mathbf{z})) \in P(i)$:
that is, if $L_l(\mathbf{y}_l)$ is a positive literal $S_l(\mathbf{y}_l)$
then we have $S_l(\mu(\mathbf{y_l})) \in P(i)$ 
and if 
$L(\mathbf{y}_l)$ is a negative literal $\lnot S_l(\mathbf{y}_l)$
then we have
$S_l(\mu(\mathbf{y}_l)) \notin P(I)$.
In this case, we know that each $\mu(\mathbf{y}_l)$ must be contained in a bag of the tree decomposition.

This property is still true in CFG-Datalog when $L(\mathbf{y}_l)$ is a positive
literal $S_l(\mathbf{y}_l)$. 
Indeed, if $S$ is an extensional relation then the fact $S(\mu(\mathbf{y}_l))$ is encoded somewhere in the tree encoding, hence $\mu(\mathbf{y}_l)$ is contained in a bag of the tree decomposition. 
If $S$ is an intensional predicate then, because $P$ is clique-frontier-guarded, $\mu(\mathbf{y}_l)$ is also contained in a bag. 
However,
when $L(\mathbf{y}_l)$ is a negative literal $\lnot S_l(\mathbf{y}_l)$,
it is now possible that $\mu(\mathbf{y}_l)$ is not contained in any bag of the tree decomposition. 
This can be equivalently rephrased as follows:
there are $y_i, y_j \in \mathbf{y}_l$ with $y_i \neq y_j$
such that the occurrence subtrees of $\mu(y_i)$ and that of $\mu(y_j)$ are disjoint.
If this happens, the automaton that we construct in the previous proof no longer
works: it cannot assign the correct values to $y_i$ and $y_j$, because once a
value is assigned to~$y_i$, the automaton cannot leave the occurrence subtree of
$\mu(y_i)$ until a value is also assigned to~$y_j$, which is not possible if the
occurrence subtrees are disjoint.

To circumvent this problem,
we will first rewrite the CFG Datalog program $P$ into another program (still of
bounded body size)
which intuitively distinguishes between two kinds of negations: the negative
atoms that will hold as in the case of clique-guarded negations in
\cfggnd,
and the ones that will hold because two variables have disjoint occurrence
subtrees.
First, we create a vacuous unary fact $\Adom$, we modify the input
instance and tree encoding in linear time to add the fact $\Adom$ for every
element~$a$ in the active domain, and we modify $P$ in linear time: for each rule $r$, for each
variable $x$ in the body of~$r$, we add the fact $\Adom(x)$. This ensures
that each variable of rule bodies occurs in at least one positive fact.

Second, we rewrite $P$ to a different program $P'$. 
Let $r$ be a rule of $P$, and let $\calN$ be the set of negative atoms in the body of~$r$.
Let $\calN_{\text{G}} \cup \calN_{\neg \text{G}}$ be a partition of~$\calN$ (where the classes in the partition may be empty),
intuitively distinguishing the guarded and unguarded negations.
For every atom of $\calN_{\neg \text{G}}$, we nondeterministically choose a pair $(y_i, y_j)$ of
distinct variables of this atom, and consider the undirected graph $\calG$ formed by the edges $\{y_i, y_j\}$
(we may choose the same edge for two different atoms). 
The graph $\calG$
intuitively describes the variables that must be mapped to elements having
disjoint occurrence subtrees in the tree encoding: if there is an edge between
two variables in~$\calG$, then they must be mapped to
two elements whose subtrees of occurrences do not intersect.
For each rule $r$ of~$P$, for each choice of $\calN_{\text{G}} \cup
\calN_{\neg \text{G}}$ and~$\calG$, we create a rule $r_{\calN_{\text{G}},
\calN_{\neg \text{G}}, \calG}$ defined as follows: it has the same head as $r$,
and its body contains the positive atoms of the body of~$r$ (including the
$\Adom$-facts) and the negative atoms of $\calN_{\text{G}}$. We call $\calG$ the
\emph{unguardedness graph of $r_{\calN_{\text{G}}, \calN_{\neg \text{G}}, \calG}$}.
Note that the semantics of~$P'$ will defined relative to the instance and also relative to the tree
encoding of the instance that we consider: specifically, a rule can fire if there is a
valuation that satisfies it in the sense of \cfggnd (i.e., for all
atoms, including negative atoms, all variables must be mapped to elements that
occur together in some node), and which further respects the unguardedness
graph, i.e., for any two variables $x \neq y$ with an edge in the graph, the
elements to which they are mapped must have disjoint occurrence subtrees in the
tree encoding. Note that we can compute $P'$ from~$P$ in FPT-linear time
parameterized by the body size, because the number of rules created in~$P'$ for
each rule of~$P$ can be bounded by the body size; further, the bound on the body
size of~$P'$ only depends on that of~$P$, specifically it only increases by the
addition of the atoms $\Adom(x)$.

The translation of~$P'$ can now be done as in the case of
\cfggnd that we presented before;
the only thing to explain 
is how the automaton can ensure that the semantics of the unguardedness graph is satisfied.
To this end, we will first make two general changes to the way that our
automaton is defined, and then present the specific tweaks to handle the
unguardedness graph. The two general changes can already be applied to the
original automaton construction that we presented, without changing its
semantics.

The first change is that, instead of isotropic automata, we will use automata
that take the directions of the tree into account, as in \cite{cachat2002two}
for example (with stratified negation as we do for SATWAs). Specifically, we
change the definition of the transition function. Remember that a SATWA has a
transition function $\Delta: \calQ \times \Gamma \to \calB(\calQ)$
that maps each pair of a state and a label to a propositional formula on 
states of~$\calQ$. To handle directions, $\Delta$ will instead map to a
propositional formula on pairs of states of~$\calQ$ and 
of directions in the tree, in $\{\bullet, \uparrow, \leftarrow, \rightarrow\}$.
The intuition is that the corresponding state is only evaluated on the tree node 
in the specified direction (rather than on any arbitrary neighbor). We will use
these directions to ensure that, while the automaton considers a rule
application and navigates to find the atoms used in the rule body, then it never
visits the same node twice. Specifically, consider two variables $y_i$ and~$y_j$
that are connected by an edge in the unguardedness graph,
and imagine that we first assign a value $a$ to $y_i$ in some
node $n$. To assign a value to $y_j$, we must leave the occurrence subtree
of the current $a$ in the tree encoding, and must choose a value 
outside of this occurrence subtree.
Thus, the automaton must ``remember'' when
it has left the subtree of occurrences of~$a$, so that it can choose a value
for~$y_j$.
However,
an isotropic automaton cannot
``remember'' that it has left the subtree of occurrences of~$a$, because it can
always come back on a previously visited node, by going back in the direction
from which it came.
However, using SATWAs with directions, and changing the automata states to
store the last direction that was followed, we can ensure that the automaton
cannot come back to a previously visited node (while locating the facts that
correspond to the body of a 
rule application). This ensures that, once the
automaton has left the subtree of occurrences of an element, then it cannot come
back in this subtree again while it is considering the same rule application.
Hence, the first general change is that we use SATWAs with directions, and we
use the directions to ensure that the automaton does not go back to a previously
visited node while considering the same rule application.
In fact, this first general change does not modify the semantics of the
automaton: indeed, in the case of isotropic automata, we did not really need the
ability to go back to a previously visited node when verifying a rule
application.

The second general change that we perform on the automaton is that, when
guessing a value for an undefined variable, then we only allow the guess to
happen as early as possible. In other words, suppose the automaton is at a node
$n$ in the tree encoding while it was previously at node $n'$. Then it can
assign a value $a \in n$ to some variable $y$ only if $a$ was not in $n'$, i.e.,
$a$ has just been introduced in~$n$. Obviously an
automaton can remember which elements have been introduced in this sense, and
which elements have not. This change can be performed on our existing
construction without changing the semantics of the automaton, by only
considering runs where the automaton assigns values to variables at the moment
when it enters the occurrence subtree of this element.

Having done these general changes, we will simply reuse the previous automaton
construction (not taking the unguardedness graph $\calG$ into
account) on the program $P'$, and make two tweaks to ensure
that the unguardedness graph is
respected. The first tweak is that,
in states of the form $q_r^{\nu,\calA}$, the automaton will also
remember, for each undefined variable~$x$ (i.e., $x$ is in the domain of~$\nu$
but $\nu(x)$ is still undefined),
a set $\beta(x)$ of \emph{blocking elements} for~$x$, which are
elements of the tree encoding. While $\beta(x)$ is non-empty, then the automaton
is not allowed to guess a value for~$x$, intuitively because we know that it is
still in the occurrence subtree of some previously mapped variable $y \in
\beta(x)$ which is adjacent to~$x$
in~$\calG$. Here is how these sets of blocking elements are computed and updated:

\begin{itemize}
  \item When the automaton starts to consider the application of a rule $r$, then
    $\beta(x) \colonequals \emptyset$ for each variable $x$ of the body of~$r$.
  \item When the automaton guesses a value $a$ for a variable~$x$, then for
    every undefined variable $y$,
    if $x$ and $y$ are adjacent in~$\calG$, then we set $\beta(y)
    \colonequals \beta(y) \cup \{a\}$, intuitively adding $a$ to the set of
    blocking elements for~$y$. This intuitively means that the automaton is not
    allowed to guess a value for~$y$ until it has exited the subtree of occurrences
    of~$a$. Note that, if the automaton wishes to guess values for multiple
    variables while visiting one node (in particular when partitioning the
    literals of~$\calA$), then the blocking sets are updated between
    each guess: this implies in particular that, if there is an edge in~$\calG$
    between two variables $x$ and $y$, then the automaton can never guess the value
    for $x$ and for~$y$ at the same node.
  \item When the automaton navigates to a new node $n'$ of the tree encoding, then for
    every variable $x$ in the domain of~$\nu$ which does not have an image yet,
    we set $\beta(x) \colonequals \beta(x) \cap n'$. Intuitively, when an element
    $a$ was blocking for $x$ but disappears from the current node, then $a$ is no
    longer blocking. Note that $\beta(x)$ may then become empty, meaning that
    the automaton is now free to guess a value for~$x$.
\end{itemize}

The blocking sets ensure that, when the automaton guesses a value 
for~$x$, then this value is guaranteed not to occur in the occurrence subtree of
variables that are adjacent to~$x$ in~$\calG$ and have been guessed
before. This also relies on the second general change above: we can only guess
values for variables as early as possible, i.e., we can only use elements in
guesses when we have 
just entered their occurrence subtree, so when $\beta(x)$ becomes empty then the
possible guesses for~$x$ do not include any element whose occurrence subtree
intersects that of~$\nu(y)$ for any variable $y$ adjacent to $x$ in $\calG$.

The second tweak is that, when we partition the set of literals to be verified,
then we use the directionality of the automaton to ensure that the remaining
literals are split across the various directions (having at most one run for every direction). 
For instance, 
considering the rule body $\{\Adom(x), \Adom(y)\}$ and the unguardedness graph $\calG$
having an edge between~$x$ and~$y$, the automaton may decide at one node to
partition $\calA = \{\Adom(x), \Adom(y)\}$ into
$\{\Adom(x)\}$ and $\{\Adom(y)\}$, and these two subsets of facts
will be verified by two independent runs: these two runs are required to go in different directions of
the tree. This will ensure that, even though the edge $\{x, y\}$ of~$\calG$ will
not be considered explicitly by either of these runs (because the domain of
their valuations will be $\{x\}$ and $\{y\}$ respectively), it will still be the
case that $x$ and $y$ will be mapped to elements whose occurrence subtrees do
not intersect: this is again using the fact that we map elements as early as
possible.

We now summarize how the modified construction works:

\begin{compactenum}[(i)]
	\item Assume that the automaton $A$ is 
          at some node $n$ in state $q_{R(\mathbf{x})}^{\nu''}$, with $\nu''$
          being total in $\mathbf{x}$.
		\label{itm:item1}
	\item At node $n$, the automaton chooses a rule $r': R(\mathbf{z}) \leftarrow L_1(\mathbf{y}_1) \ldots L_t(\mathbf{y}_t)$ of $P'$ and goes to state $q_{r'}^{\nu,\calA}$
		where $\nu \defeq \mathrm{Hom}_{\mathbf{z},\mathbf{x}}(\nu'')$ and $\calA$ is the set of literals in the body of $r'$.
                That is, it simply chooses a rule to derive $R(\nu''(\mathbf{x}))$.
                This amounts to choosing a rule of the original program $P$, and
                choosing which negative atoms will be guarded (i.e., mapped to
                variables that occur together in some node of the tree
                encoding), and choosing the unguardedness graph $\calG$ in a way to
                ensure that each unguarded negated atom has a pair of variables that forms an edge of $G$. 
		The blocking set
                $\beta(x)$ of each variable~$x$ in the rule body is initialized
                to the empty set, and whenever the automaton will move to a
                different node $n'$ then each element $a$ that is no longer
                present in~$n'$ will be removed from $\beta(x)$ for each variable
                $x$, formally, $\beta(x) \colonequals \beta(x) \cap n'$.
		\label{itm:item2}
	\item From now on, assume the automaton $A$ always remembers (stores in its state)
          which elements $N$ have just been introduced in the current node of
          the tree encoding.
          That is, $N$ is initialized with the elements in $n$, and when $A$ 
		goes from some node $n'$ to node~$n''$, $N$ becomes $n'' \setminus n'$. 
		When guessing values for variables, the automaton will only use
                values in~$N$, so as to respect
                the condition that we guess the value of variables as early as
                possible. This is how we implement our second general change.
		\label{itm:item3}
	\item While staying at node $n$, the automaton chooses 
          some undefined variables $x$ (i.e., variables in the domain of~$\nu$ that do not have a value
          yet), and guesses some values
          in~$N$ for them, one after another. For each such variable $x$, we first verify that
          $\beta(x) = \emptyset$ (otherwise we fail), we set $\nu(x) \colonequals
          a$ where $a$ is the guessed value, and then,
          for every edge $\{x, y\}$ in~$\calG$ such that $y$ is an undefined
          variable (i.e., it is in the
          domain of~$\nu$ but does not have an image by~$\nu$ yet), we set
          $\beta(y) \colonequals \beta(y) \cup \{a\}$, ensuring that no value
          will be guessed for~$y$ until the automaton has left the subtree of
          occurrences for~$a$.
		We call $\nu'$ the resulting new valuation.
		\label{itm:item6}
	\item While staying at node $n$, the automaton
          guesses a partition of $\calA$ as $\Pdir = (\calA^{\bullet}, \calA^{\uparrow},
          \calA^{\leftarrow},\calA^{\rightarrow})$ to decide in which direction
          each one of the remaining facts is sent. 
          Of course, if there is a direction for which $n$ has no neighbor (e.g.,
          $\leftarrow$ and $\rightarrow$ if $n$ is a leaf, or $\uparrow$ if $n$
          is the root), then $\calA^d$ in the corresponding direction~$d$ must
          be empty. \label{itm:item7}
        \item If $\calA^\bullet$ is the
          only class that is not empty, meaning that all remaining facts will be
          witnessed at the current node, then go to step~\ref{itm:item12}.
	\item While staying at node $n$, the automaton checks that each variable
          $x$
          that appears in two different classes of $\Pdir$ has been assigned a
          value, i.e., $\nu(x)$ is defined; otherwise, the automaton fails.
		This is to ensure that the partitioning is consistent (i.e., that a variable $x$ will not be assigned different values in different runs). 
		The automaton
          also checks that $\nu(x)$ is defined for each variable occurring in
          $\calA^{\bullet}$, that is, we assume without loss of generality that
          atoms that will be verified at the current node have all their variables
          already mapped. This ensures that the undefined variables 
          are partitioned between directions in $\{\uparrow, \leftarrow,
          \rightarrow\}$.
	\item The automaton then launches a run $q_{r'}^{\nu', \calA^d}$ for each direction $d \in \{\bullet, \uparrow, \leftarrow, \rightarrow\}$ at the corresponding node
		($n$ for $\bullet$, the parent of $n$ for $\uparrow$, the left
                of right child of $n$ for $\leftarrow$ or $\rightarrow$).
	\item For each of these runs, we update the value of $N$ and of the blocking sets, and we go back to step~\ref{itm:item6}, except
          that now $A$ remembers the direction from which it comes, and does not go back to the
          previously visited node. For example if the automaton goes from some node $n$ to the parent $n'$ of $n$ such
		that $n$ is the left child of $n'$, then in the partition that
                will be guessed at step~\ref{itm:item7} we will have
                $\calA^{\leftarrow} = \emptyset$. Further, in each of these
                runs, of course, the automaton remembers the values of the
                blocking sets $\beta(x)$ for each undefined variable $x$.
		\label{itm:item11}
	\item Check that all the variables have been assigned. Launch positive states for each positive intensional literal and negative states for 
		each negative literal, i.e., start from step~\ref{itm:item1}: in
                this case, when the automaton verifies a different rule
                application, then of course it forgets the values of the
                blocking sets, and forgets the previous direction (i.e., it can
                again visit the entire tree from the node where it starts).
                For each positive extensional literal, simply
		check that the atom is indeed encoded in the current node of the tree encoding.
		\label{itm:item12}
	
\end{compactenum}

All these modifications can be implemented in FPT-linear time provided that the
arity of~$P$ is bounded, which is the case because the body size of $P$ is bounded.
Moreover, as we pointed out after the proof of Theorem~\ref{thm:satwaprov}, the construction of provenance cycluits can easily be modified to work
for stratified alternating two way automata with directions,
so that all our results about CFG-Datalog (evaluation and provenance cycluit
computation in FPT linear time) still hold on this modified automaton.

This finishes the proof of translation.

\section{Conclusion}
We introduced CFG-Datalog, a stratified Datalog fragment whose
evaluation has FPT-bilinear complexity when parameterized by instance
treewidth and program body size. The complexity result is obtained via 
translation to alternating two-way automata, and via the computation of a
provenance representation
in the form of stratified cycluits, a generalisation of
provenance circuits that we hope to be of independent interest.

A careful inspection of the proofs shows that our results can be used to derive PTIME combined complexity results
on arbitrary instances, e.g., XP membership when parametrizing only by
program size; this recaptures in particular the tractability of some
query languages on arbitrary instances, such as $\alpha$-acyclic queries
or SAC2RPQs.
We also intend to extend our
cycluit framework to support more expressive 
provenance semirings than Boolean provenance (e.g., formal power series~\cite{green2007provenance}).

We leave open the question of practical implementation of the methods we
developed, but we have good hopes that this approach can give efficient
results in practice, in part from our experience with a preliminary
provenance prototype~\cite{monet2016probabilistic}. Optimization is possible, for
instance by not representing the full automata but building them on the
fly when needed in query evaluation. Another promising direction
supported by our experience, to deal with real-world datasets that are
not treelike, is to use partial tree
decompositions~\cite{maniu2017indexing}.

\subparagraph*{Acknowledgements.}
This work was partly funded by the Télécom ParisTech Research Chair on Big Data
and Market Insights.

\bibliographystyle{abbrv}
\bibliography{main}

\end{document}